%% file: main-full-arxiv.tex
\documentclass[acmsmall,authorversion,dvipsnames]{acmart}
\usepackage[ruled,noend,linesnumbered]{algorithm2e}
\usepackage{mathtools}
\usepackage{enumitem}
\usepackage{subcaption}
\usepackage{soul}
\usepackage{csquotes}
\usepackage{tikz}
\usepackage{bm}
\usetikzlibrary{calc}
\usetikzlibrary{topaths}
\usetikzlibrary{backgrounds}
\usetikzlibrary{graphs}
\usepackage{dashbox}
\usepackage{float}
\usepackage{multirow}
\usepackage{afterpage}
\usepackage{makecell}
\usepackage{xassoccnt}
\usepackage{multicol}
\usepackage{float}

\SetAlFnt{\small}

\tikzset{vertex/.style={minimum width=4pt,inner sep=0pt,circle,fill=black}}
\tikzset{hyperedge/.style={thick}}
\tikzset{hypergraph/.style={every label/.style={
circle,inner sep=1pt}}}
\tikzset{subtle/.style={on background layer,every path/.append style={lightgray!50!white}}}

\usepackage[capitalise]{cleveref}

\input{macros}

\newcommand{\arxivorappendix}{appendix}

\setcopyright{cc}
\setcctype{by}
\acmJournal{PACMMOD}
\acmYear{2026} \acmVolume{4}
\acmNumber{3 (SIGMOD)} \acmArticle{240}
\acmMonth{6} %
\acmDOI{10.1145/3802117}

\received{October 2025}
\received[revised]{January 2026}
\received[accepted]{February 2026}

\begin{document}
\title{Succinct Structure Representations for Efficient Query Optimization}
\author{Zhekai Jiang}
\authornote{Both authors contributed equally to this research.}
\affiliation{
  \institution{EPFL}
  \city{Lausanne}
  \state{Vaud}
  \country{Switzerland}
}
\email{zhekai.jiang@epfl.ch}
\orcid{0009-0001-0989-5926}

\author{Qichen Wang}
\authornotemark[1]
\affiliation{
  \institution{Nanyang Technological University}
  \country{Singapore}
}
\email{qichen.wang@ntu.edu.sg}
\orcid{0000-0002-0959-5536}

\author{Christoph Koch}
\affiliation{
  \institution{EPFL}
  \city{Lausanne}
  \state{Vaud}
  \country{Switzerland}
}
\email{christoph.koch@epfl.ch}
\orcid{0000-0002-9130-7205}

\input{0_abstract}

\ccsdesc[500]{Theory of computation~Database query processing and optimization (theory)}
\ccsdesc[500]{Information systems~Query optimization}
\ccsdesc[500]{Information systems~Query planning}
\ccsdesc[500]{Information systems~Join algorithms}

\keywords{Conjunctive queries; Plan enumeration; Cost-based optimization}

\maketitle
\sloppy

\input{1_introduction}
\input{2_preliminaries}

\input{3_hierarchical}
\input{4_meta-representation}
\input{5_cost}

\input{6_experiments}
\input{7_conclusions}

\begin{acks}
Qichen Wang is supported by the Ministry of Education, Singapore, through the Academic Research Fund Tier 1 grant (RS32/25), and by the Nanyang Technological University Startup Grant. Part of this work was completed while Qichen Wang was at EPFL. 
\end{acks}

\bibliographystyle{ACM-Reference-Format}
\bibliography{references.bib}

\clearpage
\SuspendCounters{TotPages}
\newpage
\appendix
\input{app_glossary}
\input{app_hierarchical}

\input{app_meta}
\input{app_cost}
\input{app_exp}

\end{document}

%% file: macros.tex
\DeclarePairedDelimiter{\set}{\lbrace}{\rbrace}

\newcommand*{\NP}{\ensuremath{\mathsf{NP}}}

\newcommand*{\prufer}{Pr\"ufer}

\AtEndPreamble{%
	\theoremstyle{acmplain}

}
\newenvironment{sketch}{\begin{proof}[Proof sketch]}{\end{proof}}

\DontPrintSemicolon
\SetKwInOut{Input}{input}
\SetKwInOut{Output}{output}
\SetKwProg{Fn}{Function}{:}{}
\SetKw{Accept}{accept}
\SetKw{Reject}{reject}
\SetKw{ACCEPT}{accept}
\SetKw{REJECT}{reject}

\SetCommentSty{commentfont}

\newcommand{\sys}{metaDecomp}
\usepackage{rotating}
\usepackage{pdflscape}
\usepackage{crossreftools}

\apptocmd\normalsize{%
   \setlength{\abovedisplayskip}{2pt}
   \setlength{\belowdisplayskip}{2pt}
   \setlength{\abovedisplayshortskip}{2pt}
   \setlength{\belowdisplayshortskip}{2pt}
}{}{}

\usepackage{marginnote}

\usepackage{xcolor}

\definecolor{ricolor}{RGB}{86,180,233}
\definecolor{riicolor}{RGB}{230,159,0}
\definecolor{riiicolor}{RGB}{0,158,115}

%% file: 0_abstract.tex
\begin{abstract}
Structural decomposition methods offer powerful theoretical guarantees for join evaluation, yet they are rarely used in real-world query optimizers. A major reason is the difficulty of combining cost-based plan search and structure-based evaluation. In this work, we bridge this gap by introducing meta-decompositions for acyclic queries, a novel representation that succinctly represents all possible join trees and enables their efficient enumeration. Meta-decompositions can be constructed in polynomial time and have sizes linear in the query size. We design an efficient polynomial-time cost-based optimizer based directly on the meta-decomposition, without the need to explicitly enumerate all possible join trees. We characterize plans found by this approach using a novel notion of width, which effectively implies the theoretical worst-case asymptotic bounds of intermediate result sizes and running time of any query plan. Experimental results demonstrate that, in practice, the plans in our class are consistently comparable to---even in many cases better than---the optimal ones found by the state-of-the-art dynamic programming approach, especially on large and complex queries, while our planning process runs by orders of magnitude faster, comparable to the time taken by common heuristic methods. 
\end{abstract}

%% file: 1_introduction.tex
\section{Introduction}
\label{sec:intro}

Joins are the cornerstone of relational database queries.  However, query optimizers face extreme difficulty when planning large multi-way join queries, as the search space of join orders grows exponentially as the number of relations increases.  In fact, finding the optimal join order of a query is well-known to be NP-hard in general~\cite{cluet_complexity_1995,scheufele_constructing_1996}. For this reason, beyond a small number of joins, current query optimizers must abandon exhaustive dynamic programming and resort to heuristic or greedy algorithms.  As a consequence, real-world optimizers often miss the best plans and produce suboptimal ones that can be by orders of magnitude slower.

The database theory community has long been providing tools that exploit query structures to tackle this complexity.  These approaches usually \emph{decompose} queries into specific tree structures, such as join trees~\cite{yannakakis_algorithms_1981} using the GYO algorithm~\cite{grahamGYO,DBLP:conf/compsac/YuO79}, which were later generalized to hypertree decompositions~\cite{DBLP:journals/jcss/GottlobLS02}. The resulting decompositions effectively capture the hardness of queries and, through specific requirements on connectivity, allow one to evaluate queries with theoretically-guaranteed instance-optimal complexity using the Yannakakis algorithm~\cite{bernstein_using_1981,yannakakis_algorithms_1981} or its variants~\cite{wang_yannakakis_2025,DBLP:journals/jcss/GottlobLS02}.
Over the years, these structure-based methods have been generalized to evaluate joins with projections~\cite{PODS251,PODS252,CDE}, aggregations~\cite{AJAR, FAQ}, differences~\cite{SIGMOD23, PODS24}, top-$k$~\cite{wang2023relational}, comparisons~\cite{SIGMOD22}, etc., all with strong theoretical performance guarantees.  These results suggest that leveraging a query’s structural properties can be of great help for query optimization and can dramatically improve query evaluation, with robust worst-case bounds. Yet, in practice, despite limited attempts by, e.g., Scarcello et al.~\cite{scarcello_weighted_2007}, RelationalAI~\cite{submodularwidth}, and Gottlob et al.~\cite{gottlob_reaching_2023}, real-world systems rarely utilize such structure-based approaches. Therefore, they continue to struggle with large, complex queries.

Why are structure-based methods rarely adopted in mainstream systems? A core obstacle is the dichotomy between cost-based plan search and structure-based evaluation~\cite{koutris_database_2026}. Current cost-based optimizers usually search for query plans with optimal or near-optimal cost under some cost model using
(1)~\textbf{exact algorithms}, mostly based on dynamic programming, such as~\cite{selinger_access_1979,vance_rapid_1996,vance_join-order_1998,moerkotte_dynamic_2008,stoian_dpconv_2024}, which have to run in exponential time, and thus do not scale to very large numbers of relations;
(2)~\textbf{heuristic strategies} such as greedy operator ordering (GOO)~\cite{fegaras_new_1998}, 
iterative dynamic programming~\cite{kossmann_iterative_2000}, UnionDP~\cite{mancini_efficient_2022}, and adaptive approaches~\cite{neumann_adaptive_2018,birler_optimizing_2025}, which have neither optimality guarantees nor a reasonable approximation factor unless under some assumption, such as independence of predicate selectivities;
(3)~\textbf{transformation-based techniques}, such as the Volcano/Cascades framework~\cite{graefe_volcano_1993,graefe_cascades_1995} and genetic algorithms~\cite{bennett_genetic_1991};
and
(4) more recently, \textbf{machine learning--based tuning and rewriting}~\cite{marcus_bao_2021,yang_balsa_2022,chen_efficient_2022,zhou_learned_2023,zhou_learned_2021,li_llm-r2_2024}.
Although certain approaches leverage, e.g., \emph{query graphs}~\cite{moerkotte_dynamic_2008,neumann_adaptive_2018,mancini_efficient_2022,birler_optimizing_2025}, they only show the existence of join attributes between relations but do not have a clear connection with the requirements of the aforementioned structural decompositions that help reason about asymptotic complexities.
\input{figs/123-12-13-23}

\begin{example} \label{ex:123-12-13-23-intro}
    Consider the query
    \begin{align*}
        Q_{\ref{ex:123-12-13-23-intro}} \gets \pi_{\emptyset} (& R_1[x_1, x_2, x_3, x_4] \bowtie R_2[x_1, x_2, x_5] \bowtie R_3[x_1, x_3, x_6] \bowtie R_4[x_2, x_3, x_7]),
    \end{align*}
    which is a Boolean acyclic query.
    Its hypergraph is shown in \cref{fig:123-12-13-23-hg}, and a join tree, which is also a hypertree decomposition, in \cref{fig:123-12-13-23-meta}.
    Its query graph, as shown in \cref{fig:123-12-13-23-qg}, is a clique with the four relations as vertices, as each pair shares some join attribute.
    \cref{fig:123-12-13-23-hierarchical-plan} is a query plan that follows the structure of the join tree (which we formally define later). \cref{fig:123-12-13-23-non-hierarchical-plan} does not, but it is still within the search space of traditional algorithms according to the query graph.
\end{example}

Theoretical research typically focuses on worst-case asymptotic data complexity, assuming that cost differences among decompositions are negligible.  
In practice, however, different decompositions may lead to significantly different concrete cost values.
Therefore, using only structural properties, such as hypertree widths, is not sufficient~\cite{gottlob_reaching_2023}.
However, selecting the best decomposition is challenging, as there can be an exponential number of candidates. 
Without an efficient cost-based exploration of alternatives, explicitly scanning through all of them would be computationally prohibitive.
Furthermore, evaluation plans for structure-based methods are not natively supported by current systems that use standard binary query plans~\cite{koutris_database_2026,gottlob_reaching_2023}. As a result, deploying them requires a dedicated optimization and execution pipeline, which limits their practical adoption.

\subsection{Our Contributions}

In this work, we take a major step towards bridging structural decomposition techniques with cost-based query optimization.
We start by introducing \textbf{a novel \emph{width} measurement for query plans} 
and show that, for query plans of width 1, we can effectively embed the query’s tree structure within a traditional query plan to guide query planning even in the absence of cost information.  In particular, we show that acyclic Boolean queries and relation-dominated queries can be evaluated in \emph{instance-optimal time} using width-1 plans.  We also present a method which, for any given join tree of a query, constructs the optimal width-1 query plan that adheres to the join tree in $O(|Q|)$,  \emph{linear in the query size} $|Q|$.

By searching for plans of minimal width, we restrict the search space to only those that follow the query's optimal tree decompositions. To make this practical, we introduce \emph{meta-decompositions} for acyclic conjunctive queries that \textbf{universally encode, in one succinct structure, the entire family of join trees of the query}. Despite the possibly exponential number of join trees of a given query, this representation \textbf{can be constructed in polynomial time and has size linear in the query size}. We give an efficient procedure that enumerates all join trees in polynomial amortized delay from a meta-decomposition, as a potential tool for other structure-based methods requiring explicit join tree enumeration.

Nevertheless, such an enumeration still incurs an exponential overhead. Therefore, we develop an efficient \textbf{\emph{structure-guided cost-based optimization} framework to find the optimal width-1 query plan among all join trees} without na\"ively scanning all candidates. For queries whose meta-decompositions have bounded fan-out, \emph{our algorithm finds the optimal width-1 plan in linear time}. For general queries, we also provide efficient heuristics with the same linear complexity and a low constant factor.

Although we currently restrict our focus to acyclic queries only, such queries in fact account for the overwhelming majority of real workloads~\cite{DBLP:journals/vldb/BonifatiMT20,luo_algorithms_2025}. As shown in \cref{tab:acyclic-in-benchmarks}, among 8,125 queries across several popular benchmarks, 7,929 (97.59\%) are acyclic. Structure-based methods for optimizing acyclic queries have therefore been of sustained interest~\cite{yannakakis_algorithms_1981,wang_yannakakis_2025,RPT2025,yang_predicate_2024,bekkers2024instance,wang2023relational,SIGMOD22,DiamondJoin2024, SIGMOD23, luo_algorithms_2025}.
\begin{table}
    \small
    \centering
    \begin{footnotesize}
    \begin{tabular}{c|c|c|c}
        \hline
        {\bf Benchmark} & {\bf \# Queries} & {\bf \# Acyclic} & {\bf \% Acyclic} \\ \hline
        {\bf TPC-H}~\cite{noauthor_tpc_2022}              & 22         & 21 & 95.45 \%       \\ \hline
        {\bf JOB}~\cite{leis_query_2018} & 113 & 113 & 100 \% \\ \hline
        {\bf STATS-CEB}~\cite{han_cardinality_2021} & 2603 & 2603 & 100 \% \\ \hline
        {\bf Spider-NLP}~\cite{yu_spider_2018} & 4712 & 4709 & 99.94 \% \\ \hline
        {\bf DSB}~\cite{ding_dsb_2021} (SPJ queries only) & 300 & 240 & 80 \% \\ \hline
        {\bf Musicbrainz}~\cite{mancini_efficient_2022} & 375 & 243 & 64.8 \% \\ \hline
        {\bf Total} & 8125 & 7929 & {\bf 97.59 \%} \\ \hline
    \end{tabular}
    \end{footnotesize}
    \caption{Number of acyclic queries in major benchmarks. Part of the data comes from \cite{luo_algorithms_2025}.}
    \label{tab:acyclic-in-benchmarks}
    \vspace{-1\baselineskip}
\end{table}

We then empirically demonstrate that our approach can indeed be game-changing for query optimization for large joins. %
We \textbf{implement an optimizer \sys{}} and evaluate it on queries from 4 benchmarks, which include many large, complex queries.  We compare the results against several state-of-the-art approaches. The results show that width-$1$ plans consistently match or even outperform the plans found by these baselines, while taking very small amounts of time for optimization.
We also demonstrate that join tree enumeration can be by orders of magnitude faster than existing approaches based on na\"ive GYO reductions, and that selecting join trees using our cost-based optimization framework significantly improves the effectiveness of existing structure-based~methods.

\paragraph{Paper organization} The rest of the paper is organized as follows:
After laying out the necessary background in \cref{sec:preliminaries}, we first define the \emph{width} notion of query plans and discuss the connection between query structures and the widths of query plans in \cref{sec:hierarchical-query-plan}.
In \cref{sec:meta}, we define \emph{meta-decompositions}, the algorithm to construct them, as well as our way of enumerating all join trees based on them.
In \cref{sec:opt}, we demonstrate how to perform cost-based optimization using meta-decompositions.
In \cref{sec:exp}, we present experimental results to show the efficiency of query optimization and execution using this approach.
We conclude and discuss possible future work in \cref{sec:conclusions}.  
Due to space constraints, missing proofs can be found in the \arxivorappendix.

%% file: figs/123-12-13-23.tex
\begin{figure}
    \centering
    \begin{subfigure}{.33\linewidth}
        \small
        \centering
        \begin{tikzpicture}[hypergraph,xscale=.7,yscale=.5,trim left=0cm,trim right=0cm]
            \node[vertex,label={east:$x_1$}] (x1) at (0, 0) {};
            \node[vertex,label={east:$x_2$}] (x2) at (-2, -2) {};
            \node[vertex,label={west:$x_3$}] (x3) at (2, -2) {};
            \node[vertex,label={east:$x_4$}] (x4) at (0, 0.75) {};
            \node[vertex,label={west:$x_5$}] (x5) at (-1, -1) {};
            \node[vertex,label={east:$x_6$}] (x6) at (1, -1) {};
            \node[vertex,label={west:$x_7$}] (x7) at (0, -2) {};
        
            \begin{scope}[on background layer]
                
                \newcommand*{\pathri}{($(x2) + (0,-0.4)$)
                    to[out=0,in=-135] ($(x5) + (-0.25,0.25)$)
                    to[out=45,in=180] ($(x1) + (0,-0.5)$)
                    to[out=0,in=135] ($(x6) + (0.25,0.25)$)
                    to[out=-45,in=180] ($(x3) + (0,-0.4)$)
                    to[out=0,in=0] ($(x4) + (0,0.5)$)
                    to[out=180,in=180] ($(x2) + (0,-0.4)$)
                }

                \newcommand*{\pathrii}{	($(x1) + (0.3,0.3)$)
                    to[out=135,in=45] ($(x5) + (-0.45,0.45)$)
                    to[out=-135,in=135] ($(x2) + (-0.3,-0.3)$)
                    to[out=-45,in=-135] ($(x5) + (0.45,-0.45)$)
                    to[out=45,in=-45] ($(x1) + (0.3,0.3)$)}
                
                \newcommand*{\pathriii}{($(x1) + (-0.3,0.3)$)
                    to[out=45,in=135] ($(x6) + (0.45,0.45)$)
                    to[out=-45,in=45] ($(x3) + (0.3,-0.3)$)
                    to[out=-135,in=-45] ($(x6) + (-0.45,-0.45)$)
                    to[out=135,in=-135] ($(x1) + (-0.3,0.3)$)}
                    
                \newcommand*{\pathriv}{	($(x2) + (0,0.4)$)
                    to[out=0,in=180] ($(x3) + (0,0.4)$)
                    to[out=0,in=90] ($(x3) + (0.5,0)$)
                    to[out=-90,in=0] ($(x3) + (0,-0.4)$)
                    to[out=180,in=0] ($(x2) + (0,-0.4)$)
                    to[out=180,in=-90] ($(x2) + (-0.5,0)$)
                    to[out=90,in=180] ($(x2) + (0,0.4)$)}
                
                \draw[hyperedge,looseness=.85, draw=BurntOrange,fill=BurntOrange,fill opacity=.1] \pathri;
                \draw[hyperedge,looseness=.85,draw=Red,fill=Red,fill opacity=.1] \pathrii;
                \draw[hyperedge,looseness=.85,draw=Fuchsia,fill=Fuchsia,fill opacity=.1] \pathriii;
                \draw[hyperedge,looseness=.85,draw=NavyBlue,fill=NavyBlue,fill opacity=.1] \pathriv;

                \node[text=BurntOrange] at (2,0.75) {$R_1$};
                \node[text=Red] at (-0.5,-1) {$R_2$};
                \node[text=Fuchsia] at (0.5,-1) {$R_3$};
                \node[text=Cyan] at (0.5,-2) {$R_4$};
                
            \end{scope}
        \end{tikzpicture}
        \caption{Query hypergraph}
        \label{fig:123-12-13-23-hg}
    \end{subfigure}
    \hfill
    \begin{subfigure}{.66\linewidth}
        \small
        \centering
        \begin{tikzpicture}[
			level distance=1.1cm,
			level 1/.style={sibling distance=2cm}
        ]
			\node { $\lambda: \{ R_1 \}$, $\chi : \{ x_1, x_2, x_3, x_4 \}$ }
				child { node[align=center] {$\lambda: \{ R_2 \}$ \\ $\chi: \{ x_1, x_2, x_5 \}$} edge from parent[-] }
                child { node[align=center] {$\lambda: \{ R_3 \}$ \\ $\chi: \{ x_1, x_3, x_6 \}$} edge from parent[-] }
                child { node[align=center] {$\lambda: \{ R_4 \}$ \\ $\chi: \{ x_2, x_3, x_7 \}$} edge from parent[-] } ;
		\end{tikzpicture}
        \caption{A join tree / hypertree decomposition}
        \label{fig:123-12-13-23-meta}
    \end{subfigure}
    \hfill\ 

    \hfill
    \begin{subfigure}{.32\linewidth}
        \centering
        \begin{tikzpicture}[xscale=.75,yscale=.75]
            \node[label={center:$R_1$}] (R1) at (0,1) {};
            \node[label={center:$R_2$}] (R2) at (-1,0) {};
            \node[label={center:$R_3$}] (R3) at (0,-1) {};
            \node[label={center:$R_4$}] (R4) at (1,0) {};
            \draw (R1) -- (R2);
            \draw (R1) -- (R3);
            \draw (R1) -- (R4);
            \draw (R2) -- (R3);
            \draw (R2) -- (R4);
            \draw (R3) -- (R4);
        \end{tikzpicture}
        \caption{Query graph}
        \label{fig:123-12-13-23-qg}
    \end{subfigure}
    \hfill
    \begin{subfigure}{.32\linewidth}
        \centering
        \begin{tikzpicture}[
			level distance=0.6cm,
			level 1/.style={sibling distance=1.25cm},
        ]
			\node { $\bowtie$ }
                child { node {$\bowtie$} edge from parent[-]
    				child { node {$\bowtie$} edge from parent[-]
                        child { node {$R_1$} edge from parent[-] }
                        child { node {$R_2$} edge from parent[-] }
                    }
                    child { node {$R_3$} edge from parent[-] }            
                } child { node {$R_4$} edge from parent[-] }
            ;
		\end{tikzpicture}
        \caption{A width-1 query plan}
        \label{fig:123-12-13-23-hierarchical-plan}
    \end{subfigure}
    \hfill
    \begin{subfigure}{.32\linewidth}
        \centering
        \begin{tikzpicture}[
			level distance=0.6cm,
			level 1/.style={sibling distance=1.25cm}
        ]
			\node { $\bowtie$ }
                child { node {$\bowtie$} edge from parent[-]
    				child { node {$\bowtie$} edge from parent[-]
                        child { node {$R_2$} edge from parent[-] }
                        child { node {$R_3$} edge from parent[-] }
                    }
                    child { node {$R_1$} edge from parent[-] }            
                } child { node {$R_4$} edge from parent[-] }
            ;
		\end{tikzpicture}
        \caption{A width-2 query plan}
        \label{fig:123-12-13-23-non-hierarchical-plan}
    \end{subfigure}
    \hfill\ 
    \caption{The hypergraph, a join tree, the query graph, and two query plans of the query $Q_{\ref{ex:123-12-13-23-intro}}$ in \cref{ex:123-12-13-23-intro}}
    \label{fig:123-12-13-23}
\end{figure}
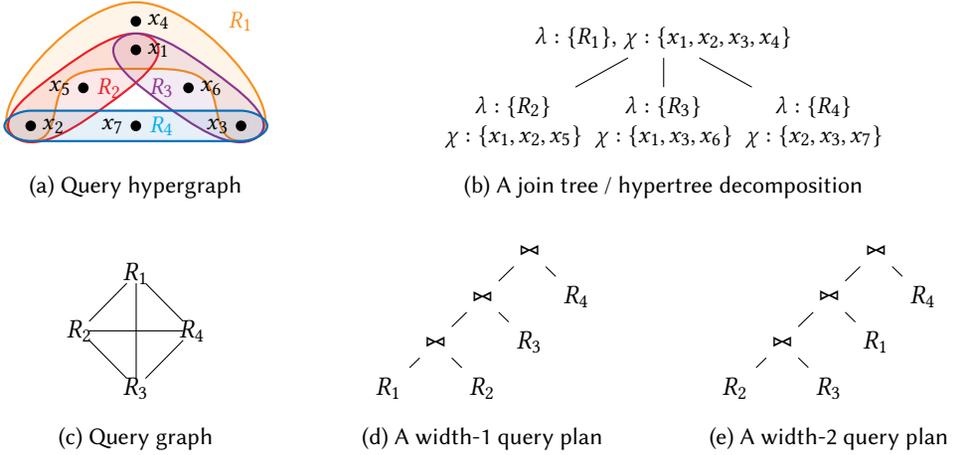

%% file: 2_preliminaries.tex
\section{Preliminaries}\label{sec:preliminaries}

We write $2^S$ for the set of all subsets, i.e., the power set, of a set $S$. If a function $f : A \to B$ and $S \subseteq A$, we let $f(S) = \set{ f(a) : a \in S }$. For a set of sets $S$, we write $\cup S = \bigcup_{x \in S} x$.  If $f : A \to T$, where $T$ is a set, and $S \subseteq A$, we let $f(S) = \bigcup_{x \in S} f(x)$.

\subsection{Conjunctive Queries and Hypergraphs}

Conjunctive queries can be represented by hypergraphs, and our terminology follows the seminal references on this topic~\cite{DBLP:journals/jcss/GottlobLS02,DBLP:journals/jcss/GottlobLS03,DBLP:conf/wg/GottlobGMSS05,DBLP:conf/pods/GottlobGLS16,wang_yannakakis_2025}.
Let $\mathcal{D} = \{R_1[\bar{x}_1], \cdots, R_n[\bar{x}_n]\}$ be a database consisting of $n$ relations, where $\bar{x}_i$ is the set of attributes of relation $R_i$. A conjunctive query (with self-predicates) is defined as 
\[
    Q \gets \pi_{\bar{y}}\left(\sigma_{1}(R_1[\bar{x}_1]) \Join \cdots \Join \sigma_{n}(R_n[\bar{x}_n]\right)),
\]
where $\sigma_{i}$ is the set of selection predicates involving only attributes in $\bar{x}_i$, $\bar{y}$ is the set of output attributes, and $\Join$ denotes natural join (with attribute renaming to avoid ambiguity). If $\bar{y} = \bigcup_{i \in [n]} \bar{x}_i$, the query is called a full join query and we omit $\bar{y}$ for simplicity; if $\bar{y} = \emptyset$, the query is a Boolean query that returns true if there exists a tuple satisfying all specified joins and predicates, or false otherwise.  We also require the query to be \emph{self-join-free} at the logical level:  If multiple relations correspond to the same physical relation, they are treated as distinct logical relations with unique~renamings.

We define the associated \emph{hypergraph} for the conjunctive query $Q$ as $H=(V(H), E(H))$, where the vertices set $V(H)$ consists of all attributes appearing in the query, i.e., $V(H) = \bigcup_{i \in [n]} \bar{x}_i$, and the \emph{hyperedges} set $E(H) \subseteq 2^{V(H)} \setminus \set{ \emptyset }$ contains one hyperedge for each relation, i.e., $E(H) = \{e_1, \cdots, e_n\}$, where each $e_i = \bar{x}_i$.

\subsection{Acyclicity and Join Trees}

A conjunctive query is acyclic if its associated hypergraph is acyclic, and we define acyclicity in the standard way in database theory, e.g.,~\cite[Section 6.4]{abiteboul_foundations_1996}. Namely, a hypergraph $H$ is acyclic if and only if the output of the GYO algorithm~\cite{grahamGYO, DBLP:conf/compsac/YuO79} on $H$ is empty~\cite[Theorem 6.4.5]{abiteboul_foundations_1996}. Equivalently, $H$ is acyclic if and only if there exists a \emph{join tree} \cite[Theorem 6.4.5]{abiteboul_foundations_1996}, or, a width-1 hypertree decomposition~\cite{DBLP:journals/jcss/GottlobLS02}.

Slightly differently from the usual definition of join trees~\cite{abiteboul_foundations_1996}, where each vertex is simply a relation, in this paper, we define them using notations for hypertree decompositions~\cite{DBLP:journals/jcss/GottlobLS02}, to make it easier to generalize to the meta-decompositions that we will define later.

Let $Q$ be an acyclic query with associated hypergraph $H$. A \emph{join tree} of $Q$ is a labeled tree $T = (V(T), r(T), E(T), \chi, \lambda)$
 with vertices $V(T)$, root $r(T)$, and edges $E(T)$.
$\chi : V(T) \to 2^{V(H)}$ is a function s.t.\ 
$\bigcup_{ p \in V(T) } \chi(p) = V(H)$,
and $\lambda : V(T) \to 2^{E(H)}$.\footnote{For simplicity, we may use hyperedges (e.g., $R_1$) to label and identify the vertices of join trees instead of writing the full $\lambda$- and $\chi$-labels when the context is clear.} And $T$ needs to satisfy the following~ conditions:
\begin{enumerate}[label=(C\textsubscript{\arabic*}),ref=(C\textsubscript{\arabic*}),leftmargin=*]
	\item\label{H1}\gdef\Giref{\ref{H1}} \textbf{Coverage condition.}  For each $e \in E(H)$, there exists exactly one $p \in V(T)$ with $\lambda(p) = \set{e}$.
	\item\label{H2}\gdef\Giiref{\ref{H2}} \textbf{Connectedness condition.} For each $v \in V(H)$, the set of vertices $\set{ p \in V(T) : v \in \chi(p) }$ induces a connected subtree of $T$.
    Or, equivalently, for all pairs of vertices $p, q \in V(T)$ and all $s$ on the (unique) path between $p$ and $q$ on $T$, we have $\chi(p) \cap \chi(q) \subseteq \chi(s)$.
	\item\label{H3}\gdef\Giiiref{\ref{H3}} \textbf{Width-1.} For all $p \in V(T)$, $|\lambda(p)| = 1$ and, without loss of generality, $\chi(p) = V(\lambda(p)) $.
\end{enumerate}

\begin{example}
    We revisit the query $Q_{\ref{ex:123-12-13-23-intro}}$ in \cref{ex:123-12-13-23-intro}. Its hypergraph is illustrated in \cref{fig:123-12-13-23-hg}. It is easy to verify that the join tree in \cref{fig:123-12-13-23-meta} satisfies conditions \ref{H1} and \ref{H3}. Regarding condition \ref{H2}, consider, e.g., attribute \( x_1 \). The vertices in the join tree that include \( x_1 \) are the ones with \( R_1 \), \( R_2 \), and \( R_3 \), which induce a connected subtree. Similar arguments apply for all other attributes as well.
\end{example}

We assume all trees are rooted and consider rerootings as distinct trees.
We write $T_p$ as the subtree rooted at $p$, i.e., with $p$ and all its descendants. For every neighbor (child or parent) $q$ of vertex $p$ in $T$, we write $T_{p \to q}$ for the subtree with
\[
    V( T_{p \to q} ) = 
    \set{ s \in V(T) \setminus \{ p \} : \text{the path from } p \text{ to } s \text{ in } T \text{ contains } q }\text.
\]
Intuitively, if $q$ is a child of $p$, it is the same as $T_q$. If $q$ is a parent of $p$, it denotes the subtree of $T$ with all vertices except those in $T_p$.

We define the \emph{fan-out} $f(p)$ of any $p \in V(T)$ as 
the number of children of $p$ on $T$.  The \emph{fan-out} of the entire tree $T$, denoted by $f(T)$, is the maximum fan-out among all nodes of $T$.

Given a join tree $T$, for every $p \in V(T)$, we define the \emph{induced query} $Q(p)$ as the join of all relations in the subtree rooted at $p$. 

One important observation is that \emph{an acyclic conjunctive query may have an exponential number of distinct join trees}.

\begin{theorem}
    \label{thm:exp-join-trees}
    Consider \emph{star queries}\footnote{We note that, in some literature, e.g.,~\cite{birler_optimizing_2025}, such queries are called ``\emph{clique}'' queries, because their \emph{query graphs} are cliques. 
    }
    of the form
    \[
        Q_{\rm star} \gets R_1[\bar{x}_1] \Join \cdots \Join R_n[\bar{x}_n], 
    \]
    where for any $i, j \in [n]$, $\bar{x}_i \cap \bar{x}_j = \set{x}$, for some attribute $x$.  There are $n^{n-1}$ possible join trees for such a star query with $n$ relations. 
\end{theorem}

\begin{example}
    \label{ex:star}
    The query $Q_{\ref{ex:star}} \gets R_1[x_1, x_2] \bowtie R_2[x_1, x_3] \bowtie R_3[x_1, x_4] \bowtie R_4[x_1, x_5]$ is such a star query. \cref{fig:star-hg} shows its associated hypergraph, and \cref{fig:star-jts} shows two of its valid join trees.
    \input{figs/star}
\end{example}

In fact, the possible join trees of such queries in \cref{thm:exp-join-trees} and \cref{ex:star} are simply \emph{all possible trees with $n$ vertices} (or, equivalently, all spanning trees of a complete graph $K_n$), each labeled by one relation. Any such tree satisfies the connectedness condition~\ref{H2}. 
There are a total of $n^{n-1}$ such trees ($n^{n-2}$ distinct non-rooted trees~\cite{prufer_neuer_1918}, each of which has $n$ possible root nodes).

\subsection{Query Plans and Join Ordering}
\label{sec:prelim-plans}

A \emph{query plan} $\mathcal{P} = (V(\mathcal{P}), r(\mathcal{P}), E(\mathcal{P}))$ for a query $Q$ is a tree rooted at $r(\mathcal{P})$, where each non-leaf node is a relational operator (e.g., join, projection, selection), and each leaf node corresponds to a base relation. 
A plan specifies an order of operations to evaluate the query, where the outputs of child nodes serve as inputs to the parent.
In this work, we focus on query optimization for binary joins, where each join operation in the query plan has exactly two child nodes.  
For the purpose of join ordering, \emph{we use only join operators in internal nodes and assume selections and projections are inherently integrated in the join
or scan
operation at the lowest possible level.}
Namely, in each step of join or scan, we apply a filter condition if all relevant attributes already appear in the subtree, and we project out attributes that neither appear in the output nor act as a join key with some remaining relation.

Given a query plan $\mathcal{P}$, for every $p \in V(\mathcal{P})$, we define the \emph{induced query} $Q(p)$ as the join of all relations in the subtree rooted at $p$ (with all possible projection and filters inherently applied).

We exclude Cartesian products: for every join node, the outputs of its children must share some common attribute as the join key.

A query plan is called \emph{left-deep} (or \emph{right-deep}) if all join operations have their right (left) child node corresponding to a base relation; otherwise, the query plan is referred to as \emph{bushy}.

We apply cost functions $\mathcal{C}$ on query plans, in the following form, to measure the optimality of query plans.
\[
    \mathcal{C}(\mathcal{P})= \left\{ \begin{array}{ll}
        c(r(\mathcal{P})), & \text{if $r(\mathcal{P})$ is single relation} \\
        c(r(\mathcal{P}))\oplus \mathcal{C}(\mathcal{P}_1) \oplus \mathcal{C}(\mathcal{P}_2), & \text{if $r(\mathcal{P})$ is join operation}
    \end{array} 
    \right.
\]
where $c(r(\mathcal{P}))$ is the cardinality of the output of $r(\mathcal{P})$; $\mathcal{P}_1$ and $\mathcal{P}_2$ (if applicable) are the left and right subtrees of $r(\mathcal{P})$, respectively; and $\oplus$ is some operator to aggregate $c(r)$, $\mathcal{C}(\mathcal{P}_1)$, and $\mathcal{C}(\mathcal{P}_2)$.  In this work, we adopt the cost function $C_{\rm out}$ to minimize the total intermediate result size by substituting the $\oplus$ operator simply with arithmetic $+$, as is standard in, e.g., \cite{stoian_dpconv_2024,fegaras_new_1998,neumann_adaptive_2018,cluet_complexity_1995}.  

The typical approach to find the optimal plan is to recursively find the optimal split of a set of relations into two disjoint subsets~\cite{selinger_access_1979,vance_rapid_1996,vance_join-order_1998,moerkotte_dynamic_2008,stoian_dpconv_2024}.   With a given cost function $\mathcal{C}$, the problem can be solved by the dynamic programming~(DP) algorithm DPccp~\cite{moerkotte_dynamic_2008} in $O(3^n)$. The recently proposed algorithm DPconv~\cite{stoian_dpconv_2024} reduces the complexity for the cost model $C_{\rm out}$ to $O(2^n n^3 W \log (Wn))$, where $W$ is the largest join cardinality.

%% file: figs/star.tex
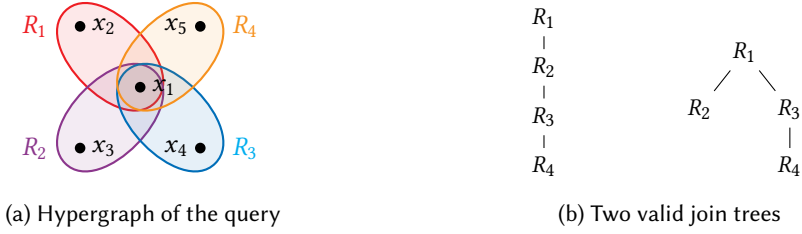
\begin{figure}
    \centering
    \begin{subfigure}{.49\linewidth}
        \centering
        \vspace{-\baselineskip}
        \begin{tikzpicture}[hypergraph,xscale=.8,yscale=0.8,trim left=0cm,trim right=0cm]
            \node[vertex,label={east:$x_1$}] (x1) at (0, 0) {};
            \node[vertex,label={east:$x_2$}] (x2) at (-1, 1) {};
            \node[vertex,label={east:$x_3$}] (x3) at (-1, -1) {};
            \node[vertex,label={west:$x_4$}] (x4) at (1, -1) {};
            \node[vertex,label={west:$x_5$}] (x5) at (1, 1) {};
        
            \begin{scope}[on background layer]
                
                \newcommand*{\pathri}{($(x2) + (-0.25,0.25)$)
                    to[out=45,in=45] ($(x1) + (0.25,-0.25)$)
                    to[out=-135,in=-135] ($(x2) + (-0.25,0.25)$)}
                    
                \newcommand*{\pathrii}{	($(x1) + (0.25,0.25)$)
                    to[out=135,in=135] ($(x3) + (-0.25,-0.25)$)
                    to[out=-45,in=-45] ($(x1) + (0.25,0.25)$)}
                
                \newcommand*{\pathriii}{($(x1) + (-0.25,0.25)$)
                    to[out=45,in=45] ($(x4) + (0.25,-0.25)$)
                    to[out=-135,in=-135] ($(x1) + (-0.25,0.25)$)}

                \newcommand*{\pathriv}{	($(x5) + (0.25,0.25)$)
                    to[out=135,in=135] ($(x1) + (-0.25,-0.25)$)
                    to[out=-45,in=-45] ($(x5) + (0.25,0.25)$)}
                
                \draw[hyperedge,looseness=1,draw=Red,fill=Red,fill opacity=.1] \pathri;
                \draw[hyperedge,looseness=1,draw=Fuchsia,fill=Fuchsia,fill opacity=.1] \pathrii;
                \draw[hyperedge,looseness=1,draw=NavyBlue,fill=NavyBlue,fill opacity=.1] \pathriii;
                \draw[hyperedge,looseness=1, draw=BurntOrange,fill=BurntOrange,fill opacity=.1] \pathriv;
        
                \node[text=Red] at (-1.75,1) {$R_1$};
                \node[text=Fuchsia] at (-1.75,-1) {$R_2$};
                \node[text=Cyan] at (1.75,-1) {$R_3$};
                \node[text=BurntOrange] at (1.75,1) {$R_4$};
                
            \end{scope}
        \end{tikzpicture}
        \vspace{-0.5\baselineskip}
        \caption{Hypergraph of the query}
        \label{fig:star-hg}
    \end{subfigure}
    \hfill
    \begin{subfigure}{.49\linewidth}
        \small
        \centering
        \hfill
        \begin{tikzpicture}[
			level distance=0.65cm,
			level 1/.style={sibling distance=1.2cm},
        ]
			\node { $R_1$ }
				child { node {$R_2$} edge from parent[-]
                    child { node {$R_3$} edge from parent[-]
                        child { node {$R_4$} edge from parent[-] }
                    } 
                };
		\end{tikzpicture}
        \hfill
        \begin{tikzpicture}[
			level distance=0.75cm,
			level 1/.style={sibling distance=1.2cm},
        ]
			\node { $R_1$ }
				child { node {$R_2$} edge from parent[-]
                }
                child { node {$R_3$} edge from parent[-]
                    child { node {$R_4$} edge from parent[-] }
                } ;
		\end{tikzpicture}
        \hfill
        \ 
        \caption{Two valid join trees}
        \label{fig:star-jts}
    \end{subfigure}
    \vspace{-0.5\baselineskip}
    \caption{Hypergraph and two valid join trees of the star query $Q_{\ref{ex:star}}$ in \cref{ex:star}}
    \label{fig:star}
\end{figure}

%% file: 3_hierarchical.tex
\section{Widths of Query Plans and Width-1 Query Plans}
\label{sec:hierarchical-query-plan}

In this section, we begin by introducing our notion of widths for query plans, which connects a plan's structure to the potential sizes of its intermediate results.  
We consider the ``interface'' of each intermediate step in the query plan, i.e., the set of attributes that serve as join keys with the remaining relations.
Intuitively, the width of a query plan shows the number of relations needed to cover the interface of each step in the query plan.
This measure helps quantify the sizes of intermediate results and the complexity of a query plan before looking at the specific database~instance.  
\input{figs/hierarchical}
\begin{example} \label{ex:hierarchical}
    Consider the query 
    \[Q_{\ref{ex:hierarchical}} \gets \pi_\emptyset(R_1[x_1, x_2, x_3] \bowtie R_2[x_1, x_4, x_5] \bowtie R_3[x_5, x_6] \bowtie R_4[x_3, x_7]). \]
    Its hypergraph is shown in \cref{fig:hierarchical-hg} and a join tree in \cref{fig:hierarchical-jt}.
    Given this join tree, \cref{fig:hierarchical-plan} is an example of a valid width-1 query plan.
    After joining $R_1$ and $R_4$, $x_1$ is the only attribute in the intersection with the remaining relations $R_2$ and $R_3$, and it can be covered by only one relation $R_1$.
    Similarly, at the step that joins $R_2$ and $R_3$, $x_1$ is the only attribute in the interface.
    A width of 1 implies that, if the size of any one single relation is bounded by $N$, then the size of each intermediate step can be bounded by $N$, as it is at most the size of the relation that covers the interface.
    
    On the contrary, \cref{fig:non-hierarchical-plan} shows a query plan that is width-2: When we join $R_1$ and $R_2$, the interface contains $x_5$ and $x_3$, which are needed to join $R_3$ and $R_4$, respectively.
    But there is no single relation that can cover both $x_5$ and $x_3$, so the intermediate sizes can be as high as $N^2$, instead of $N$ as in the width-1 plan in \cref{fig:hierarchical-plan}.
\end{example}

More formally, given a query plan $\mathcal{P} = (V(\mathcal{P}), r(\mathcal{P}), E(\mathcal{P}))$, for any node $p \in V(\mathcal{P})$, let $H_p = (V(H_p), E(H_p))$ be the hypergraph induced by the relations being joined in the subtree rooted at $p$, and let the induced query $Q(p) = \mathop{\Join}_{e \in E(H_p)} e$. Let $\overline H_p = (V(\overline H_p), E(\overline H_p))$ be the \emph{residual hypergraph} with relations that are not yet joined in the subtree rooted at $p$, i.e., $E(\overline H_p) = E(H) \setminus E(H_p)$, and $V(\overline H_p) = \cup(E(\overline H_p))$.  We define the \emph{interface} $I(p) = V(H_p) \cap V(\overline H_p)$, i.e., the set of attributes that will serve as join keys with the remaining relations. 
The \emph{width of node} $p$ is the smallest number of relations needed to cover this interface, i.e.,
\[
    w(p) = \min \left\{ |S| : S \subseteq E(H_p) \text{ and } I(p) \subseteq \cup S \right\}.
\]
The width of the query plan $\mathcal{P}$ is the maximum width over all nodes.

The width of a query plan indicates an upper bound on the minimal intermediate results over all nodes $p$ during evaluation.

\begin{theorem} \label{thm:w1plan-size-bound}
Given a Boolean conjunctive query $Q$ and a query plan $\mathcal{P}$ for the query, the maximum intermediate result size, i.e.,
    $\max_{p \in V(H_p)} |\pi_{I(p)} Q(p)|,$
is upper-bounded by $O(N^{w(\mathcal{P})})$ in the worst case, where $N$ is the maximum cardinality of any base relation.
\end{theorem}

This notation of width has a direct connection with join trees: %

\begin{example}
    We revisit the query $Q_{\ref{ex:hierarchical}}$ from \cref{ex:hierarchical}.
    Take $p$ as the node of the join tree with $R_2$, as an example. The induced query $Q(p) = R_2 \bowtie R_3$. In the width-1 query plan in \cref{fig:hierarchical-plan}, the right child of the root node, which we denote by $q$, induces exactly $Q(p) = R_2 \bowtie R_3$.
    In the other direction, for the node $q$ of the query plan, we consider its children $u$ and $v$ (which correspond to the leaf nodes with $R_2$ and $R_3$ respectively). $u$ has interface $I(u) = \{ x_1, x_5 \}$, which can be covered by a single relation $R_2$. And $v$ has interface $I(v) = \{ x_5 \}$, which can be covered by $R_3$.
    However, for the width-2 query plan in \cref{fig:non-hierarchical-plan}, there is no such $q$ of the query plan such that $Q(p) = Q(q) = R_2 \bowtie R_3$.
\end{example}

We formally define this correspondence as follows:

\begin{definition} \label{def:w1plan-induced-by-jt}
    We say a query plan $\mathcal{P}$ is \emph{induced by} a join tree $T$~if
    \begin{enumerate}[leftmargin=*]
        \item for each node $p \in V(T)$, there exists $q \in V(\mathcal{P})$ with $Q(p) = Q(q)$,
        i.e., the induced query of node $p$ in the join tree is exactly the induced query of node $q$ in the query plan; and
        \item for each non-leaf node $q \in V(\mathcal{P})$ with two children $u$ and $v$, there exists some $s \in E(H_u)$ and $t \in E(H_v)$ such that $I(u) \subseteq s$, $I(v) \subseteq t$, and there exists an edge between $s$ and $t$ in $E(T)$.
    \end{enumerate}
\end{definition}

Under this definition, we have that

\begin{theorem} \label{thm:w1plan-jt}
    Given a query plan $\mathcal{P}$ for a query $Q$, the width $w(\mathcal{P}) = 1$ if and only if there exists a join tree $T$ of $Q$ that induces $\mathcal{P}$.
\end{theorem}

Theorem~\ref{thm:w1plan-jt} directly leads to the relationship between width-1 query plans and acyclic queries:

\begin{theorem}\label{thm:acyclicwidth}
    A conjunctive query $Q$ has width-1 query plans if and only if $Q$ is acyclic.
\end{theorem}

Width-1 query plans are closely related to structure-based query evaluation methods. Such methods follow a similar pattern: given a join tree, they apply a series of reductions from the leaves upwards until only a single node remains.  
Width-1 query plans exploit the join tree structure in a similar way, evaluating the query bottom-up: each node $p\in V(T)$ completes all its internal joins with its child subtrees before joining with its parent. Such a strategy provides several benefits.

\paragraph{Limiting intermediate result sizes.}
We have seen that, by projecting out attributes that are neither part of the final output nor in the interface at each step, we limit the sizes of intermediate results, as shown by \cref{thm:w1plan-size-bound}.
Although this approach excludes some alternative plans with higher widths, the excluded ones do not preserve such a theoretical guarantee.
    Taking again \cref{ex:hierarchical}, 
    if we choose the width-2 query plan in \cref{fig:non-hierarchical-plan}, $x_3$ has to be kept in all intermediate results before $R_4$ is joined at the very end, whereas, in the width-1 query plan in \cref{fig:hierarchical-plan}, it can already be projected out after completing $R_1 \bowtie R_4$.

\paragraph{Local optimization.}  Instead of searching in a vast space of arbitrary join orders for all relations, we can locally optimize the join order at each node of the join tree. For each node $p \in V(T)$ with children $c_1, \cdots, c_{f(p)}$ in the join tree, we only need to determine the optimal join order between $p$ and $Q(c_1), \dots, Q(c_{f(p)})$. The order in which these joins are performed can be optimized independently at each node, without considering other parts of the query. Such an approach, therefore, significantly prunes the search space for optimization, making plan selection much more tractable, and can yield near-optimal performance for most practical queries. In the experiments~(\cref{sec:exp}), we show in \cref{tab:query-props} that the fan-out $f(T)$ tends to be very low in real benchmarks.

\begin{example}
    Consider again the join tree in \cref{fig:hierarchical-jt} for the query $Q_{\ref{ex:hierarchical}}$. For width-1 query plans induced by this join tree, there must exist a subtree with $R_2 \bowtie R_3$, as required by \cref{def:w1plan-induced-by-jt}(1). It therefore remains only to optimize the order to join $R_1$, $(R_2 \bowtie R_3)$, and $R_4$, which is the local optimization problem at the root node $R_1$ of the join tree. If the cardinalities $|R_1 \bowtie R_4| < |(R_2 \bowtie R_3) \bowtie R_1|$, then the optimal plan in our cost model is the one in \cref{fig:hierarchical-plan}.
\end{example}

With only local search needed, assuming a low fan-out, finding the optimal width-1 query plans induced by a join tree can be done more efficiently than through a global search:

\begin{theorem} \label{thm:opt-hierarchical-time}
    Given a join tree $T$, the optimal width-1 query plan induced by $T$ can be found in time $O(f(T)2^{f(T)} |Q|)$.
\end{theorem}

\paragraph{Structure-guided query evaluation.} An additional advantage of width-1 query plans is that they naturally accommodate structure-guided query evaluation methods. For example, we can simulate the Yannakakis algorithm for evaluating \emph{relation-dominated} conjunctive queries \cite{wang_yannakakis_2025} by following a width-1 query plan. A conjunctive query is relation-dominated if it is acyclic and there exists a relation $R_i[\bar{x}_i]$ that contains all the output attributes, i.e., $\bar{y} \subseteq \bar{x}_i$.  It shall be noted that an acyclic Boolean query is also a relation-dominated query, since $\bar{y} = \emptyset$.
For this type of query, we have the following guarantee of evaluation time:

\begin{theorem}
    \label{thm:rel-dom}
    A relation-dominated conjunctive query can be evaluated in $O(|db|)$ time using a width-1 query plan, where $|db|$ represents the size of the database.
\end{theorem}

To create a linear-time width-1 query plan for evaluating a relation-dominated query, let $R_i$ denote the relation containing all output attributes. One can construct the query plan using any join tree with $R_i$ as the root. The width-1 query plan will then evaluate the query bottom-up. Because all output attributes are located at the root of the join tree and at the end of the query plan, when the query plan is evaluating node $p$, all attributes in $T_p$ but not in $p$ are removed from the query, guaranteeing linear running time.

To conclude this section, we highlight some important distinctions between our definition of width-1 query plans and Yannakakis-style query plans~\cite{yannakakis_algorithms_1981,wang_yannakakis_2025}. Although they both evaluate queries by traversing join trees, width-1 query plans do not use full reductions via semi-joins as in Yannakakis-style algorithms. We have shown, nevertheless, that these plans, by appropriate projections of irrelevant attributes, still preserve favorable theoretical guarantees of asymptotic complexity for Boolean conjunctive queries (\cref{thm:w1plan-size-bound}) and relation-dominated queries (\cref{thm:rel-dom}). In \cref{sec:exp}, we will experimentally show that these plans, in practice, are indeed highly efficient. Furthermore, without the need for full reductions, width-1 query plans are much easier to integrate into the query execution frameworks of current database systems, enabling more efficient cost-based optimization, as we will discuss in the remainder of the paper.

%% file: figs/hierarchical.tex
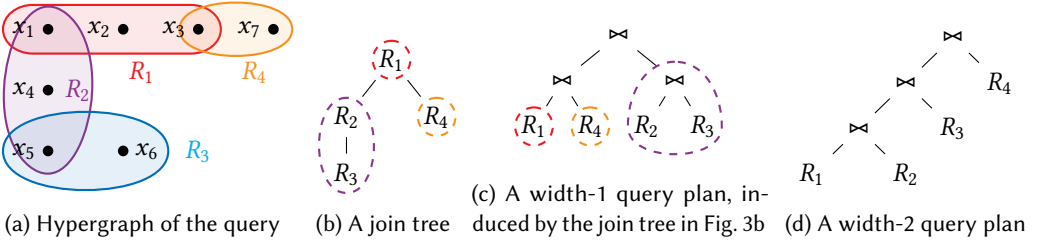
\begin{figure}
    \centering
    \begin{subfigure}{.29\linewidth}
        \centering
        \vspace{-0.5\baselineskip}
        \hspace{-1.5em}
        \begin{tikzpicture}[hypergraph,xscale=1,yscale=0.8,trim left=0cm,trim right=0cm]
            \node[vertex,label={west:$x_1$}] (x1) at (-1, 0) {};
            \node[vertex,label={west:$x_2$}] (x2) at (0, 0) {};
            \node[vertex,label={west:$x_3$}] (x3) at (1, 0) {};
            \node[vertex,label={west:$x_4$}] (x4) at (-1, -1) {};
            \node[vertex,label={west:$x_5$}] (x5) at (-1, -2) {};
            \node[vertex,label={east:$x_6$}] (x6) at (0, -2) {};
            \node[vertex,label={west:$x_7$}] (x7) at (2, 0) {};
        
            \begin{scope}[on background layer]
                    
                \newcommand*{\pathri}{	($(x1) + (-0.2,0.4)$)
    					to[out=0,in=180] ($(x3) + (-0.1,0.4)$)
    					to[out=0,in=90] ($(x3) + (0.3,0)$)
    					to[out=-90,in=0] ($(x3) + (-0.1,-0.4)$)
    					to[out=180,in=0] ($(x1) + (-0.2,-0.4)$)
    					to[out=180,in=-90] ($(x1) + (-0.6,0)$)
    					to[out=90,in=180] ($(x1) + (-0.2,0.4)$)}

                \newcommand*{\pathrii}{	($(x1) + (0.25,0.25)$)
                    to[out=135,in=135] ($(x5) + (-0.25,-0.25)$)
                    to[out=-45,in=-45] ($(x1) + (0.25,0.25)$)}
                
                \newcommand*{\pathriii}{($(x5) + (-0.6,0)$)
        					to[out=90,in=90] ($(x6) + (0.6,0)$)
        					to[out=-90,in=-90] ($(x5) + (-0.6,0)$)}
                
                \newcommand*{\pathriv}{($(x3) + (-0.25,0)$)
                    to[out=90,in=90] ($(x7) + (0.25,0)$)
                    to[out=-90,in=-90] ($(x3) + (-0.25,0)$)}
                
                \draw[hyperedge,looseness=1,draw=Red,fill=Red,fill opacity=.1] \pathri;
                \draw[hyperedge,looseness=1,draw=Fuchsia,fill=Fuchsia,fill opacity=.1] \pathrii;
                \draw[hyperedge,looseness=1,draw=NavyBlue,fill=NavyBlue,fill opacity=.1] \pathriii;
                \draw[hyperedge,looseness=1, draw=BurntOrange,fill=BurntOrange,fill opacity=.1] \pathriv;
        
                \node[text=Red] at (0.25,-0.7) {$R_1$};
                \node[text=Fuchsia] at (-0.6,-1) {$R_2$};
                \node[text=Cyan] at (1,-2) {$R_3$};
                \node[text=BurntOrange] at (1.75,-0.7) {$R_4$};
                
            \end{scope}
        \end{tikzpicture}
        \vspace{-0.5\baselineskip}
        \caption{Hypergraph of the query}
        \label{fig:hierarchical-hg}
    \end{subfigure}
    \hfill
    \begin{subfigure}{.15\linewidth}
        \centering
        \begin{tikzpicture}[
			level distance=0.75cm,
			level 1/.style={sibling distance=1.2cm}
        ]
			\node (R1) { $R_1$ }
				child { node (R2) {$R_2$} edge from parent[-]
                    child { node (R3) {$R_3$} edge from parent[-] }
                }
                child { node (R4) {$R_4$} edge from parent[-]
                } ;

            \begin{scope}[on background layer]

            \draw[hyperedge,dashed,looseness=1,draw=Red,fill opacity=.1] {
                ($(R1) + (0,0.25)$)
                to[out=180,in=90] ($(R1) + (-0.25,0)$)
        		to[out=-90,in=180] ($(R1) + (0,-0.25)$)
                to[out=0,in=-90] ($(R1) + (0.25,0)$)
        		to[out=90,in=0] ($(R1) + (0,0.25)$)
            };
            \draw[hyperedge,dashed,looseness=1,draw=Fuchsia,fill opacity=.1] {
                ($(R2) + (0,0.25)$)
        		to[out=180,in=180] ($(R3) + (0,-0.25)$)
        		to[out=0,in=0] ($(R2) + (0,0.25)$)
            };
            \draw[hyperedge,dashed,looseness=1,draw=BurntOrange,fill opacity=.1] {
                ($(R4) + (0,0.25)$)
                to[out=180,in=90] ($(R4) + (-0.25,0)$)
        		to[out=-90,in=180] ($(R4) + (0,-0.25)$)
                to[out=0,in=-90] ($(R4) + (0.25,0)$)
        		to[out=90,in=0] ($(R4) + (0,0.25)$)
            };
    
        \end{scope}
		\end{tikzpicture}
        \caption{A join tree}
        \label{fig:hierarchical-jt}
    \end{subfigure}
    \hfill
    \begin{subfigure}{.28\linewidth}
        \centering
        \begin{tikzpicture}[
			level distance=0.6cm,
			level 1/.style={sibling distance=1.5cm},
            level 2/.style={sibling distance=0.75cm}
        ]
			\node { $\bowtie$ }
				child { node {$\bowtie$} edge from parent[-]
                    child { node (R1) {$R_1$} edge from parent[-] }
                    child { node (R4) {$R_4$} edge from parent[-] }
                }
                child { node (R2R3) {$\bowtie$} edge from parent[-]
                    child { node (R2) {$R_2$} edge from parent[-] }
                    child { node (R3) {$R_3$} edge from parent[-] }
                };
            \draw[hyperedge,dashed,looseness=1,draw=Red,fill opacity=.1] {
                ($(R1) + (0,0.25)$)
                to[out=180,in=90] ($(R1) + (-0.25,0)$)
        		to[out=-90,in=180] ($(R1) + (0,-0.25)$)
                to[out=0,in=-90] ($(R1) + (0.25,0)$)
        		to[out=90,in=0] ($(R1) + (0,0.25)$)
            };
            \draw[hyperedge,dashed,looseness=1,draw=Fuchsia,fill opacity=.1] {
                ($(R2R3) + (0,0.25)$)
        		to[out=180,in=90] ($(R2) + (-0.25,0)$)
        		to[out=-90,in=-90] ($(R3) + (0.25,0)$)
                to[out=90,in=0] ($(R2R3) + (0,0.25)$)
            };
            \draw[hyperedge,dashed,looseness=1,draw=BurntOrange,fill opacity=.1] {
                ($(R4) + (0,0.25)$)
                to[out=180,in=90] ($(R4) + (-0.25,0)$)
        		to[out=-90,in=180] ($(R4) + (0,-0.25)$)
                to[out=0,in=-90] ($(R4) + (0.25,0)$)
        		to[out=90,in=0] ($(R4) + (0,0.25)$)
            };
		\end{tikzpicture}
        \caption{A width-1 query plan, induced by the join tree in \cref{fig:hierarchical-jt}}
        \label{fig:hierarchical-plan}
    \end{subfigure}
    \hfill
    \begin{subfigure}{.25\linewidth}
        \centering
        \begin{tikzpicture}[
			level distance=0.6cm,
			level 1/.style={sibling distance=1.25cm}
        ]
			\node { $\bowtie$ }
				child { node {$\bowtie$} edge from parent[-]
    				child { node {$\bowtie$} edge from parent[-]
                        child { node {$R_1$} edge from parent[-] }
                        child { node {$R_2$} edge from parent[-] }
                    }
                    child { node {$R_3$} edge from parent[-] }
                }
                child { node {$R_4$} edge from parent[-] };
		\end{tikzpicture}
        \caption{A width-2 query plan}
        \label{fig:non-hierarchical-plan}
    \end{subfigure}
    \caption{The hypergraph, a join tree, a width-1 query plan, and a width-2 query plan of the query $Q_{\ref{ex:hierarchical}}$}
    \label{fig:hierarchical}
\end{figure}

%% file: 4_meta-representation.tex
\section{Meta-Decompositions}
\label{sec:meta}

Even though we have shown a connection between join trees and width-1 query plans in \cref{thm:w1plan-size-bound} and \cref{thm:w1plan-jt}, it is impractical to find the optimal width-1 query plan by enumerating every possible join tree---As noted earlier, a query can have an exponential number of possible join trees. Therefore, in this section, we introduce and explore a generalized representation that efficiently captures all join trees of an acyclic query. This enables us to optimize queries directly on this representation, eliminating the need to enumerate all possible join trees. 

We highlight an observation that \emph{the exponential number of possible join trees often arises when multiple relations share common join attributes}, as is the case for star queries discussed in \cref{thm:exp-join-trees} and illustrated by \cref{ex:star}.
In such a query, any pair of relations has valid join predicates, so in join trees of such queries, these relations can be connected in any arbitrary tree structure.
On the contrary, in the case of a query like in \cref{ex:hierarchical}, where the join attributes are unique for all possible pairs of relations to be joined, the query demonstrates a more ``tree-like'' structure.  There is a unique tree structure for such a query, with $n$ distinct, structurally isomorphic join trees, differing only by the choice of the root relation.

In light of this, we introduce \emph{meta-decompositions}. In the case of acyclic conjunctive queries, they can capture all possible join trees but require only linear space.
One of the major differences from standard join trees is the addition of ``minor nodes'' to handle cases when more than two relations share the same join key.
These nodes have empty $\lambda$-labels, meaning they are not associated with a single relation, but their $\chi$-labels contain the shared join keys.
Additionally, we add a $\kappa$-label to each node to denote the ``interface'' with the remaining relations of the join tree that are not in the subtree rooted at the node.
\begin{example}
    The meta-decomposition of the query $Q_{\ref{ex:star}}$ in \cref{ex:star} is shown in \cref{fig:star-meta}. The root of this tree is a minor node, showing that the four children share precisely the attribute $x_1$. Note also that the $\kappa$-label of all 4 non-root nodes is exactly $\{ x_1 \}$.
\end{example}
\begin{example}
    The meta-decomposition of the query $Q_{\ref{ex:hierarchical}}$ in \cref{ex:hierarchical} is shown in \cref{fig:hierarchical-meta}. It has the exact same structure as the join tree shown in \cref{fig:hierarchical-jt} and has no minor node, as all pairs of relations have different intersections of attributes.
\end{example}
\begin{figure*}
    \input{figs/star-meta}
    \hfill
    \input{figs/hierarchical-meta}
\end{figure*}

The formal definition is given as follows.

\begin{definition}
    Given a conjunctive query $Q$ and the associated hyergraph $H$, its \emph{meta-decomposition} $M = (V(M), r(M), E(M), \lambda, \chi, \kappa)$ is a labeled tree, where
    \begin{itemize}[leftmargin=*]
        \item $V(M)$ is the set of vertices, $r(M) \in V(M)$ is the root, and $E(M)$ is the set of edges,
        \item $\lambda : V(M) \rightarrow 2^{E(H)}$ maps each vertex of $M$ to a set of hyperedges of $H$ (possibly empty), and
        \item $\chi : V(M) \rightarrow 2^{V(H)}$ and $\kappa : V(M) \rightarrow 2^{V(H)}$ each maps each vertex of $M$ to a set of vertices of $H$,
    \end{itemize}
    such that it satisfies \ref{H1}, \ref{H2}, and
    \begin{enumerate}[label=(C\textsubscript{\arabic*}),ref=(C\textsubscript{\arabic*}),leftmargin=*]
        \setcounter{enumi}{3}
        \item[{\crtcrossreflabel{(C$_3'$)}[H3']}] For all $p \in V(M)$, $\chi(p) \subseteq V(\lambda(p))$ if $\lambda(p) \neq \emptyset$.
        \item\label{H4}\gdef\Giref{\ref{H4}} \textbf{Interface condition.} For every non-root $p \in V(M)$, let $q$ be $p$'s parent node. $\kappa(p)$ satisfies
        \begin{enumerate}
            \item $\kappa(p) = \chi(p) \cap \chi(V(M) \setminus V(M_p)) \subseteq \chi(q)$, where $M_p$ is the subtree of $M$ rooted at $p$; and
            \item $\kappa(p) \nsubseteq \chi(s) \text{ for all } s \in V(M) \setminus V(M_q)$ if $\kappa(p) \neq \chi(q)$.
        \end{enumerate}
        For the root node $r(M)$, we set $\kappa(r(M)) = \emptyset$.
        
        \item\label{H5}\gdef\Giref{\ref{H5}} \textbf{Minor nodes and uniqueness condition.} For every maximal subset of nodes $S = \{ p_1, \dots, p_n \} \subseteq V(M) $ $(n \geq 2)$ such that $ \kappa(p_1) = \dots = \kappa(p_n) = \kappa $ for some value of $\kappa$, there exists one unique node $ m \in V(M) $ with $ \lambda(m) = \emptyset $ and $ \chi(m) = \kappa$.
    \end{enumerate}
\end{definition}

In this paper, we consider acyclic queries and their ``minimal-width'' (``width-1'') meta-decompositions, which satisfy
\begin{enumerate}[label=(C\textsubscript{\arabic*}),ref=(C\textsubscript{\arabic*}),leftmargin=*]
    \setcounter{enumi}{3}
    \item[{\crtcrossreflabel{(C$_3''$)}[H3'']}] For all $p \in V(M)$ such that $\lambda(p) \neq \emptyset$, we have $|\lambda(p)| = 1$ and, w.l.o.g, $\chi(p) = V(\lambda(p))$.
\end{enumerate}

Besides introducing minor nodes in condition \ref{H5} and defining $\kappa$ as the interface in condition \ref{H4}(a), we also include condition \ref{H4}(b) that limits the shape of the meta-decomposition $M$: There can be multiple nodes that can satisfy $\kappa(p) \subseteq \chi(q)$, while condition \ref{H4}(b) limits $q$ to be the highest node with $\kappa(p) \subseteq \chi(q)$. We illustrate this with the following example:

\begin{example} \label{ex:meta-complicated}
    Consider the query
    \begin{align*}
        Q_{\ref{ex:meta-complicated}} \gets\ & R_1[x_1, x_2, x_6] \bowtie R_2[x_1, x_2, x_3, x_7] \bowtie R_3[x_1, x_3, x_4, x_8]
        \bowtie R_4[x_1, x_4, x_9] \bowtie R_5[x_1, x_{5}]
    \end{align*}
    Its hypergraph and meta-decomposition are shown in \cref{fig:complicated-hg} and \cref{fig:complicated-meta}, respectively.
    \cref{fig:complicated-jt1} and \cref{fig:complicated-jt2} show two valid join trees of this hypergraph. We highlight that $R_5$ could in fact be the child of any other node of the join tree, because its only interface with the rest of the hypergraph is $\{ x_1 \}$. For example, \cref{fig:complicated-jt3} is also a valid join tree.
    But, for the meta-decomposition in \cref{fig:complicated-meta}, according to condition \ref{H4}(b), we make $R_5$ the child of the highest possible node that contains $x_1$, which in this case is the root node, i.e., the minor node with $\chi$-label $\{ x_1, x_3 \}$.
    
    \input{figs/meta-complicated}
\end{example}

\subsection{Constructing Meta-Decompositions}

\begin{algorithm}
\caption{Constructing the meta-decomposition of an acyclic hypergraph} \label{alg:meta}
\SetKwInOut{Input}{input}
\SetKwInOut{Output}{output}
\SetKwInOut{Known}{known}
\Input{An acyclic hypergraph $H$}
\Output{The meta-decomposition $M$ of $H$}
$H' \gets $ a copy of the input hypergraph $H$ with $V(H') = V(H)$ and $E(H') = E(H)$ \;
$V(M) \gets \emptyset$, $E(M) \gets \emptyset$ \;

\tcp*[l]{Reduction, creating the vertices in the meta-decomposition at the same time}
\While{$E(H')$ is not empty \label{alg:meta-while}} {
    \tcp*[l]{First check if it is possible to create minor nodes for mutually reducible hyperedges}
  \ForEach{maximal $S \subseteq E(H')$ such that $|S| > 1$, and, for all $e \in S$, $o(e, H')$ has the same value $o$ \label{alg-line:meta-minor-origins-sets}}{
    \ForEach{$e \in S$}{
        Add a node $p$ to $V(M)$ with $\lambda(p) = \set{e}$ if $e \in E(H)$, or $\lambda(p) = \emptyset$ if it is special, $\chi(p) = e$, and $\kappa(p) = o(e, H')$ \label{alg-line:meta-origin-node}, if such $p$ does not exist yet \;
        Remove $e$ from $E(H')$ \label{alg-line:meta-origin-remove} \;
    }
    Add a special hyperedge $e_m = o$ to $E(H')$ \label{alg-line:meta-add-special-edge} \;
  }
  $\mathcal{E} \gets$ the set of all ears in $E(H')$ \;

  \ForEach{$e \in \mathcal{E}$}{
    Add a new node $p$ to $V(M)$ with $\lambda(p) = \set{e}$ if $e \in E(H)$, or $\lambda(p) = \emptyset$ if it is special, $\chi(p) = e$, and $\kappa(p) = o(e, H')$, if such $p$ does not exist yet \label{alg-line:meta-ear-node} \;
    Remove $e$ from $E(H')$ \label{alg-line:meta-ear-remove} \;
    \If {there exists at least two vertices $p' \in V(M) \cup \mathcal{E}$ such that $\kappa(p') = \kappa(p)$ \label{alg-line:no-minor-check}} {
        Add a (minor) node $v$ to $V(M)$ with $\lambda(v) = \emptyset$, $\chi(v) = \kappa(p)$, and $\kappa(v) = \kappa(p)$, if such $v$ does not exist yet \label{alg-line:no-minor-create}\;
    }
  }

}
$r(M) \gets $ the node $p \in V(M)$ such that $\kappa(p) = \emptyset$ \;
\tcp*[l]{Generate the edges in the meta-decomposition}
\ForEach {non-root node $p \in V(M)$} {
    \If {there exists a (minor) node $q \in V(M)$ such that $\lambda(q) = \emptyset$ and $\chi(q) = \kappa(p)$ \label{alg-line:minor-origin-edge-cond}} {
        Add an edge $(q, p)$ in $E(M)$ \;
    } \Else( ) {
        Find the node $q \in V(M)$ such that $\kappa(p) \subseteq \chi(q)$ and $\kappa(p) \not\subseteq \kappa(q)$ \label{alg-line:phys-find-parent} \tcp*[r]{``highest'' possible parent}
        Add an edge $(q, p)$ in $E(M)$ \;
    }
}
\Return $M = (V(M), r(M), E(M), \lambda, \chi, \kappa)$ \;
\end{algorithm}

Similar to the original GYO algorithm~\cite{grahamGYO,DBLP:conf/compsac/YuO79}, a meta-decomposition can be constructed by a reduction-based algorithm, as shown in \cref{alg:meta}. In this algorithm, $o(e, H) = e \cap ( \cup (E(H) \setminus \set{e}) )$ denotes the set of overlapping vertices of hyperedge $e$ and the remaining hyperedges $E(H) \setminus \set{e}$.  As in the original GYO algorithm~\cite{grahamGYO,DBLP:conf/compsac/YuO79}, a hyperedge $e$ is \emph{reducible} (also called an \emph{ear}), if there exists another single hyperedge $e' \in E(H) \setminus \{ e \}$ such that $o(e, H) \subseteq e'$.  
The subsuming hyperedge $e'$ is called the \emph{witness}.
And, exceptionally, if there is only one hyperedge $e$ in the hypergraph, it is also considered an ear.
We note that the GYO algorithm removes exactly one ear (among possibly many) in each round, thereby establishing a total order of all hyperedge removal.
In contrast, Algorithm~\ref{alg:meta} \emph{simultaneously removes all ears in every round} until only a single hyperedge remains in the hypergraph. This process introduces only a partial order for all hyperedges, which \emph{maintains the flexibility to represent all possible ordering of ear removal, and hence all possible join trees}, through the resulting meta-decomposition.

\begin{example}
    For the star query $Q_{\ref{ex:star}}$ in \cref{ex:star}, the four hyperedges can be removed in any order in the standard GYO algorithm. For example:
    \begin{enumerate}[leftmargin=*]
        \item With the order $R_4$ (witness $R_3$) -- $R_3$ (witness $R_2$) -- $R_2$ (witness $R_1$) -- $R_1$, we obtain the first join tree in \cref{fig:star-jts}.
        \item With the order $R_4$ (witness $R_3$) -- $R_2$ (witness $R_1$) -- $R_3$ (witness $R_1$) -- $R_1$, we obtain the second join tree in \cref{fig:star-jts}.
    \end{enumerate}
    Furthermore, we note that we can also obtain the second join tree by keeping the order in (1) but only changing witnesses, namely $R_4$ (witness $R_3$) -- $R_3$ (witness $R_1$) -- $R_2$ (witness $R_1$) -- $R_1$.
\end{example}

To address this source of variation, in every round, our algorithm first handles such cases where some edges $ e_1, \dots, e_n$ are \emph{mutually reducible}, i.e., $o(e_1, H) = \dots = o(e_n, H) = o$ (Lines \ref{alg-line:meta-minor-origins-sets}--\ref{alg-line:meta-add-special-edge}).  
(For example, in the case of $Q_{\ref{ex:star}}$, all pairs among $R_1$, $R_2$, $R_3$, and $R_4$ are mutually reducible.)
The algorithm creates a minor node $p$ with $\chi(p) = o$ if such a node does not yet exist, and then removes $e_1, \dots, e_n$ from the hypergraph.  We then add a special hyperedge $e_m = o$, which conceptually corresponds to the minor node, back to the hypergraph, so that the acyclicity of $H$ is preserved. The rationale for adding such an $e_m$ can be illustrated by the following example:

\begin{example}
    \label{ex:special-hyperedge}
    Consider the query
    \begin{align*}
        Q_{\ref{ex:special-hyperedge}} \gets \pi_{\emptyset} (& R_1[x_1, x_2, x_3, x_4] \bowtie R_2[x_1, x_2, x_5] \bowtie R_3[x_1, x_3, x_6] \bowtie R_4[x_2, x_3, x_7] \bowtie R_5[x_1, x_2, x_3, x_8]),
    \end{align*}
    which is similar to \cref{ex:123-12-13-23-intro} but contains an additional relation $R_5$. In this case, at the beginning, $R_1$ and $R_5$ are both ears and mutually reducible, as $o(R_1, H) = o(R_5, H) = \{ x_1, x_2, x_3 \}$. However, if we simply remove both $R_1$ and $R_5$ at the same time, \emph{the remaining hyperedges, $R_2$, $R_3$, and $R_4$, form a cyclic hypergraph}. But, if we add $e_m = o(R_1, H) = o(R_5, H) = \{ x_1, x_2, x_3 \}$, then $e_m$, $R_2$, $R_3$, and $R_4$ still form an acyclic query, and the algorithm can proceed to pick $R_2$, $R_3$, and $R_4$ as ears, before $e_m$ is picked and removed at the very end.
\end{example}

We formalize this idea as follows:

\begin{proposition} \label{lem:add-back-minor}
    Given an acyclic hypergraph $H$, if there exists a set of distinct ears $S = \{ e_1, \dots, e_n \} \subseteq E(H)$ $(n \geq 2)$ such that $o(e_1, H) = \dots = o(e_n, H) = o$, then, the hypergraph $H'$, with $E(H') = (E(H) \setminus S) \cup \{ o \}$ and $V(H') = \cup E(H')$, is acyclic.
\end{proposition}

In each round, the algorithm then finds and removes the set $\mathcal{E}$ of all ears of $H'$.
For each $e \in \mathcal{E}$, we create a node $p$ with set $\kappa(p) = o(e, H')$.  
$p$ may be either physical or minor, depending on whether $e$ is a hyperedge in $H$ or a special hyperedge created while removing mutually reducible hyperedges.
During the procedure, the algorithm creates minor nodes whenever necessary (Lines~\ref{alg-line:meta-origin-node}, \ref{alg-line:meta-ear-node}, and \ref{alg-line:no-minor-check}--\ref{alg-line:no-minor-create}).

This algorithm has the following guarantee of asymptotic complexity:

\begin{theorem} \label{thm:meta-complexity}
    \cref{alg:meta} terminates in time $O(|E(H)|^3)$.
\end{theorem}

\begin{sketch}
    The key observations are (1) that the size $|E(H')|$ decreases by at least one in each iteration of the while-loop (\cref{alg:meta-while}), and (2) that, in each iteration of the while loop, all sets $S$ (\cref{alg-line:meta-minor-origins-sets}) can be enumerated in $O(|E(H')|) = O(|E(H)|)$ time, if we maintain a hash map that maps each $o \subseteq V(H')$ to all hyperedges $e \in E(H')$ such that $o(e, H') = o$, which is updated whenever a hyperedge $e$ is added or removed from $E(H')$.
\end{sketch}

We show the correctness of the algorithm:

\begin{theorem} \label{thm:meta-correct}
    Given an acyclic query with an acyclic hypergraph $H$, \cref{alg:meta} returns a valid meta-decomposition of $H$.  
\end{theorem}

Finally, we note that meta-decompositions are highly space-efficient. They only require space linear to the size of the hypergraph.

\begin{theorem} \label{thm:meta-size}
    In the meta-decomposition $M$ of an acyclic hypergraph $H$ as constructed by \cref{alg:meta}, the number of vertices $|V(M)|$ and the number of edges $|E(M)|$ are $O(|E(H)|)$.
\end{theorem}

\begin{sketch}
    Vertices in $M$ are either physical or minor. We create one physical node per hyperedge in $E(H)$. For minor nodes, we note that each minor node has fan-out at least 2, meaning there must be at least two edges incident to each minor node.
\end{sketch}

\subsection{Join Tree Enumeration}
\label{sec:enum}

The meta-decomposition allows one to efficiently enumerate all possible join trees of an acyclic query, as we show with an algorithm in this section. This not only serves as proof of the universality of this representation but also provides a powerful approach to enumerating join trees, which could be of great help to other structure-based methods that require doing so.

As an intuition, we can observe that the distinct join trees of the same hypergraph may differ in the following ways:
\begin{enumerate}[leftmargin=*]
    \item \textbf{``Re-branching''.} If a vertex $p$ in a join tree has children $c_1$ and $c_2$ such that
    $\chi(c_1) \cap \chi(p) \subseteq \chi(c_2)$,
    then $c_1$ could have been a child of $c_2$ without violating any condition. In the meta-decomposition, condition \ref{H4} ensures that $c_1$ is positioned at the highest possible node. When enumerating join trees, we can make $c_1$ a child of $c_2$ if $\kappa(c_1) \subseteq \chi(c_2)$, without violating any condition. The join tree in \cref{fig:complicated-jt3} based on the meta-decomposition in \cref{fig:complicated-meta} from \cref{ex:meta-complicated} demonstrates such a case, where $R_5$ could be a child of any one of $R_1$, $R_2$, $R_3$, and $R_4$.
    \item \textbf{Rerooting.} A rerooting of a join tree is still a valid join tree.
    \item \textbf{Mutually reducible relations.} If multiple relations share some exact same set of attributes, they can be connected in arbitrary tree structures. This situation is represented by minor nodes in meta-decompositions. An example is the root node with $\chi$-label $\{ x_1 \}$ in \cref{fig:star-meta} for \cref{ex:star}, based upon which we should be able to enumerate, among others, the join trees in \cref{fig:star-jts}.
\end{enumerate}

These alternatives should be considered to enumerate all join trees, and the meta-decomposition retains all the necessary information to facilitate this. \cref{alg:enum} is an algorithm to enumerate all join trees. For each node in the meta-decomposition starting from the root, it recursively enumerates all possible join trees given by the subtree rooted at each of its children (Lines \ref{alg-line:enum-for-each-child}--\ref{alg-line:enum-subtrees}), and then enumerates all possible ways to combine these join trees into one. For minor nodes $r$, we find its neighbors that share the common set of attributes given by $\chi(r)$ (\cref{alg-line:enum-origin-set}) and generate all possible distinct structures of join trees, where each vertex will be replaced by the join tree for a subtree given by a neighbor. (At this point, we consider all rerootings of the same tree as identical.) At the very end of this algorithm, we return all combined join trees we obtain, and it remains only to collect all rerootings.  Due to space constraints, we omit some lower-level details, e.g., the exact steps for combining sub-join trees and algorithms to enumerate all trees and rerootings without repetition. The full details are given in the \arxivorappendix.

\begin{algorithm}
\small
\caption{Enumerating all join trees based on a meta-decomposition} \label{alg:enum}
\SetInd{0.5em}{0.6em}
\SetKwInOut{Input}{input}
\SetKwInOut{Output}{output}
\SetKwInOut{Known}{known}
\SetKwProg{Fn}{Function}{:}{}
\SetKwFunction{enumRec}{$\mathsf{enumRec}$}
\Input{A meta-decomposition $M = (V(M), r(M), E(M), \lambda, \chi, \kappa)$}
\Output{All join trees of the hypergraph induced by $M$}

\Fn{\enumRec{$v$: a node on $M$}}{

$\mathcal{C} \gets $ set of children $c$ of $v$ on $M$, sorted by partial order $\supseteq$ on $\kappa(c)$ \;
\ForEach{$c \in \mathcal{C}$ \label{alg-line:enum-for-each-child}} {
    $\mathcal{T}_c \gets$ all rerootings $T_c'$ of all trees $T_c$ returned by $\mathsf{enumRec}(c)$ such that $\kappa(c) \subseteq \chi(r(T_c'))$ \label{alg-line:enum-subtrees} \;
}

\If(\tcp*[f]{minor node}){$\lambda(v) = \emptyset$ \label{alg-line:enum-minor-cond}} {
    $\mathcal{C}_{\sf origin} \gets \{ c \in \mathcal{C} : \kappa(c) = \chi(v) \}$ \label{alg-line:enum-origin-set} \tcp*[r]{They can be found at the head of the sorted $\mathcal{C}$}
    \If{$\kappa(v) = \chi(v)$}{
        $\mathcal{T} \gets $ all non-isomorphic trees with vertices $\mathcal{C}_{\sf origin} \cup \{ v \}$, all rerooted to $v$
        \label{alg-line:enum-minor-dummy} \;
        \tcp*[r]{The minor node $v$ will eventually be replaced on \cref{alg-line:enum-simplified-dummy-replace} by the parent of $r$}
    } \Else(\tcp*[f]{$\kappa(v) \neq \chi(v)$, the parent should not participate}) {
        $\mathcal{T} \gets $ all non-isomorphic trees with vertices $\mathcal{C}_{\sf origin}$ \label{alg-line:enum-minor-no-dummy} \;
    }
            $\mathcal{T} \gets $ set of trees $T'$ that can be obtained from some $T \in \mathcal{T}$, where each $c \in V(T)$, in bottom-up order, is replaced by some $T_c \in \mathcal{T}_c$, and, for each non-root $c \in V(T)$ with parent $p$, $T_c$ could be attached as a child of any $u \in V(T_p)$ such that $\kappa(c) \subseteq \chi(u)$ \;
    $\mathcal{C}_{\sf proper} \gets \mathcal{C} \setminus \mathcal{C}_{\sf origin}$, preserving the order \tcp*[r]{The set of children $c$ such that $\kappa(c) \subsetneq \chi(v)$}
} \Else {
    $\mathcal{T} \gets $ a set of one tree with one vertex $v$ only \label{alg-line:enum-physical-root} \;
    $\mathcal{C}_{\sf proper} \gets \mathcal{C}$, preserving the order \;
}
\ForEach {child $c \in \mathcal{C}_{\sf proper}$ in order,
    $T_c \in \mathcal{T}_c$} {
        \If{$\lambda(c) = \emptyset$ and $\kappa(c) = \chi(c)$}{
            \tcp*[r]{The root of $T_c$ would be a minor node from \cref{alg-line:enum-minor-dummy}. We find a parent of each of its children.}
            $\mathcal{T} \gets $ set of trees $T'$ that can be obtained from some $T \in \mathcal{T}$, where for each child $d$ of $r(T_c)$, we attach the subtree rooted at $d$ as a child of some $u \in V(T)$ such that $\kappa(d) \subseteq \chi(u)$ \label{alg-line:enum-simplified-dummy-replace} \;
            
        } \Else {
            $\mathcal{T} \gets $ set of trees $T'$ that can be obtained from some $T \in \mathcal{T}$, where $T_c$ is attached as a child of some $u \in V(T)$ such that $\kappa(c) \subseteq \chi(u)$ \;
        }
}

\Return{$\mathcal{T}$} \;
    
}

\Return{all rerootings of all trees returned by $\mathsf{enumRec}(r(M))$ } \;
\end{algorithm}

\begin{example}
    Consider again the query $Q_{\ref{ex:meta-complicated}}$ in \cref{ex:meta-complicated}. Its meta-decomposition is shown in \cref{fig:complicated-meta}. Let the root node be $r$, and we use $v_1, \dots v_5$ to denote the nodes such that $\lambda(v_i) = \{ R_i \}$.
    The root node $r$ is a minor node, with $\chi (r) = \{ x_1, x_3 \}$.
    We start by enumerating all possible join trees induced by the subtrees of the meta-decomposition, rooted at each child of $r$.
    For the subtree with $v_2$ and $v_1$, there is only one valid join tree, rooted at $v_2$, which has one child $v_1$. It is not possible to reroot to $v_1$ because $\kappa(v_2) = \{ x_1, x_3 \} \not\subseteq \chi(v_1)$. The case is similar for the subtree with $v_3$ and $v_4$. And the subtree with $v_5$ induces only one join tree with one node $v_5$.
    Since $v_2$ and $v_3$ (and not $v_5$) have $\kappa(v_2) = \kappa(v_3) = \chi(r) = \{ x_1, x_3 \}$, we enumerate all possible skeleton tree structures with vertices $v_2$ and $v_3$.
    There are only two such trees, each of which contains one edge connecting $v_2$ and $v_3$, and they are rerootings of each other.
    For each tree, we replace $v_2$ and $v_3$ by their induced join tree.
    In the end, we have one child $v_5$ with $\kappa(v_5) \subsetneq \chi(r)$.
    It can be a child of any node $v$ such that $\kappa(v_5) \subseteq \chi(v)$. In this example, it can, in fact, be a child of any other node.
    \cref{fig:complicated-jt1} is a join tree where the skeleton tree structure is rooted at $R_2$, and $R_5$ is attached as a child of $R_2$. \cref{fig:complicated-jt2} is a similar example where the skeleton is rooted at $R_3$, and $R_5$ is attached as a child of $R_3$. And \cref{fig:complicated-jt3} is based on the same skeleton as for \cref{fig:complicated-jt2}, but $R_5$ is attached as a child of $R_4$ instead of $R_3$.
\end{example}

We claim that this algorithm is both correct (producing valid join trees) and complete (enumerating all possible join trees).

\begin{theorem}[Correctness]
\label{the:correctness}
    Given a meta-decomposition $M$ of an acyclic hypergraph $H$, all trees $T$ output by \cref{alg:enum} are valid join trees of $H$.
\end{theorem}

\begin{theorem}[Completeness] \label{thm:enum-complete}
    Given a meta-decomposition $M$ of an acyclic hypergraph $H$, \cref{alg:enum} enumerates all possible join trees of $H$.
\end{theorem}

As is detailed in the proof of \cref{thm:enum-complete}, for each join tree, we can pinpoint exactly one way in which it is enumerated by the algorithm based on the meta-decomposition. Noting that the number of children of each node on a meta-decomposition is bounded by the fan-out, we get the following guarantee of enumeration delay.

\begin{theorem}[Time complexity]
    Given a meta-decomposition $M$ of an acyclic hypergraph $H$, \cref{alg:enum} enumerates all possible join trees of $H$ in an amortized delay of $O(f(M) \cdot|V(H)|)$.
\end{theorem}

%% file: figs/star-meta.tex
\begin{minipage}[b]{0.48\linewidth}
    \centering
    \small
    \begin{tikzpicture}[
        level distance=1.25cm,
        level 1/.style={sibling distance=1.75cm},
    ]
    \node[align=center] { $\lambda: \emptyset$, $\chi: \{ x_1 \}$, $\kappa: \emptyset$} 
        child { node [align=center] {$\lambda: \{ R_1 \}$ \\ $\chi: \{ x_1, x_2 \}$ \\ $\kappa: \{ x_1 \}$ } edge from parent[-] }
        child { node [align=center] { $\lambda: \{ R_2 \}$ \\ $\chi: \{ x_1, x_3 \}$ \\ $\kappa: \{ x_1 \}$ } edge from parent[-] }
        child { node [align=center] { $\lambda: \{ R_3 \}$ \\ $\chi: \{ x_1, x_4 \}$ \\ $\kappa: \{ x_1 \}$ } edge from parent[-] }
        child { node [align=center] { $\lambda: \{ R_4 \}$ \\ $\chi: \{ x_1, x_5 \}$ \\ $\kappa: \{ x_1 \}$ } edge from parent[-] }
    ;
    \end{tikzpicture}
    \captionof{figure}{Meta-decomposition of the query $Q_{\ref{ex:star}}$}
    \label{fig:star-meta}
\end{minipage}

%% file: figs/hierarchical-meta.tex
\begin{minipage}[b]{0.48\linewidth}
    \centering
    \small
    \begin{tikzpicture}[
			level distance=0.7cm,
			level 1/.style={sibling distance=3.5cm}
        ]
        \node[align=center] { $\lambda: \{ R_1 \}$, $\chi: \{ x_1, x_2, x_3 \}$, $\kappa: \emptyset$ }
            child[align=center] { node {$\lambda: \{ R_2 \}$, $\chi: \{ x_1, x_4, x_5 \}$, $\kappa: \{ x_1 \}$} edge from parent[-]
                child[align=center] { node {$\lambda: \{ R_3 \}$, $\chi: \{ x_5, x_6 \}$, $\kappa: \{ x_5 \}$} edge from parent[-] }
            }
            child[align=center,level distance=1cm] { node {$\lambda: \{ R_4 \} $,\\$\chi: \{ x_3, x_7 \}$,\\$\kappa: \{ x_3 \}$} edge from parent[-]
            } ;
    \end{tikzpicture}
    \captionof{figure}{Meta-decomposition of the query $Q_{\ref{ex:hierarchical}}$}
    \label{fig:hierarchical-meta}
\end{minipage}

%% file: figs/meta-complicated.tex
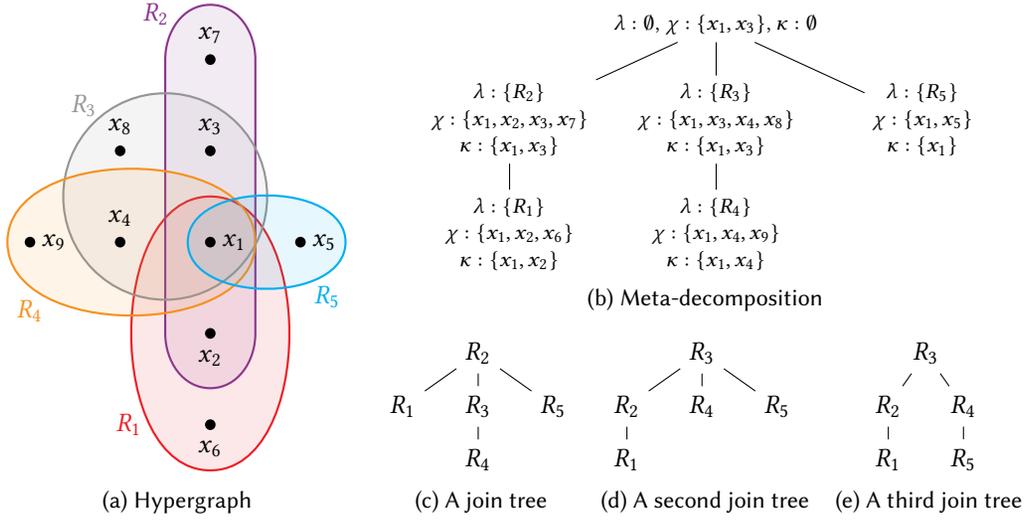
\begin{figure}
    \centering
    \begin{minipage}[b]{0.36\linewidth}
    \begin{subfigure}{\linewidth}
        \centering
        \begin{tikzpicture}[hypergraph,xscale=1.2,yscale=1.2]
            \node[vertex,label={east:$x_1$}] (x1) at (0, 0) {};
            \node[vertex,label={south:$x_2$}] (x2) at (0, -1) {};
            \node[vertex,label={north:$x_3$}] (x3) at (0, 1) {};
            \node[vertex,label={north:$x_4$}] (x4) at (-1, 0) {};
            \node[vertex,label={east:$x_5$}] (x5) at (1, 0) {};

            \node[vertex,label={south:$x_6$}] (x6) at (0, -2) {};
            \node[vertex,label={north:$x_7$}] (x7) at (0, 2) {};
            \node[vertex,label={north:$x_8$}] (x8) at (-1, 1) {};
            \node[vertex,label={east:$x_9$}] (x9) at (-2, 0) {};
        
            \begin{scope}[on background layer]
                
                \newcommand*{\pathri}{($(x1) + (0,0.5)$)
                    to[out=0,in=0] ($(x6) + (0,-0.5)$)
                    to[out=180,in=180] ($(x1) + (0,0.5)$)
                }
                    
                \newcommand*{\pathrii}{($(x7) + (0,0.6)$)
                    to[out=0,in=90] ($(x7) + (0.5,0)$)
                    to[out=-90,in=90] ($(x2) + (0.5,0)$)
                    to[out=-90,in=0] ($(x2) + (0,-0.6)$)
                    to[out=180,in=-90] ($(x2) + (-0.5,0)$)
                    to[out=90,in=-90] ($(x7) + (-0.5,0)$)
                    to[out=90,in=180] ($(x7) + (0,0.6)$)
                }
                
                \newcommand*{\pathriii}{($(x1) + (0.3,-0.3)$)
                    to[out=45,in=-45] ($(x3) + (0.3,0.3)$)
                    to[out=135,in=45] ($(x8) + (-0.3,0.3)$)
                    to[out=-135,in=135] ($(x4) + (-0.3,-0.3)$)
                    to[out=-45,in=-135] ($(x1) + (0.3,-0.3)$)
                }

                \newcommand*{\pathriv}{($(x9) + (-0.25,0)$)
                    to[out=90,in=90] ($(x1) + (0.5,0)$)
                    to[out=-90,in=-90] ($(x9) + (-0.25,0)$)
                }

                \newcommand*{\pathrv}{($(x1) + (-0.25,0)$)
                    to[out=90,in=90] ($(x5) + (0.5,0)$)
                    to[out=-90,in=-90] ($(x1) + (-0.25,0)$)
                }
                
                \draw[hyperedge,looseness=1,draw=Red,fill=Red,fill opacity=.1] \pathri;
                \draw[hyperedge,looseness=1,draw=Fuchsia,fill=Fuchsia,fill opacity=.1] \pathrii;
                \draw[hyperedge,looseness=1,draw=Gray,fill=Gray,fill opacity=.1] \pathriii;
                \draw[hyperedge,looseness=1, draw=BurntOrange,fill=BurntOrange,fill opacity=.1] \pathriv;
                \draw[hyperedge,looseness=1, draw=Cyan,fill=Cyan,fill opacity=.1] \pathrv;
        
                \node[text=Red] at (-0.9,-2) {$R_1$};
                \node[text=Fuchsia] at (-0.6,2.5) {$R_2$};
                \node[text=Gray] at (-1.4,1.5) {$R_3$};
                \node[text=BurntOrange] at (-2,-0.75) {$R_4$};
                \node[text=Cyan] at (1.3,-0.6) {$R_5$};
                
            \end{scope}
        \end{tikzpicture}
        \caption{Hypergraph}
        \label{fig:complicated-hg}
    \end{subfigure}
    \end{minipage}
    \hfill
    \begin{minipage}[b]{.63\linewidth}
    \begin{subfigure}{\linewidth}
        \centering
        \begin{tikzpicture}[
			level 1/.style={level distance=1.25cm, sibling distance=2.75cm},
            level 2/.style={level distance=1.5cm, sibling distance=1.4cm},
            xscale=1,yscale=1
        ]
        \footnotesize
			\node { $\lambda: \emptyset$, $\chi: \{ x_1, x_3 \}$, $\kappa: \emptyset$ }
				child { node[align=center] {$\lambda: \{R_2\}$ \\ $\chi: \{ x_1, x_2, x_3, x_7 \}$ \\ $\kappa: \{ x_1, x_3 \}$} edge from parent[-]
                    child { node[align=center] {$\lambda: \{ R_1 \}$ \\ $\chi: \{ x_1, x_2, x_6 \}$ \\ $\kappa: \{ x_1, x_2 \}$} edge from parent[-] }
                }
                child { node[align=center] {$\lambda: \{ R_3 \}$ \\ $\chi: \{ x_1, x_3, x_4, x_8 \}$ \\ $\kappa : \{ x_1, x_3 \}$} edge from parent[-]
                        child { node[align=center] {$\lambda: \{ R_4 \}$ \\ $\chi: \{ x_1, x_4, x_9 \}$ \\ $\kappa: \{ x_1, x_4 \}$} edge from parent[-] }
                }
                    child { node[align=center] {$\lambda: \{ R_5 \}$ \\ $\chi: \{ x_1, x_5 \}$ \\ $\kappa : \{ x_1 \}$} edge from parent[-] }
                ;
		\end{tikzpicture}
        \vspace{-0.25\baselineskip}
        \caption{Meta-decomposition}
        \label{fig:complicated-meta}
    \end{subfigure}
    
    \vspace{0.5\baselineskip}
    
    \begin{subfigure}{.32\linewidth}
        \centering
        \begin{tikzpicture}[
			level distance=0.7cm,
			level 1/.style={sibling distance=1cm}
        ]
			\node { $R_2$ }
				child { node {$R_1$} edge from parent[-] }
                child { node {$R_3$} edge from parent[-]
                    child { node {$R_4$} edge from parent[-] }
                }
                child { node {$R_5$} edge from parent[-] } ;
		\end{tikzpicture}
        \vspace{-0.25\baselineskip}
        \caption{A join tree}
        \label{fig:complicated-jt1}
    \end{subfigure}
    \hfill
    \begin{subfigure}{.32\linewidth}
        \centering
        \begin{tikzpicture}[
			level distance=0.7cm,
			level 1/.style={sibling distance=1cm}
        ]
			\node { $R_3$ }
                child { node {$R_2$} edge from parent[-]
                    child { node {$R_1$} edge from parent[-] }
                }
				child { node {$R_4$} edge from parent[-] }
                child { node {$R_5$} edge from parent[-] } ;
		\end{tikzpicture}
        \vspace{-0.25\baselineskip}
        \caption{A second join tree}
        \label{fig:complicated-jt2}
    \end{subfigure}
    \hfill
    \begin{subfigure}{.32\linewidth}
        \centering
        \begin{tikzpicture}[
			level distance=0.7cm,
			level 1/.style={sibling distance=1cm}
        ]
			\node { $R_3$ }
                child { node {$R_2$} edge from parent[-]
                    child { node {$R_1$} edge from parent[-]
                    }
                }
				child { node {$R_4$} edge from parent[-]
                    child { node {$R_5$} edge from parent[-] }
                };
		\end{tikzpicture}
        \vspace{-0.25\baselineskip}
        \caption{A third join tree}
        \label{fig:complicated-jt3}
    \end{subfigure}
    \end{minipage}
    \caption{Hypergraph, meta-decomposition, and three join trees of the query $Q_{\ref{ex:meta-complicated}}$ in \cref{ex:meta-complicated}}
    \label{fig:complicated}
\end{figure}

%% file: 5_cost.tex
\section{Cost-Based Optimization from Meta-Decompositions}
\label{sec:opt}

Even though we have presented an algorithm, enumerating all possible join trees is infeasible, as their number grows exponentially with the number of relations, as shown by \cref{thm:exp-join-trees}. Instead, we propose a cost-based optimization algorithm that \emph{operates directly on the meta-decomposition} to find an optimal width-1 query plan.  
Similar to the case for join tree enumeration discussed in \cref{sec:enum}, the search space for such plans is defined by four key factors: (1)~selecting a root of the join tree to induce the width-1 query plan; (2) unnesting minor nodes; (3) handling possible re-branching, and (4) determining the local join order for each node and its neighbors.  In this section, we detail our procedure to address them.  

\subsection{Re-Branching}

As previously discussed, re-branching can happen when a child node $c$ of a node $p$ in the meta-decomposition could also be a child of another descendant node $d$ of $p$ in a join tree, if $\kappa(c) \subsetneq \kappa(d)$.
This scenario, however, is in fact very rare in practice: we found no query in the JOB benchmark that exhibited this scenario, and only one query in the TPC-H benchmark can involve re-branching.
Therefore, for simplicity, we do not consider re-branching when describing the algorithm in this section. 
As we show in the \arxivorappendix, accommodating a possible re-branching adds only a constant factor to the overall complexity.

\begin{algorithm}[t]
\caption{Finding the optimal width-1 query plan to evaluate the query induced by a meta-decomposition}
\label{alg:cost}
\SetKwInOut{Input}{input}
\SetKwInOut{Output}{output}
\SetKwInOut{Known}{known}
\Input{A meta-decomposition $M = (V(M), r(M), E(M), \lambda, \chi, \kappa)$}
\Output{The optimal query plan}

\tcp*[l]{Bottom-up traversal}

\ForEach{non-root node $q \in V(M)$, in bottom-up order}{
    $p \gets$ $q$'s parent node on $M$\;
    $ {\sf plan}(T_{p \to q}) \gets {\sf optimizeLocal}(q, \set{{\sf plan}(T_{q \to c}) : c \in \mathsf{children}(q)})$ 
}

${\sf plan}(M) \gets {\sf optimizeLocal}(r, \set{ {\sf plan}(T_{r \to c}) : c \in \mathsf{children}(r) })$ \;

\tcp*[l]{Top-down traversal}

\ForEach{$c \in \mathsf{children}(r)$}{
    ${\sf plan}(T_{c \to r}) \gets {\sf optimizeLocal}(r, \set{{\sf plan}(T_{r \to c'}) : c' \in \mathsf{children}(r)\setminus \set{c}})$ \;
}

\ForEach{non-root node $q \in V(M)$, in top-down order}{
    $p \gets q$'s parent node on $M$\;
    \ForEach{child $c \in C(q)$} {
        ${\sf plan}(T_{c \to q}) \gets {\sf optimizeLocal}(q, \set{{\sf plan}(T_{q \to p})} \cup \set{{\sf plan}(T_{q \to c'}) : c' \in \mathsf{children}(r)\setminus \set{c}})$ \;
    }
}

\Return{the plan $
{\sf plan}(T_{p \to q}) \bowtie {\sf plan}(T_{q \to p})
$ with the minimum cost among all $(p, q) \in E(M)$ }

\end{algorithm}

\subsection{The Overall Algorithm}

Given a meta-decomposition $M$, our algorithm to produce a width-1 query plan is shown in \cref{alg:cost}. This algorithm first makes two traversals of $M$, one bottom-up and one top-down.

In the bottom-up traversal, for each non-root node $q$, we find the optimal plan to join the relation in $q$ and all the plans given by the children of $q$, which gives the optimal plan to evaluate $Q(T_q)$, the query induced by the subtree of $M$ rooted at $q$.

\begin{example}
    For the query $Q_{\ref{ex:hierarchical}}$ in \cref{ex:hierarchical}, with meta-decomposition shown in \cref{fig:hierarchical-meta}, this bottom-up traversal visits the nodes in the following order. 
    (For simplicity, we simply use $R_i$ to refer to the node with $\lambda$-label $\{ R_i \}$.)
    \begin{enumerate}[leftmargin=*]
        \item $R_3$. It generates the plan for the subtree $T_{R_3}$, or $T_{R_2 \rightarrow R_3}$, containing $R_3$ only. The plan is simply a scan of $R_3$.
        \item $R_2$. It generates the plan for $T_{R_2}$, or $T_{R_1 \rightarrow R_2}$, containing $R_2$ and $R_3$. It simply needs to join $R_2$ with the plan for $T_{R_3}$ in the previous step, i.e., $R_2 \bowtie R_3$.
        \item $R_4$. It generates the plan for the subtree $T_{R_4}$, or $T_{R_1 \rightarrow R_4}$, containing $R_4$ only. The plan is simply a scan of $R_4$.
        \item $R_1$. This is the root node. The traversal terminates here.
    \end{enumerate}
\end{example}

Then, in the top-down traversal, for each node $q$ and each child $c$ of $q$ on the meta-decomposition, we find the optimal plan to join $q$ and all plans given by the neighbors (parent and children) of $q$ except $c$, which gives the optimal plan to join all relations in $T_{c \to q}$.

\begin{example}
    Again for $Q_{\ref{ex:hierarchical}}$, the top-down traversal visits the nodes in the following order.
    \begin{enumerate}[leftmargin=*]
        \item $R_1 \to$ child $R_2$. It generates the plan for the subtree $T_{R_2 \rightarrow R_1}$, containing $R_1$ and $R_4$. The plan is simply $R_1 \bowtie R_4$.
        \item $R_2 \to$ child $R_3$. It generates the plan for the subtree $T_{R_3 \rightarrow R_2}$, containing $R_1$, $R_2$, and $R_4$. Similarly, we join $R_2$ with the result of $R_1 \bowtie R_4$, which is the plan for the subtree $T_{R_2 \rightarrow R_1}$.
        \item $R_1 \to$ child $R_4$. It generates the plan for the subtree $T_{R_4 \rightarrow R_1}$, containing $R_1$, $R_2$, and $R_3$. Here, $R_2$ and $R_3$ are in the same subtree $T_{R_1 \rightarrow R_2}$ of $R_1$. So, we first take the plan generated for this subtree, which is $R_2 \bowtie R_3$, and then join this result with $R_1$.
    \end{enumerate}
\end{example}

At the end, we choose the root of the query plan, which is the minimum-cost plan of the form ${\sf plan}(T_{p \to q}) \bowtie {\sf plan}(T_{q \to p})$ among all $(p, q) \in E(M)$. For example, the width-1 plan shown in \cref{fig:hierarchical-plan} corresponds to ${\sf plan}(T_{R_2 \to R_1}) \bowtie {\sf plan}(T_{R_1 \to R_2})$.

In the example we just considered, the fan-out is 2. For higher fan-outs, at each node in the meta-decomposition, we need to decide on an order to join the results of multiple subtrees. To do so, the algorithm calls the function $\sf optimizeLocal $, which optimizes the local join for each node and its neighbors. We will elaborate on this in \cref{sec:local-join}.

\subsection{Optimizing the Local Join with Neighbors}
\label{sec:local-join}
One essential component of this algorithm is to, for a vertex $q$ in the meta-decomposition, find the optimal plan to join the relation in this vertex and the optimal query plans given by (a subset of) its neighbors ($\sf optimizeLocal$).  This can be viewed as a classic star join optimization problem, where $R_q$ is the hub relation, and the result of the induced query $Q(T_{c_i})$ from each of its neighbors $c_i$ is a satellite relation, denoted by $R_{i}$ for $i = 1, \dots, f(q)$, where $\chi(c_i) \cap \chi(q) \neq \emptyset$ and $\chi(c_i) \cap \chi(c_j) \subseteq \chi(q)$ for all $i \neq j$.
Unfortunately, even in this restricted setting, this problem of optimizing the join of $R_q$ and all $R_i$'s is still \NP{}-hard, because this can be reduced from the problem of optimizing the ordering of a Cartesian product $S_1 \times S_2 \times \cdots \times S_n$---which is already shown to be \NP{}-hard \cite{scheufele_constructing_1996}---by adding a shared attribute to each $S_i$ such that all tuples have the same value for this attribute.  By using dynamic programming algorithms~\cite{moerkotte_dynamic_2008,stoian_dpconv_2024}, the problem can be solved in $\widetilde{O}(2^{f(q)} \cdot (f(q))^3)$
, where $\widetilde{O}$ suppresses poly-logarithmic factors, and the fan-out $f(q)$ is the number of children of $q$ on the meta-decomposition.  However, as mentioned in Section~\ref{sec:hierarchical-query-plan}, \emph{the fan-out of real-world queries is typically very small}. This makes the exponential-time exact algorithm practical for the vast majority of cases. 
In the rare case that a node has a high fan-out, making the exact DP algorithm too slow, we can substitute it with a heuristic. An effective choice is Greedy Operator Ordering~(GOO)~\cite{fegaras_new_1998}, which iteratively joins the neighbor that results in the smallest intermediate result. This method is outlined in the \arxivorappendix.

\subsection{Case Study: Optimizing Relation-Dominated Queries with Meta-Decompositions}

Our framework naturally incorporates projection pushdown for early cardinality reduction. When building a plan for a subquery, such as $\textsf{plan}(T_{p \to q})$ in \cref{alg:cost}, we can immediately project the result to only the attributes needed for subsequent operations.
Specifically, after computing the joins within the subtree $T_q$, we only need to retain (1)~attributes that are part of the final output, and
(2)~attributes required for joining with the remaining relations, which is $\kappa(q)$ by definition.

This strategy is especially effective for relation-dominated queries, where a single relation contains all output attributes.  If a subtree contains no output attributes, we can project its result to only $\kappa(q)$, i.e., the join keys with the remaining relations, which can significantly reduce the size of intermediate results.
However, traditional dynamic-programming-based optimizers lack such a comprehensive structural view. Therefore, they typically do not identify such opportunities and must consider the entire search space of query plans, relying solely on cost estimation.

\begin{example} \label{ex:rel-dom}
\input{figs/relation-dominated}
    Take query 17f in the join order benchmark~(JOB) as an example. Omitting irrelevant attributes in each relation, this query can be abstractly represented as
    \begin{align*}
        Q_{\ref{ex:rel-dom}} \gets \pi_{x_{\rm n}} ( 
            & R_{\sf ci}[x_{\sf mid}, x_{\sf pid}] \bowtie R_{\sf cn}[x_{\sf cid}, x_{\sf cc}] \bowtie \sigma_{x_{\sf k} = \cdots}(R_{\sf k}[x_{\sf kid}, x_{\sf k}])
            \bowtie R_{\sf mc}[x_{\sf mid}, x_{\sf cid}] \\
            & \bowtie R_{\sf mk}[x_{\sf mid}, x_{\sf kid}] \bowtie \sigma_{x_{\sf n} = \cdots}(R_{\sf n}[x_{\sf pid}, x_{\sf n}]) \bowtie R_{\sf t}[x_{\sf mid}, x_{\sf t}] 
        ).
    \end{align*}
    This is a relation-dominated query, as it has only one output attribute $x_{\sf n}$, which is an attribute of $R_{\sf n}$.
    \cref{fig:rel-dom} shows the hypergraph, the optimal width-1 query plan (found by \sys{} given cardinalities of the real data set), along with the join tree inducing it, and a width-2 query plan (which is the plan given by dynamic programming).
    In the width-1 plan in \cref{fig:rel-dom-w1}, \emph{all intermediate join results can be projected to a single attribute}, since the interface is always only one attribute, and there is no output attribute before the final join with $R_{\sf n}$.
    But, with the width-2 query plan in \cref{fig:rel-dom-duckdb}, \emph{the result of $((R_{\sf mk} \bowtie R_{\sf k}) \bowtie R_{\sf t}) \bowtie R_{\sf ci}$ has to be projected to two attributes $x_{\sf mid}$ and $x_{\sf pid}$}, which are join keys with $R_{\sf mc}$ and $R_{\sf n}$, respectively.
    Executing these plans, we indeed observe that \emph{the width-1 plan has a {\rm 1.72x} speedup over the width-2 plan}.
\end{example}

As with evaluating select-project-join queries, one can easily modify the cost model and use a join tree–based query evaluation algorithm, including any Yannakakis-style algorithm, that incorporates cost-based optimization via meta-decompositions. We hope this work inspires further research in related areas and encourages the integration of these theoretically desirable query evaluation techniques into practical systems.

%% file: figs/relation-dominated.tex
\begin{figure}
    \centering
    \small
    \hfill
    \begin{subfigure}{.5\linewidth}
        \centering
        \vspace{-0.5\baselineskip}
        \begin{tikzpicture}[hypergraph,xscale=.75,yscale=0.7,trim left=0cm,trim right=0cm]
            \node[vertex,label={west:$x_{\sf mid}$}] (mid) at (0, 0) {};
            \node[vertex,label={east:$x_{\sf kid}$}] (kid) at (-2.5, 0) {};
            \node[vertex,label={east:$x_{\sf k}$}] (k) at (-4, 0) {};

            \node[vertex,label={east:$x_{\sf pid}$}] (pid) at (0, -1.5) {};
            \node[vertex,label={east:$x_{\sf n}$}] (n) at (0, -2.5) {};
            
            \node[vertex,label={west:$x_{\sf cid}$}] (cid) at (2.5, 0) {};
            \node[vertex,label={west:$x_{\sf cc}$}] (cc) at (4, 0) {};
            
            \node[vertex,label={west:$x_{\sf t}$}] (t) at (0, 1) {};
        
            \begin{scope}[on background layer]

                \newcommand*{\patht}{($(t) + (0,0.5)$)
        					to[out=180,in=180] ($(mid) + (0,-0.5)$)
        					to[out=0,in=0] ($(t) + (0,0.5)$)}

                \newcommand*{\pathmk}{	($(mid) + (0.25,0.25)$)
                    to[out=135,in=135] ($(kid) + (-0.25,-0.25)$)
                    to[out=-45,in=-45] ($(mid) + (0.25,0.25)$)}
                
                \newcommand*{\pathk}{($(kid) + (0.6,0)$)
        					to[out=90,in=90] ($(k) + (-0.5,0)$)
        					to[out=-90,in=-90] ($(kid) + (0.6,0)$)}

                \newcommand*{\pathci}{($(mid) + (0,0.5)$)
        					to[out=180,in=180] ($(pid) + (0,-0.5)$)
        					to[out=0,in=0] ($(mid) + (0,0.5)$)}

                \newcommand*{\pathn}{($(pid) + (0,0.5)$)
        					to[out=180,in=180] ($(n) + (0,-0.5)$)
        					to[out=0,in=0] ($(pid) + (0,0.5)$)}
                
                \newcommand*{\pathmc}{($(mid) + (-0.25,0.25)$)
                    to[out=45,in=45] ($(cid) + (0.25,-0.25)$)
                    to[out=-135,in=-135] ($(mid) + (-0.25,0.25)$)}
                
                \newcommand*{\pathcn}{($(cid) + (-0.5,0)$)
        					to[out=90,in=90] ($(cc) + (0.5,0)$)
        					to[out=-90,in=-90] ($(cid) + (-0.5,0)$)}
                
                \draw[hyperedge,looseness=1,draw=Red,fill=Red,fill opacity=.1] \patht;
                \draw[hyperedge,looseness=1,draw=Fuchsia,fill=Fuchsia,fill opacity=.1] \pathmk;
                \draw[hyperedge,looseness=1,draw=NavyBlue,fill=NavyBlue,fill opacity=.1] \pathk;
                \draw[hyperedge,looseness=1, draw=BurntOrange,fill=BurntOrange,fill opacity=.1] \pathmc;
                \draw[hyperedge,looseness=1, draw=TealBlue,fill=TealBlue,fill opacity=.1] \pathcn;
                \draw[hyperedge,looseness=1, draw=Gray,fill=Gray,fill opacity=.1] \pathci;
                \draw[hyperedge,looseness=1, draw=Thistle,fill=Thistle,fill opacity=.1] \pathn;
        
                \node[text=Red] at (0.75,1.25) {$R_{\sf t}$};
                \node[text=Fuchsia] at (-1.5,1) {$R_{\sf mk}$};
                \node[text=Cyan] at (-3.25,1) {$R_{\sf k}$};
                \node[text=BurntOrange] at (1.6,1) {$R_{\sf mc}$};
                \node[text=TealBlue] at (3.25,1) {$R_{\sf cn}$};
                \node[text=Gray] at (-1,-1) {$R_{\sf ci}$};
                \node[text=Thistle] at (-1,-2.25) {$R_{\sf n}$};
                
            \end{scope}
        \end{tikzpicture}
        \caption{Hypergraph}
        \label{fig:rel-dom-hg}
    \end{subfigure}
    \hfill
    \begin{subfigure}{.33\linewidth}
        \centering
        \begin{tikzpicture}[
			level distance=0.9cm,
			level 1/.style={sibling distance=0.9cm}
        ]
			\node[align=center] { minor node \\ $\chi: \{ x_{\sf mid} \}$ }
				child { node {$R_{\sf mk}$} edge from parent[-]
                    child { node {$R_{\sf k}$} edge from parent[-] }
                }
                child { node {$R_{\sf ci}$} edge from parent[-]
                    child { node {$R_{\sf n}$} edge from parent[-] }
                }
                child { node {$R_{\sf mc}$} edge from parent[-]
                    child { node {$R_{\sf cn}$} edge from parent[-] }
                }
                child { node {$R_{\sf t}$} edge from parent[-] };
                
		\end{tikzpicture}
        \caption{Meta-decomposition (simplified)}
        \label{fig:rel-dom-meta}
    \end{subfigure}
    \hfill \ 

    \begin{subfigure}{.3\linewidth}
        \centering
        \begin{tikzpicture}[
			level distance=0.65cm,
			level 1/.style={sibling distance=1cm}
        ]
			\node { $R_{\sf n}$ }
				child { node {$R_{\sf ci}$} edge from parent[-]
                    child { node {$R_{\sf mk}$} edge from parent[-]
                        child { node {$R_{\sf k}$} edge from parent[-] }
                        child { node {$R_{\sf t}$} edge from parent[-] }
                        child { node {$R_{\sf mc}$} edge from parent[-]
                            child { node {$R_{\sf cn}$} edge from parent[-] }
                        }
                    }
                };
                
		\end{tikzpicture}
        \caption{Join tree inducing the width-1 plan in \cref{fig:rel-dom-w1}}
        \label{fig:rel-dom-jt}
    \end{subfigure}
    \hfill
    \begin{subfigure}{.33\linewidth}
        \centering
        \begin{tikzpicture}[
			level distance=0.55cm,
            level 1/.style={sibling distance=1.6cm},
			level 2/.style={sibling distance=1.6cm},
            level 3/.style={sibling distance=1.6cm},
            level 4/.style={sibling distance=0.8cm}
        ]
			\node { $\bowtie$ }
                child { node {$\bowtie$} edge from parent[-]
                    child { node {$\bowtie$} edge from parent[-]
                        child { node {$\bowtie$} edge from parent[-]
            				child { node {$\bowtie$} edge from parent[-]
                                child { node {$R_{\sf mk}$} edge from parent[-] }
                                child { node {$R_{\sf k}$} edge from parent[-] }
                            }
                            child { node {$R_{\sf t}$} edge from parent[-] }
                        }
                        child { node {$\bowtie$} edge from parent[-]
                            child { node {$R_{\sf mc}$} edge from parent[-] }
                            child { node {$R_{\sf cn}$} edge from parent[-] }
                        }
                    }
                    child { node {$R_{\sf ci}$} edge from parent[-] } 
                }
                child { node {$R_{\sf n}$} edge from parent[-] };
		\end{tikzpicture}
        \caption{Optimal width-1 query plan, induced by the join tree in \cref{fig:rel-dom-jt}}
        \label{fig:rel-dom-w1}
    \end{subfigure}
    \hfill
    \begin{subfigure}{.3\linewidth}
        \centering
        \begin{tikzpicture}[
			level distance=0.55cm,
			level 1/.style={sibling distance=1cm}
        ]
			\node { $\bowtie$ }
                child { node {$\bowtie$} edge from parent[-]
                    child { node {$\bowtie$} edge from parent[-]
                        child { node {$\bowtie$} edge from parent[-]
            				child { node {$\bowtie$} edge from parent[-]
                				child { node {$\bowtie$} edge from parent[-]
                                    child { node {$R_{\sf mk}$} edge from parent[-] }
                                    child { node {$R_{\sf k}$} edge from parent[-] }
                                }
                                child { node {$R_{\sf t}$} edge from parent[-] }
                            }
                            child { node {$R_{\sf ci}$} edge from parent[-] }
                        }
                        child { node {$R_{\sf n}$} edge from parent[-] }
                    }
                    child { node {$R_{\sf mc}$} edge from parent[-] }
                }
                child { node {$R_{\sf cn}$} edge from parent[-] };
		\end{tikzpicture}
        \caption{Optimal plan by dynamic programming, of width 2}
        \label{fig:rel-dom-duckdb}
    \end{subfigure}
    \caption{The hypergraph, the meta-decomposition, the optimal width-1 query plan with the inducing join tree, and the optimal query plan (of width 2) found by dynamic programming, of the query JOB 17f in \cref{ex:rel-dom}}
    \label{fig:rel-dom}
\end{figure}
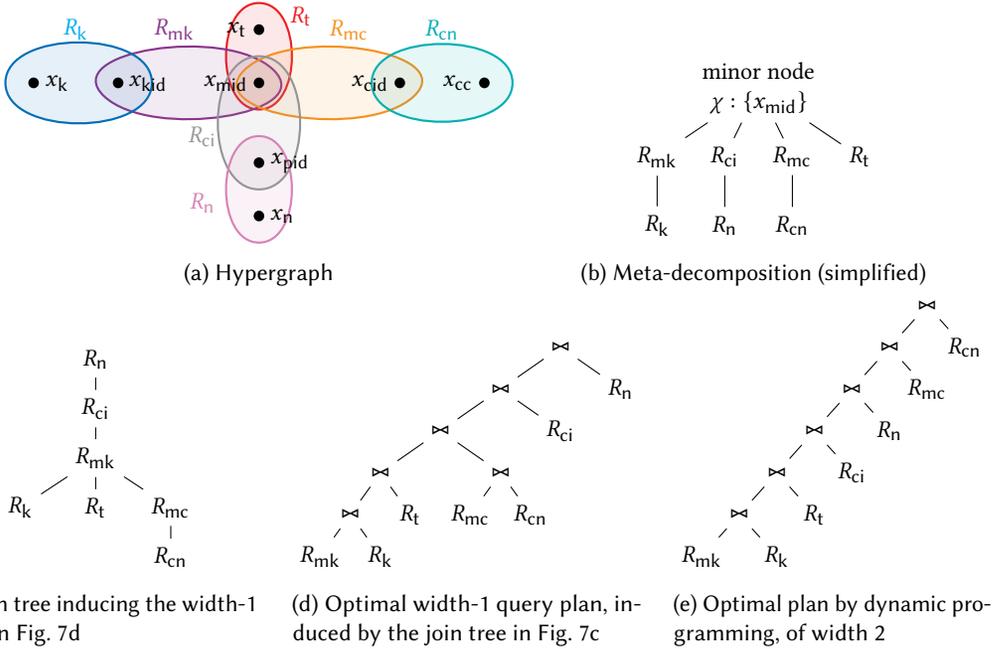

%% file: 6_experiments.tex
\section{Experimental Evaluation and Analysis}
\label{sec:exp}

In this section, we present experimental results to answer the following research questions:

\begin{enumerate}[label=\textbf{(RQ\arabic*)},ref=(RQ\arabic*),leftmargin=*]
    \item\label{RQ1}\gdef\Giref{\ref{RQ1}} How fast is query optimization using meta-decompositions?
	\item\label{RQ2}\gdef\Giref{\ref{RQ2}} How good are width-1 query plans, compared to the globally optimal ones?
    \item\label{RQ3}\gdef\Giref{\ref{RQ3}} To what extent is our proposed meta-decompositions-based query optimization sensitive to errors in cardinality estimation?
    \item\label{RQ4}\gdef\Giref{\ref{RQ4}} To what extent can our proposed meta-decompositions-based join tree enumeration technique help structure-based query evaluation?
\end{enumerate}

\subsection{Experimental Setup}
\subsubsection{Experimental Environment} Experiments are conducted on an Apple M4 Pro machine with 14 cores and 48 GB RAM, running macOS Tahoe 16.0. We use DuckDB as our test database because of its excellent performance and simple cost model. The experiments use DuckDB 1.2.2, Scala 3.6.3, GraalVM 17.0.12, and LLVM 20.1.3. The full source code for experiments is available at~\cite{github}.

We implement \sys{} in Scala. For each query, it first computes the meta-decomposition (Section~\ref{sec:meta}) and then applies our cost-based optimization algorithm (Section~\ref{sec:opt}) to select the best width-1 query plan.  
In the experiments, we compare the following query optimization approaches:
\begin{enumerate}[leftmargin=*]
    \item \sys{}, using our implementation.
    \item DPconv~\cite{stoian_dpconv_2024}, the state-of-the-art dynamic programming optimizer.
    \item DuckDB's native optimizer~\cite{raasveldt_duckdb_2019}, which uses \textsf{DPccp}~\cite{moerkotte_dynamic_2008} for queries with up to 12 relations and Greedy Operator Ordering~(GOO)~\cite{fegaras_new_1998} otherwise.
    \item UnionDP, a heuristic approach proposed in~\cite[Algorithm~4]{mancini_efficient_2022} which (i) partitions the query graph into subgraphs, each with size up to a threshold (which we set to 12, consistent with DuckDB), (ii) runs dynamic programming on each subgraph, and (iii) runs dynamic programming on the entire graph, treating each subgraph as a composite node.
    \item Yannakakis$^+$~\cite{wang_yannakakis_2025}, a state-of-the-art structure-guided query evaluation approach.
    \item LearnedRewrite~\cite{zhou_learned_2023,zhou_learned_2021}, a machine learning--based query rewriter.
    \item LLM-R$^2$~\cite{li_llm-r2_2024}, a query rewriter based on machine learning and large language models (LLMs).
\end{enumerate}
For Yannakakis$^+$, we use the implemented system as a black box and run the rewritten queries in DuckDB.
For LearnedRewrite and LLM-R$^2$, we similarly evaluate them using their pre-trained models as-is and run the rewritten queries in DuckDB.
For \sys{}, DPconv, and UnionDP, the query plans are enforced in DuckDB by rewriting queries into a sequence of subqueries that create a temporary view at each step, as is standard in, e.g., \cite{wang_yannakakis_2025,gottlob_reaching_2023}.  For all three methods, we apply the same strategies of selection and projection pushdown as explained in \cref{sec:prelim-plans}. Namely, we project out all irrelevant attributes (those that are neither in the output nor in subsequent joins) at every step, and we integrate each selection (filter) condition in the join or scan operation at the lowest possible level, i.e., as soon as all relevant attributes exist.

Each query plan is executed $10$ times, and we report the median optimization and execution time. To ensure accuracy, we exclude I/O time from time measurements. The full dataset is preloaded into the DuckDB database, and the estimated cardinalities are preloaded into memory, with all loading time excluded from timing. We enforce a 5-minute timeout on execution time. Queries that none of the optimizers can complete within this time limit are excluded from the results.

\subsubsection{Queries, Benchmarks, and Cardinalities.} We experiment with the following four benchmarks:
\begin{enumerate}[leftmargin=*]
    \item The SPJ (select-project-join) queries in the \textbf{Decision Support Benchmark (DSB)}~\cite{ding_dsb_2021}. We use the implemented system to generate a 10 GB database instance and generate queries using the default settings. We only consider acyclic queries (240 out of 300). 
    \item The \textbf{Join Order Benchmark (JOB)}~\cite{leis_query_2018} over the IMDB dataset. All queries in this benchmark are already acyclic.
    \item \textbf{Musicbrainz}, a benchmark proposed by \cite{mancini_efficient_2022} with significantly larger queries that join up to 26 relations. We use the implemented query generator with default settings, except that we project the output to at most 5 randomly-selected attributes. We retain only the acyclic queries (243 out of 375).
    \item \textbf{JOBLarge}, a new benchmark we introduce to extend JOB to simulate practical but larger analytical scenarios, where each query is generated by selecting two queries from the JOB that share at least one common output attribute and merging them into a single query, with a join on a shared attribute, such that the query remains acyclic. 
We only retain generated queries that take DuckDB at least 50 ms to evaluate.
\end{enumerate}
In \cref{tab:query-props}, we show detailed statistics of some structural properties of queries in each benchmark.

\input{figs/exp/tab-benchmark-stats}

For \sys{}, DPconv, and UnionDP, which all follow the same cost model, we precompute (by running appropriate \texttt{COUNT(*)} queries) true cardinalities of all possible intermediate join results for the relatively small benchmarks DSB and JOB, in the approach and format consistent with the implementation of DPconv~\cite{stoian_dpconv_2024}. 
For the larger benchmarks Musicbrainz and JOBLarge, computing exact cardinalities for all possible intermediate joins is infeasible due to the much larger query sizes. Therefore, we use estimated cardinalities instead, via DuckDB's \texttt{EXPLAIN} command. Note, however, that these estimates may sometimes be highly inaccurate and may not reflect the true relative sizes of intermediate results.
DuckDB, Yannakakis$^+$, Learned Rewrite, and LLM-R$^2$ are used as black boxes with their own cost models.

\subsection{Results}

In this section, we present the results of our experiments.

\subsubsection{\ref{RQ1} Optimization Time}

\begin{figure}
    \includegraphics[width=\linewidth]{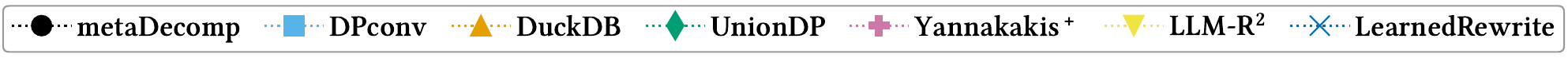}
    
    \vspace{0.5\baselineskip}
    
    \begin{subfigure}[b]{0.49\textwidth}
        \centering
        \includegraphics[width=\linewidth]{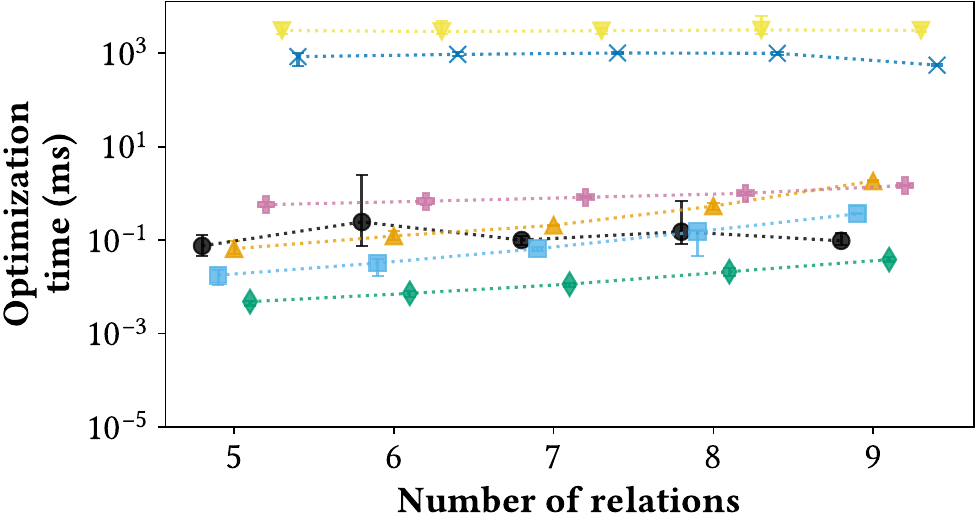}
        \vspace{-1.25\baselineskip}
        \caption{DSB}
        \label{fig:opt-time-dsb}
    \end{subfigure}
    \hfill
    \begin{subfigure}[b]{0.49\textwidth}
        \centering
        \includegraphics[width=\linewidth]{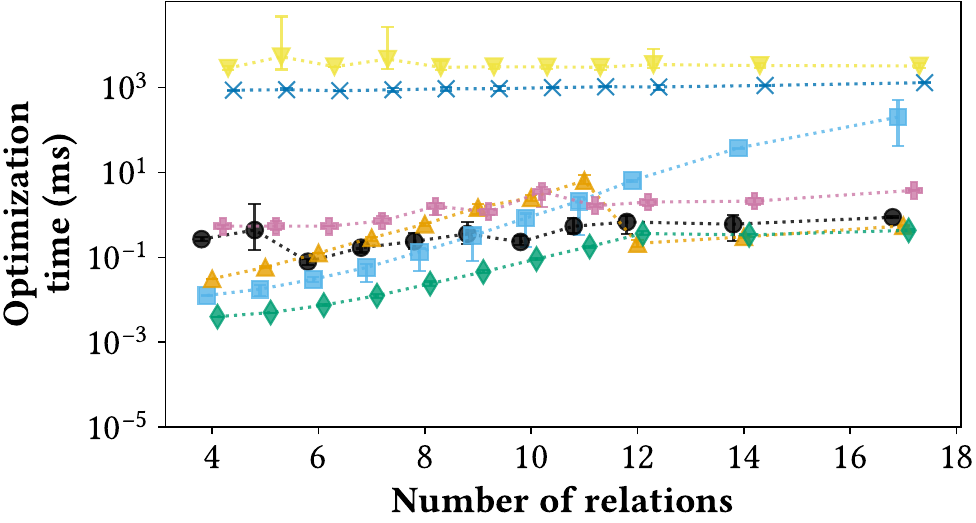}
        \vspace{-1.25\baselineskip}
        \caption{JOB}
        \label{fig:opt-time-job-original}
    \end{subfigure}
    
    \vspace{0.75\baselineskip}
    
    \begin{subfigure}[b]{0.49\textwidth}
        \centering
        \includegraphics[width=\linewidth]{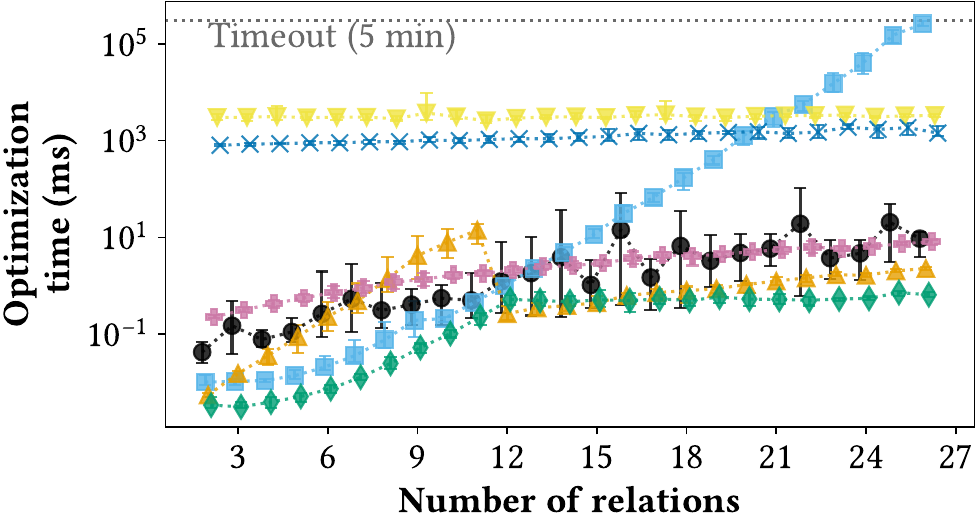}
        \vspace{-1.25\baselineskip}
        \caption{Musicbrainz}
        \label{fig:opt-time-musicbrainz}
    \end{subfigure}
    \hfill
    \begin{subfigure}[b]{0.49\textwidth}
        \centering
        \includegraphics[width=\linewidth]{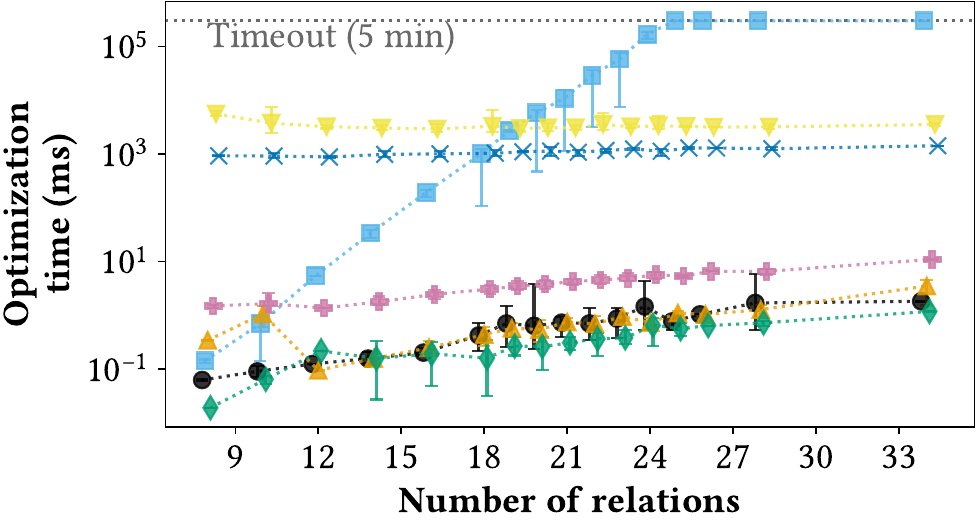}
        \vspace{-1.25\baselineskip}
        \caption{JOBLarge}
        \label{fig:opt-time-job-large}
    \end{subfigure}
    \vspace{-0.25\baselineskip}
    \caption{Optimization time for queries in each benchmark. For each number of relations, the dot shows the average, and the whiskers show the minimum and maximum optimization time.}
    \label{fig:opt-time-overall}
\end{figure}

We first compare the optimization time. The results are shown in \cref{fig:opt-time-overall}.
The figures clearly show that the optimization time of \sys{} remains almost constant as the number of relations increases, typically below 10 ms, as is also the case for DuckDB, UnionDP, and Yannakakis$^+$.
The two machine learning--based approaches, LLM-R$^2$ and LearnedRewrite, exhibit stable but higher optimization time, in the order of seconds, due to the overhead of the complex models and, for LLM-R$^2$, the interaction with LLMs.
DPconv, however, has an exponential growth, and its optimization time surpasses that of \sys{} starting at approximately 9 relations. 
For many queries in Musicbrainz and JOBLarge (starting from approximately 25 relations), DPconv fails to find a valid plan within the 5-minute limit, again highlighting that traditional exponential-time dynamic programming–based optimizers do not scale to such large queries.

\subsubsection{\ref{RQ2} Quality of width-1 query plans}

Even though \sys{} restricts its search to width-1 plans only, as shown in \cref{fig:cost-ratio}, for most queries, it \emph{finds query plans with cost values comparable to those of the optimal plans}.
Furthermore, \cref{fig:exec-speedup} shows the speedups achieved by executing the optimal width-1 query plans over executing the globally optimal query plans. Note that, even though one would typically not expect speedups since the optimal plans have the minimum possible cost values, in many cases, the optimal width-1 query plans in fact run faster in the actual execution, thanks to the small intermediate result sizes through exploiting queries' structural properties.

\begin{figure}
    \begin{subfigure}[b]{0.24\textwidth}
        \centering
        \includegraphics[width=\linewidth]{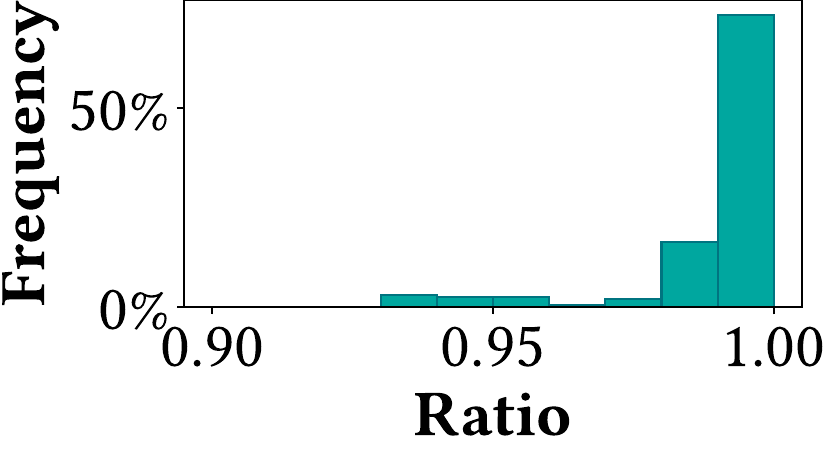}
        \vspace{-1.25\baselineskip}
        \caption{DSB}
        \label{fig:cost-ratios-dsb}
    \end{subfigure}
    \hfill
    \begin{subfigure}[b]{0.24\textwidth}
        \centering
        \includegraphics[width=\linewidth]{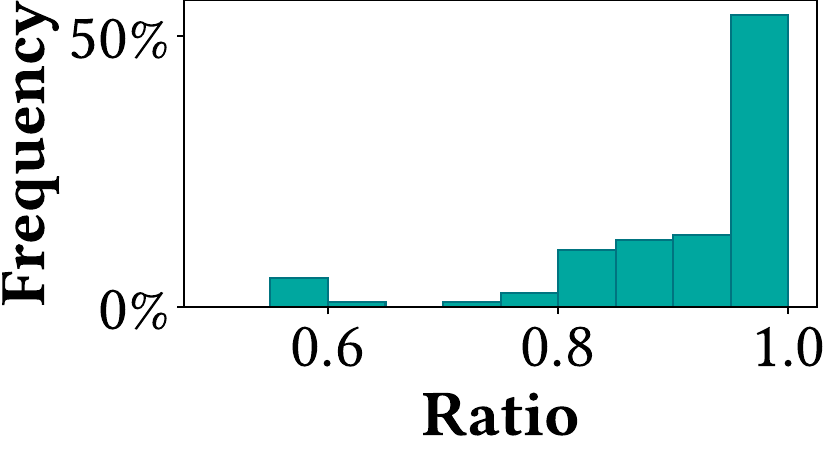}
        \vspace{-1.25\baselineskip}
        \caption{JOB}
        \label{fig:cost-ratios-job-original}
    \end{subfigure}
    \hfill
    \begin{subfigure}[b]{0.24\textwidth}
        \centering
        \includegraphics[width=\linewidth]{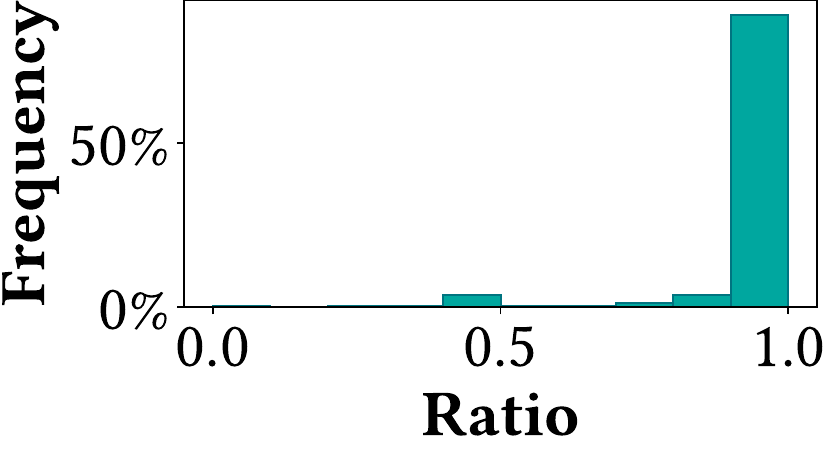}
        \vspace{-1.25\baselineskip}
        \caption{Musicbrainz}
        \label{fig:cost-ratios-musicbrainz}
    \end{subfigure}
    \hfill
    \begin{subfigure}[b]{0.24\textwidth}
        \centering
        \includegraphics[width=\linewidth]{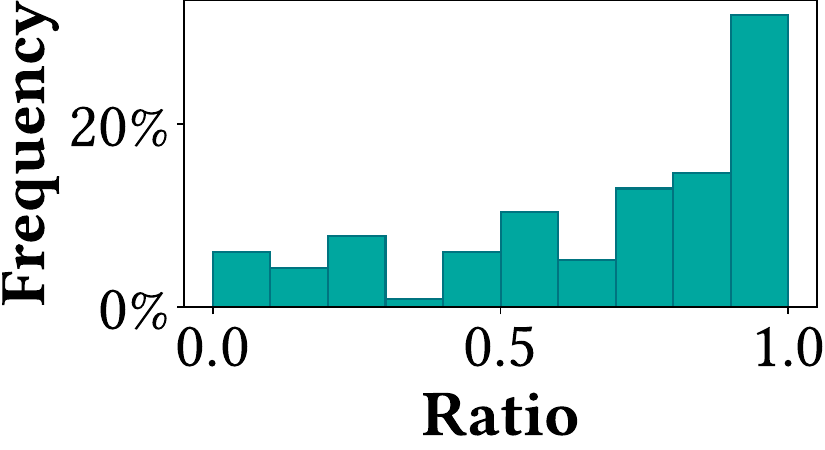}
        \vspace{-1.25\baselineskip}
        \caption{JOBLarge}
        \label{fig:cost-ratios-job-large}
    \end{subfigure}
    
    \vspace{-0.5\baselineskip}
    \caption[Ratios between the costs of the globally optimal query plans and the costs of the optimal width-1 plans for queries in each benchmark]{Ratios between the costs of the globally optimal query plans and the costs of the optimal width-1 plans. A ratio of 1 indicates they have the same cost, and is included in the bar between 0.9 and 1.0.\footnotemark[3]}
    \vspace{0.5\baselineskip}
    \label{fig:cost-ratio}
\end{figure}

    \footnotetext[3]{This measure is analogous to speedups. Since width-1 plans are a restricted class, their theoretical cost values cannot be better (lower) than the global optimal, and thus the ratios are at most 1.}

\begin{figure}
    \begin{subfigure}[b]{0.49\textwidth}
        \centering
        \includegraphics[width=\linewidth]{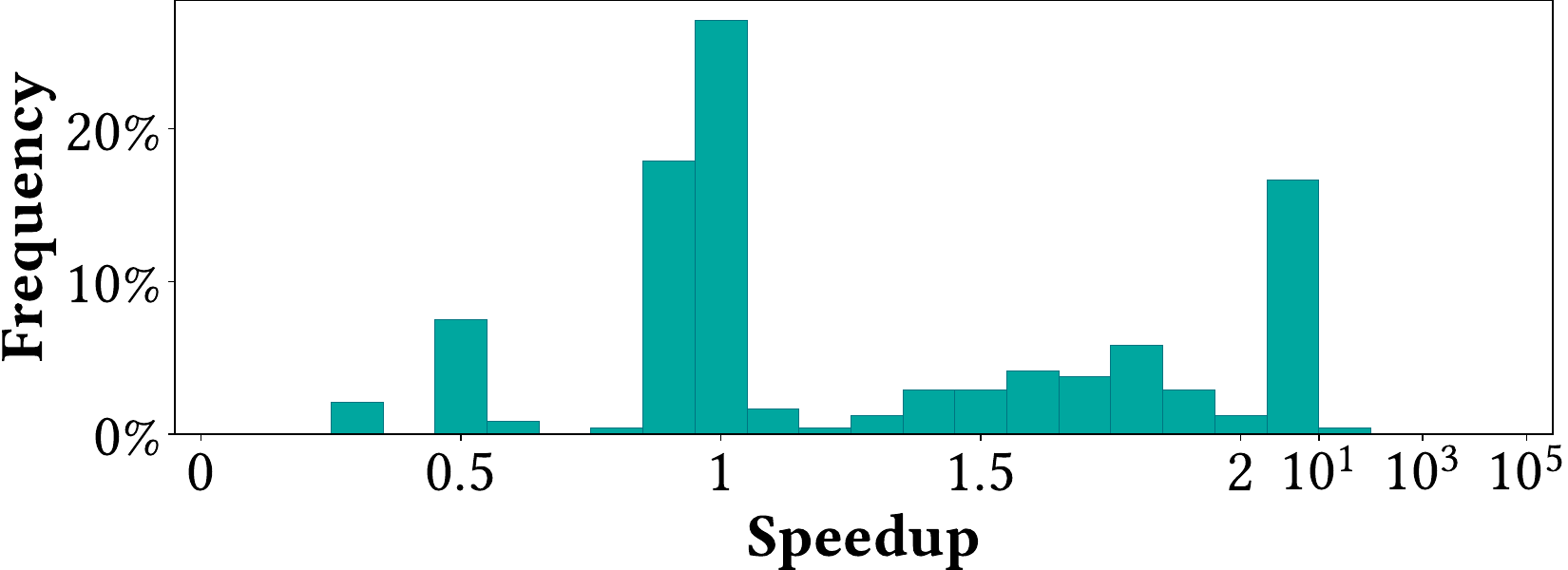}
        \vspace{-1.25\baselineskip}
        \caption{DSB}
        \label{fig:exec-speedup-dsb}
    \end{subfigure}
    \hfill
    \begin{subfigure}[b]{0.49\textwidth}
        \centering
        \includegraphics[width=\linewidth]{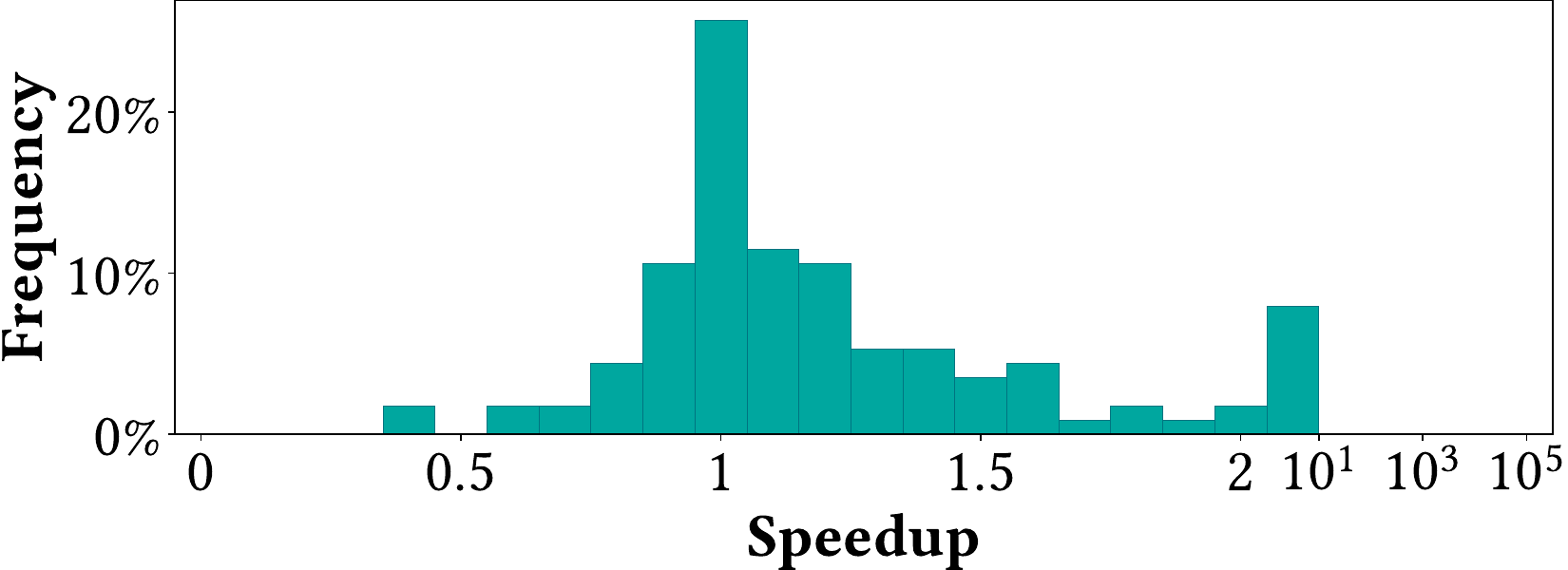}
        \vspace{-1.25\baselineskip}
        \caption{JOB}
        \label{fig:exec-speedup-job-original}
    \end{subfigure}

    \vspace{0.5\baselineskip}

    \begin{subfigure}[b]{0.49\textwidth}
        \centering
        \includegraphics[width=\linewidth]{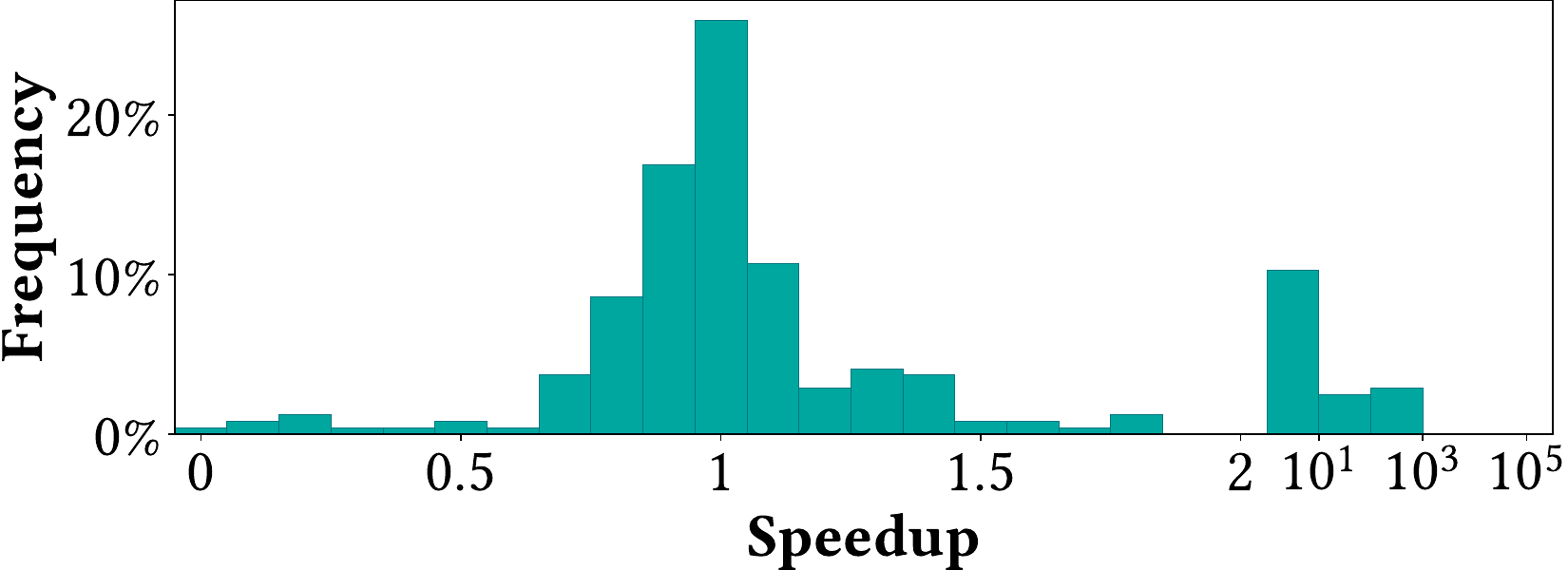}
        \vspace{-1.25\baselineskip}
        \caption{Musicbrainz}
        \label{fig:exec-speedup-musicbrainz}
    \end{subfigure}
    \hfill
    \begin{subfigure}[b]{0.49\textwidth}
        \centering
        \includegraphics[width=\linewidth]{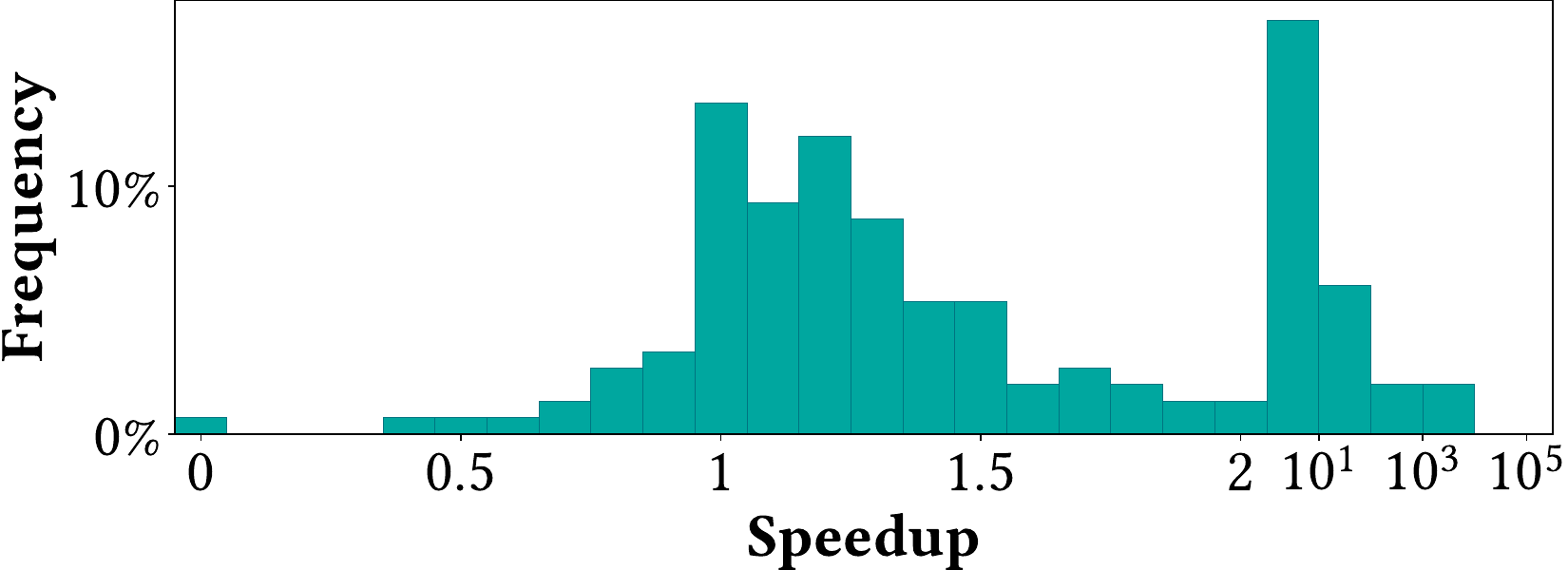}
        \vspace{-1.25\baselineskip}
        \caption{JOBLarge}
        \label{fig:exec-speedup-job-large}
    \end{subfigure}
    
    \vspace{-0.5\baselineskip}
    \caption{Histograms of execution speedups of the optimal width-1 plans over the globally optimal plans}
    \label{fig:exec-speedup}
\end{figure}

\subsubsection{\ref{RQ1} \& \ref{RQ2} Combined: Overall Performance}

\input{figs/exp/tab-results}

\afterpage{
\begin{figure}
    \includegraphics[width=0.55\linewidth]{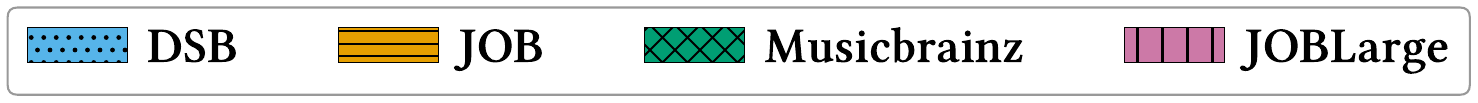}
    
    \vspace{0.25\baselineskip}
    
    \begin{subfigure}[b]{0.49\linewidth}
        \centering
        \includegraphics[width=\linewidth]{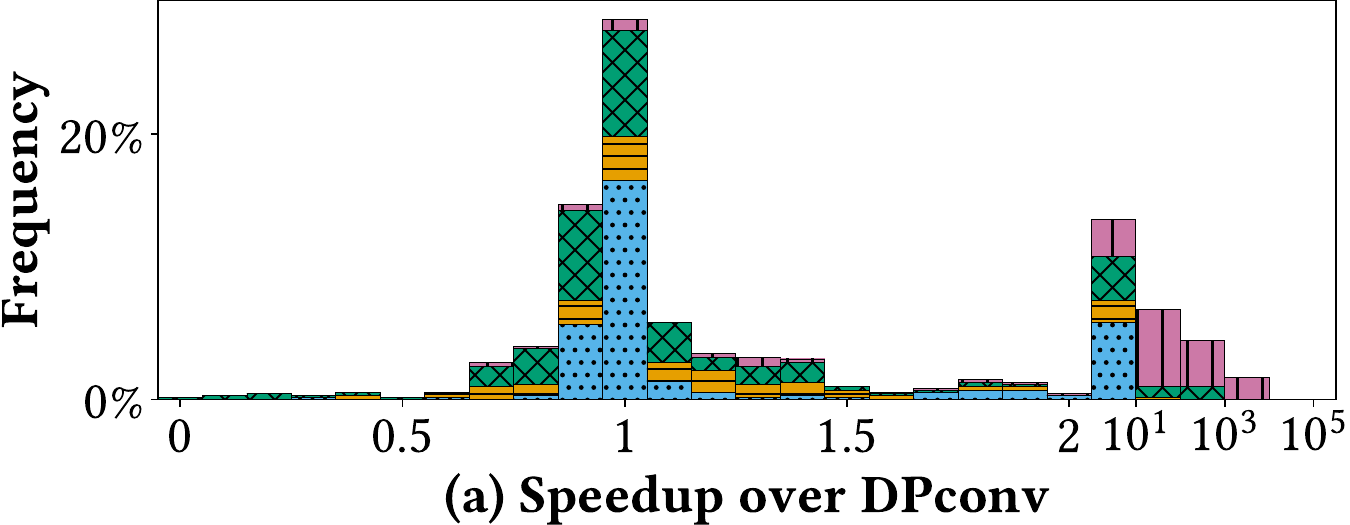}
        \label{fig:overall-speedup-dpconv}
    \end{subfigure}
    \hfill
    \begin{subfigure}[b]{0.49\linewidth}
        \centering
        \includegraphics[width=\linewidth]{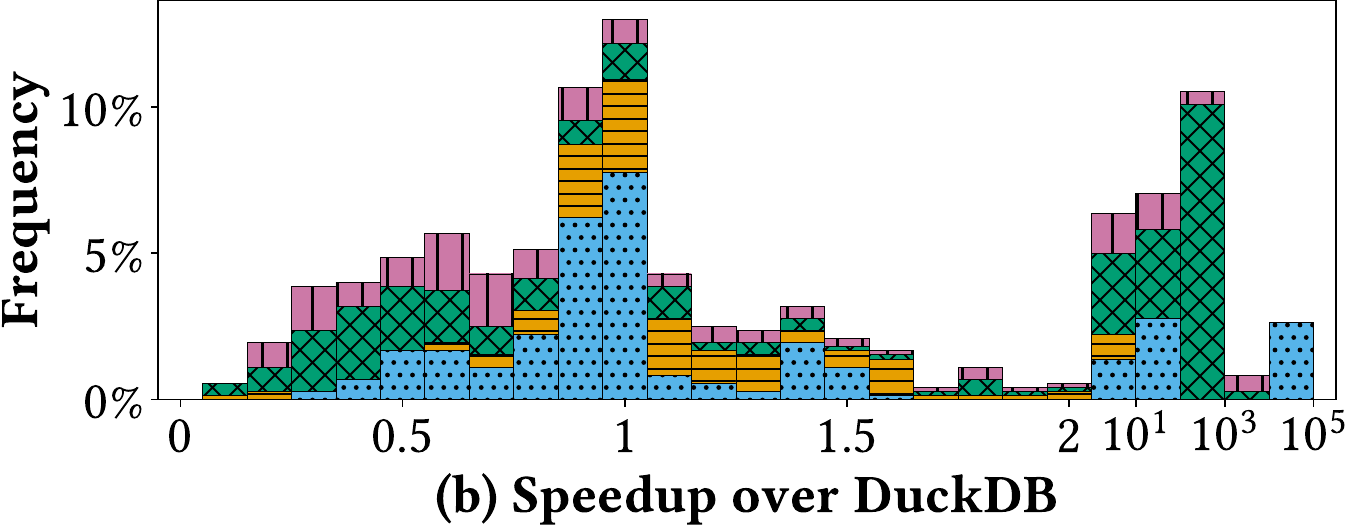}
        \label{fig:overall-speedup-duckdb}
    \end{subfigure}
    
    \vspace{-0.5\baselineskip}
    
    \begin{subfigure}[b]{0.49\linewidth}
        \centering
        \includegraphics[width=\linewidth]{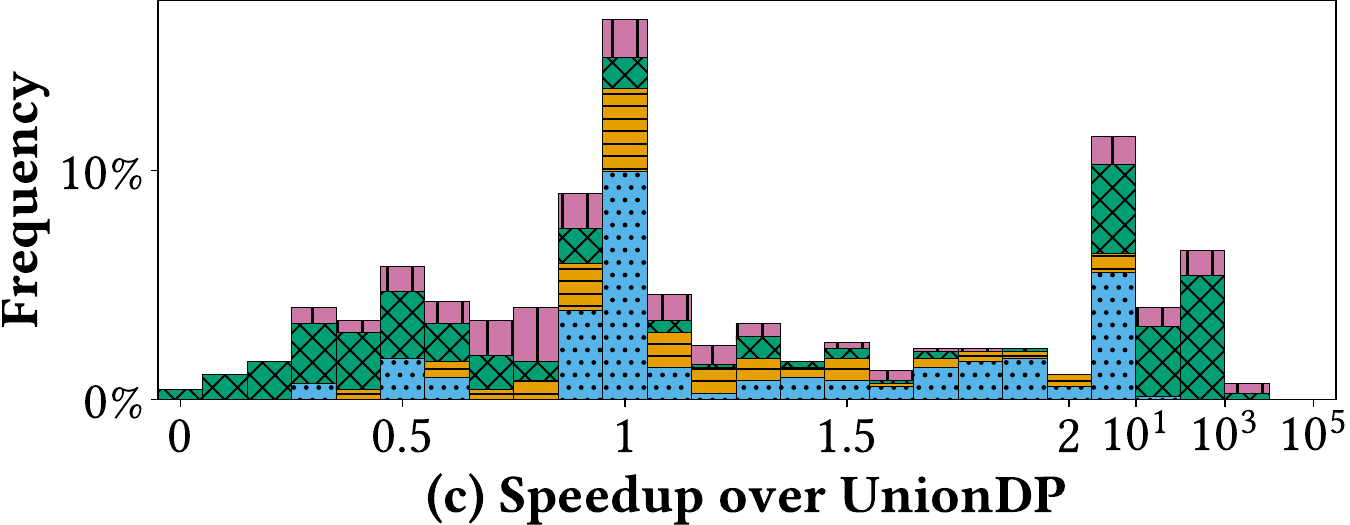}
        \label{fig:overall-speedup-uniondp}
    \end{subfigure}
    \hfill
    \begin{subfigure}[b]{0.49\linewidth}
        \centering
        \includegraphics[width=\linewidth]{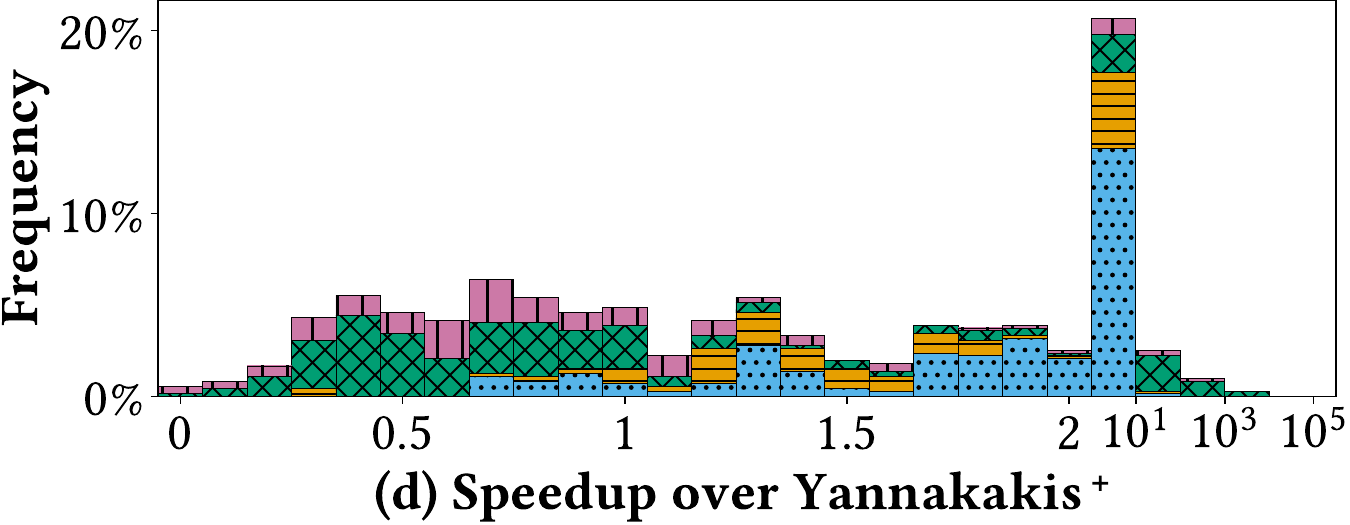}
        \label{fig:overall-speedup-yanplus}
    \end{subfigure}

    \vspace{-0.5\baselineskip}

    \begin{subfigure}[b]{0.49\linewidth}
        \centering
        \includegraphics[width=\linewidth]{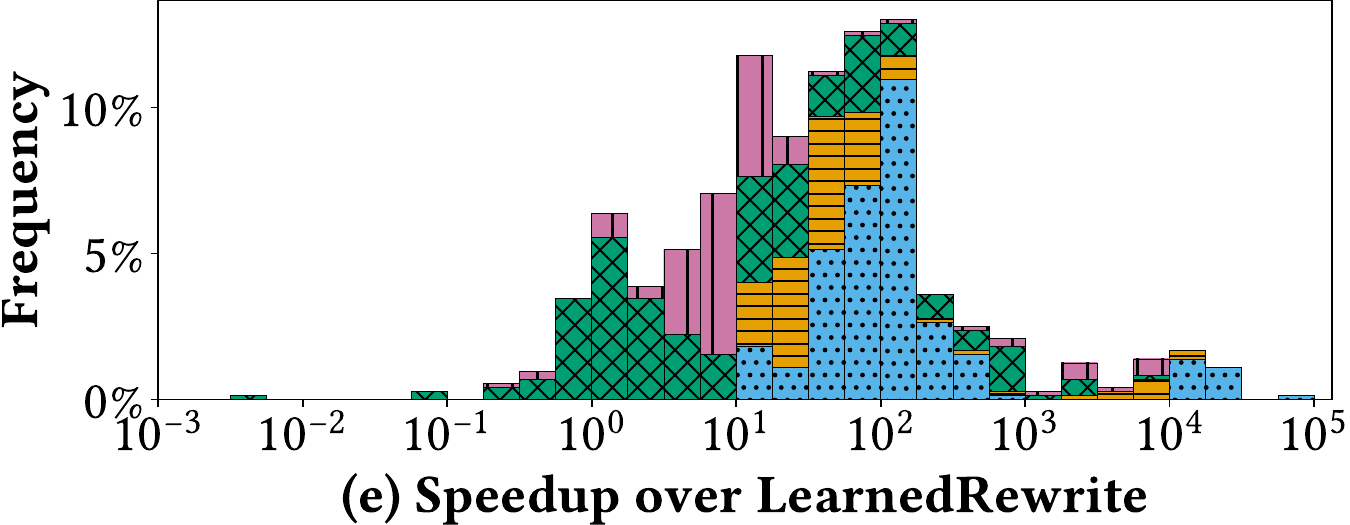}
        \label{fig:overall-speedup-learned-rewrite}
    \end{subfigure}
    \hfill
    \begin{subfigure}[b]{0.49\linewidth}
        \centering
        \includegraphics[width=\linewidth]{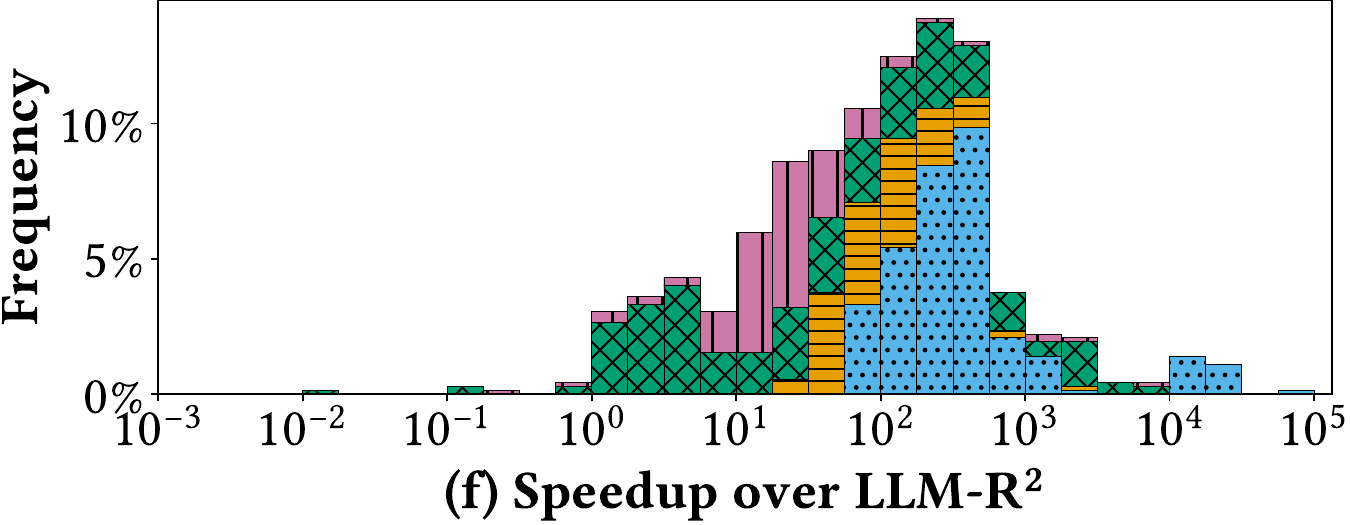}
        \label{fig:overall-speedup-llm-r2}
    \end{subfigure}
    
    \vspace{-1.5\baselineskip}
    
    \caption{Histograms of overall speedups of using \sys{} over using other methods}
    \label{fig:overall-speedup-histogram}
\end{figure}
\begin{figure}
    \includegraphics[width=0.55\linewidth]{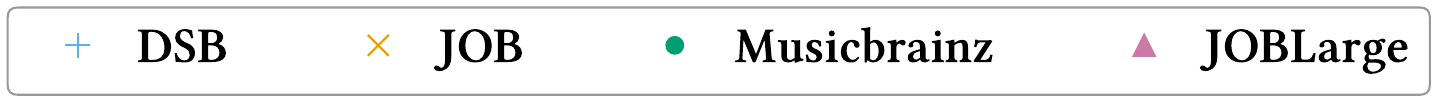}
    
    \vspace{0.5\baselineskip}
    
    \begin{subfigure}[b]{0.28\linewidth}
        \centering
        \includegraphics[width=\linewidth]{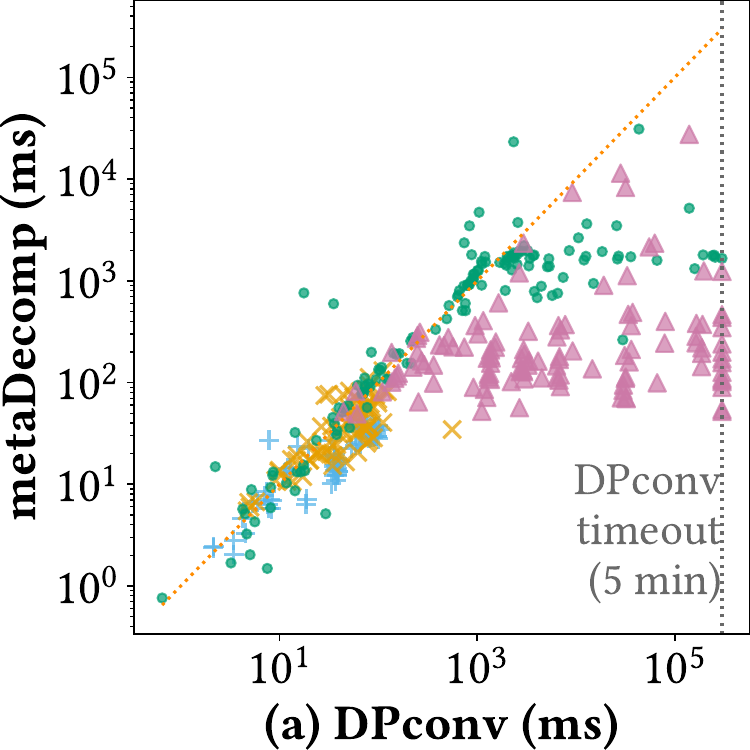}
        \label{fig:overall-scatter-dpconv}
    \end{subfigure}
    \hfill
    \begin{subfigure}[b]{0.28\linewidth}
        \centering
        \includegraphics[width=\linewidth]{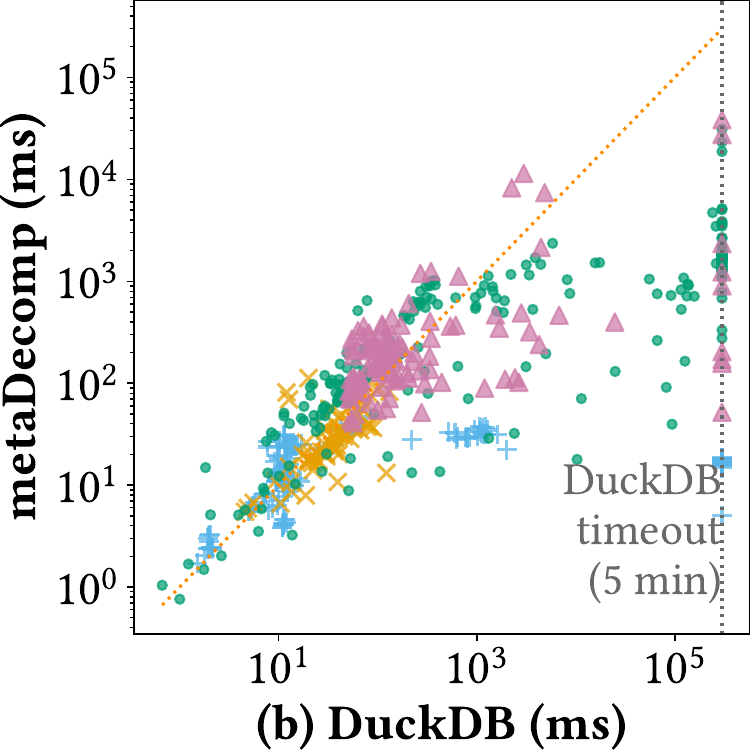}
        \label{fig:overall-scatter-duckdb}
    \end{subfigure}
    \hfill
    \begin{subfigure}[b]{0.28\linewidth}
        \centering
        \includegraphics[width=\linewidth]{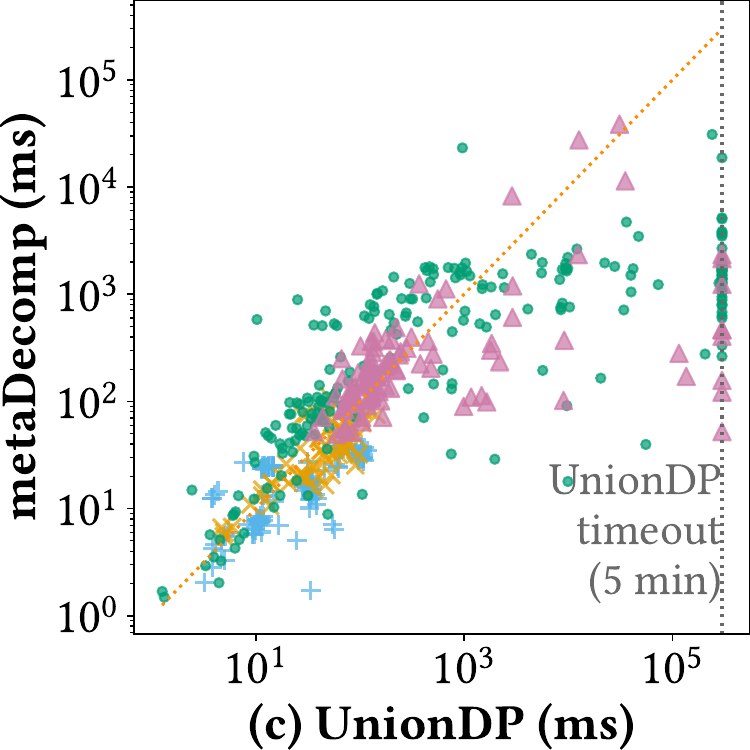}
        \label{fig:overall-scatter-uniondp}
    \end{subfigure}
    
    \vspace{-0.5\baselineskip}
    
    \begin{subfigure}[b]{0.28\linewidth}
        \centering
        \includegraphics[width=\linewidth]{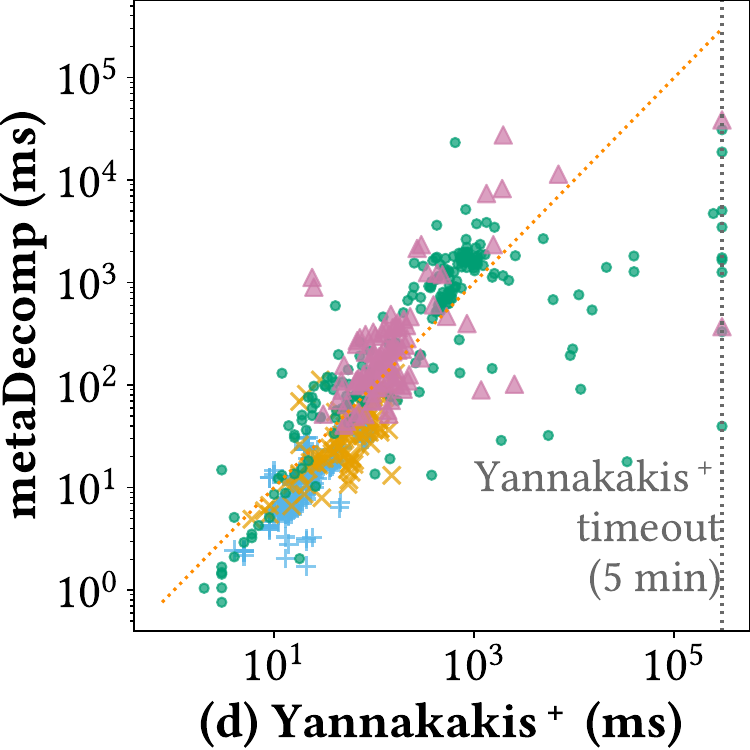}
        \label{fig:overall-scatter-yanplus}
    \end{subfigure}
    \hfill
    \begin{subfigure}[b]{0.28\linewidth}
        \centering
        \includegraphics[width=\linewidth]{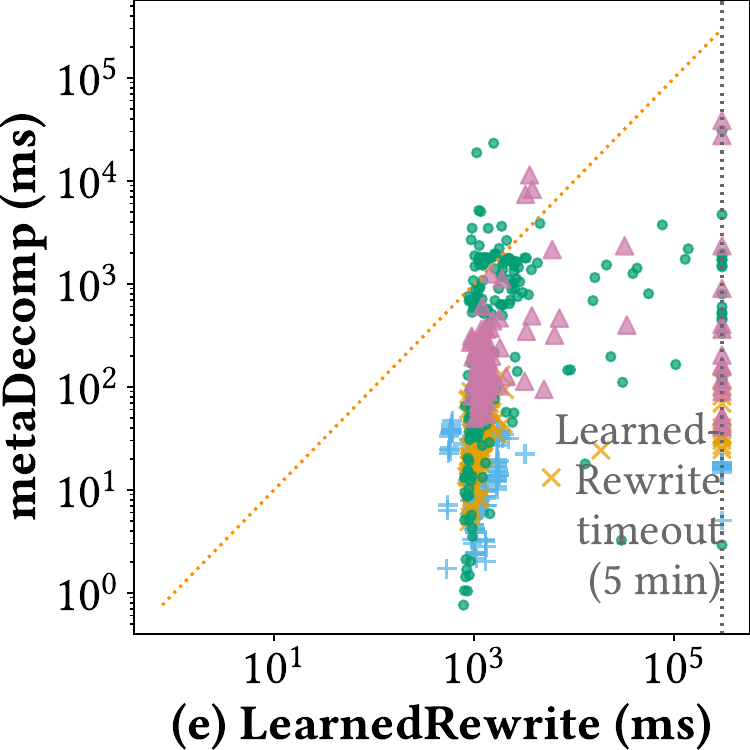}
        \label{fig:overall-scatter-learned-rewrite}
    \end{subfigure}
    \hfill
    \begin{subfigure}[b]{0.28\linewidth}
        \centering
        \includegraphics[width=\linewidth]{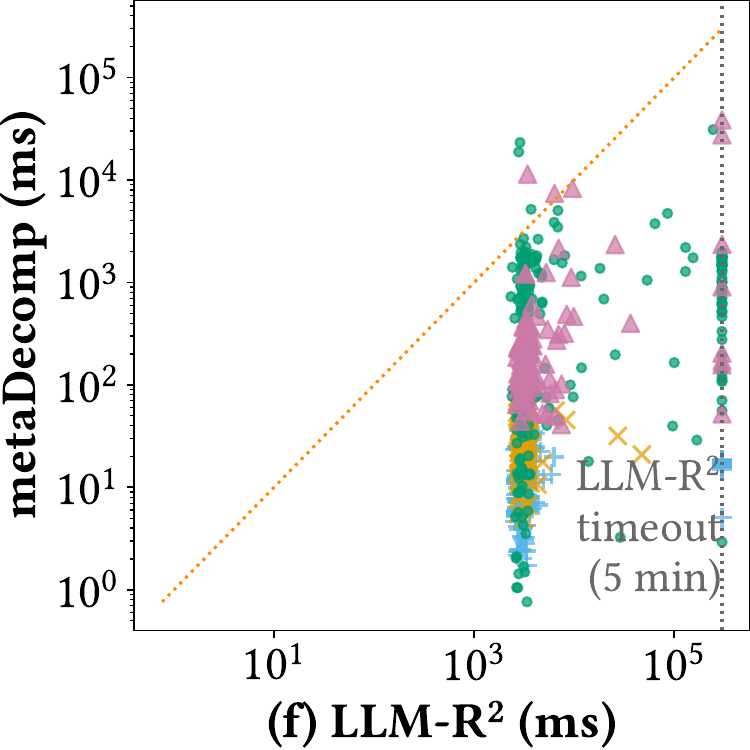}
        \label{fig:overall-scatter-llm-r2}
    \end{subfigure}
    
    \vspace{-1.5\baselineskip}
    \caption{Scatter plots of overall evaluation time using metaDecomp versus using other methods. Queries for which the difference is less than 10\% are omitted in the figures.}
    \label{fig:overall-scatter}
    \vspace{0.5\baselineskip}
\end{figure}
}

Taking both optimization and execution time into account, we observe in \cref{tab:stats-new}, \cref{fig:overall-speedup-histogram}, and \cref{fig:overall-scatter} that \sys{} shows a clear advantage overall. This is especially the case for queries that are complex, as shown by the 95th and 99th percentiles in \cref{tab:stats-new}, and in the points towards the right in the scatter plots in \cref{fig:overall-scatter}. 
For DuckDB, UnionDP, Yannakakis$^+$, LearnedRewrite, and LLM-R$^2$, although their optimization is fast, they are also very likely to generate plans that are significantly worse than the optimal width-1 query plans, especially for larger queries.
In addition, for Yannakakis$^+$, there is an overhead in the reduction phase, which may be costly for certain queries.
For DPconv, the main disadvantage is the exponential optimization time, which becomes the dominant factor in the overall evaluation time for large queries.
The overall results demonstrate that finding and executing width-1 plans with \sys{} is indeed highly efficient in practice.
More detailed, separate figures for individual benchmarks can be found in the \arxivorappendix.

\subsubsection{\ref{RQ3} Sensitivity to cardinality misestimation}
Recent research has shown that structure-based query optimization tends to be more robust even in the presence of errors in cardinality estimation~\cite{wang_yannakakis_2025,DiamondJoin2024,RPT2025}.
This is also the case for \sys{}.
To illustrate, we compare \sys{} with DPconv and UnionDP, which use the exact same cost model. We run them on DSB and JOB queries with (1) exact and (2) misestimated cardinalities.
The misestimated cardinalities are generated by adding a multiplicative noise to the exact cardinalities using a log-normal distribution. Specifically, for each exact cardinality $C$, we generate a perturbed estimate $\hat{C} = C \cdot e^\epsilon$, where $\epsilon \sim \mathcal{N}(0, \sigma^2)$, i.e., following a normal (Gaussian) distribution, with standard deviation $\sigma$ set to $10$.
From the results in \cref{tab:misestimation}, \sys{} still has excellent performance despite such errors in cardinality estimation.
The advantage is particularly pronounced at the 99th percentile of the JOB benchmark, where DPconv and UnionDP, which rely solely on cardinality values, are significantly misled by the highly inaccurate estimation.

\input{figs/exp/tab-misestimation}

\subsubsection{\ref{RQ4} Meta-decompositions applied to other structure-based optimization techniques}
\begin{figure}
    \includegraphics[width=0.37\linewidth]{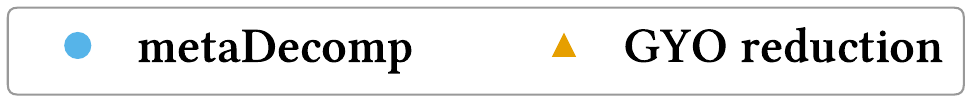}
    
    \vspace{0.25\baselineskip}
    
    \begin{subfigure}[b]{0.24\textwidth}
        \centering
        \includegraphics[width=\linewidth]{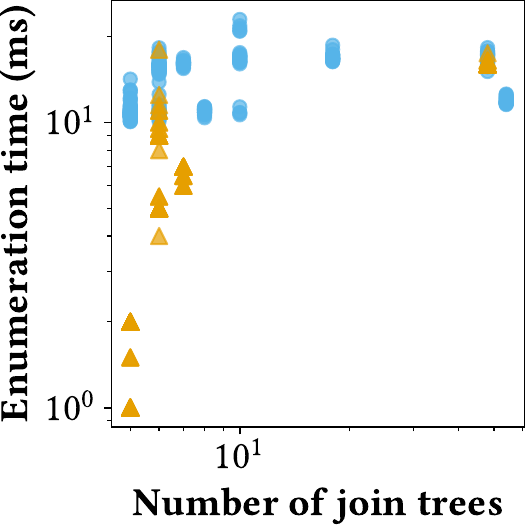}
        \caption{DSB}
        \label{fig:enum-dsb}
    \end{subfigure}
    \hfill
    \begin{subfigure}[b]{0.24\textwidth}
        \centering
        \includegraphics[width=\linewidth]{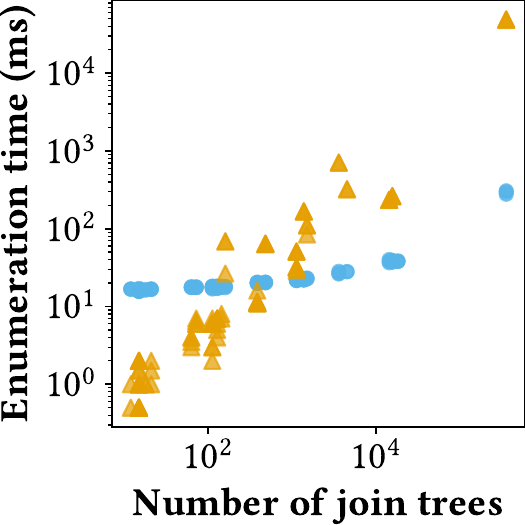}
        \caption{JOB}
        \label{fig:enum-job-original}
    \end{subfigure}
    \hfill
    \begin{subfigure}[b]{0.24\textwidth}
        \centering
        \includegraphics[width=\linewidth]{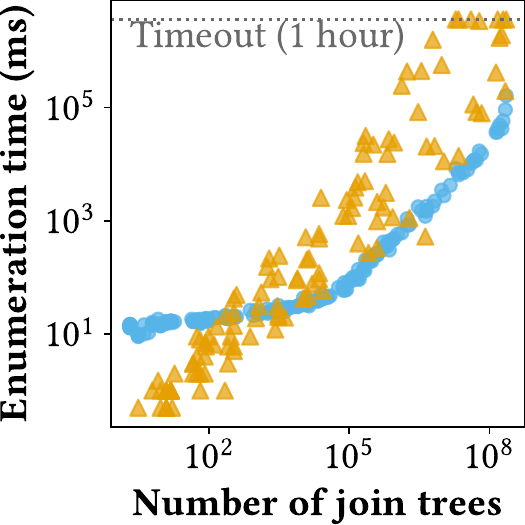}
        \caption{Musicbrainz}
        \label{fig:enum-musicbrainz}
    \end{subfigure}
    \hfill
    \begin{subfigure}[b]{0.24\textwidth}
        \centering
        \includegraphics[width=\linewidth]{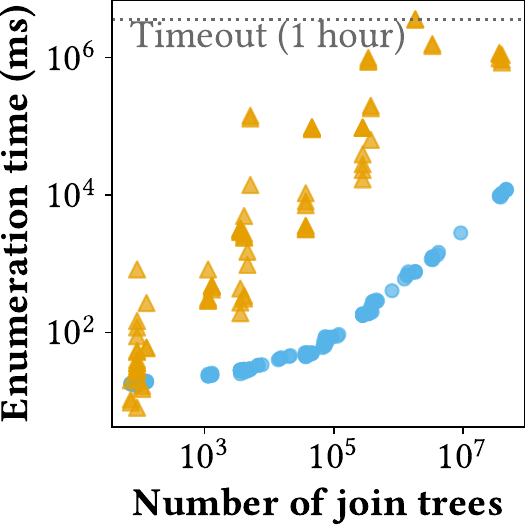}
        \caption{JOBLarge}
        \label{fig:enum-job-large}
    \end{subfigure}
    \vspace{-0.5\baselineskip}
    \caption{Scatter plots of join tree enumeration time using \sys{} versus using na\"ive GYO reduction}
    \label{fig:enum-time}
\end{figure}
Finally, we demonstrate how meta-decompositions can help other recently developed structure-based optimization techniques.
A crucial step in those techniques is selecting a good join tree, as it can significantly influence the structures and, consequently, the costs of the generated query plans.

Thus, we first evaluate \sys{}'s efficiency in \emph{enumerating join trees}, compared to the standard GYO reduction–based approach from \cite{dai_sparksql_2023}. 
As shown in \cref{fig:enum-time}, \sys{} shows up to \emph{four orders of magnitude of speedups} over the traditional approach. In particular, for many highly complex queries in the Musicbrainz and JOBLarge benchmarks with more than $10^7$ join trees, the traditional approach already takes over $1$ hour, whereas \sys{} can complete the enumeration in $2$ seconds to $1$ minute.  

A more effective approach is to apply cost-based optimization directly on meta-decompositions, as has already been done in \sys{}, since this eliminates the need to enumerate exponentially large numbers of join trees. Therefore, in our second experiment, we use Yannakakis$^+$~\cite{wang_yannakakis_2025} to rewrite queries based on join trees selected by \sys{} that induce the optimal width-1 query plans. \emph{This modification achieves an average speedup of {\rm 1.11x} on JOB queries.} This demonstrates that structure-based evaluation algorithms, such as Yannakakis$^+$, can benefit significantly from our efficient cost-based optimization framework.

%% file: figs/exp/tab-benchmark-stats.tex
\begin{footnotesize}
\begin{table}
    \centering
    \begin{tabular}{cc|c|c|c|c|c}
        \hline
        \multicolumn{2}{r|}{\bf{Benchmarks $\rightarrow$}}               & \bf{DSB} & \bf{JOB} & \bf\makecell{Musicbrainz} & \bf\makecell{JOBLarge} & \bf\makecell{Overall} \\ \hline
        \multirow{3}{*}{\bf\shortstack{Number of\\relations}}  & Min     & 5        & 4        & 2                            & 8                         & 2 \\
                                                               & Median  & 6        & 8        & 10                           & 18                        & 8 \\
                                                               & Max     & 9        & 17       & 26                           & 34                        & 34 \\ \hline
        \multirow{3}{*}{\bf\shortstack{Fan-out}}               & Min     & 2        & 2        & 1                            & 2                         & 1 \\
                                                               & Median  & 3        & 3        & 4                            & 4                         & 3 \\
                                                               & Max     & 4        & 5        & 9                            & 5                         & 9 \\ \hline
        \multirow{3}{*}{\bf\shortstack{Number of\\join trees}} & Min     & 5        & 12       & 2                            & 72                        & 2 \\
                                                               & Median  & 7        & 144      & 22500                        & 40000                     & 126 \\
                                                               & Max     & 54       & 352512   & $> 10^8$                     & $> 10^8$                  & $> 10^8$ \\ \hline
    \end{tabular}
    \caption
    {Structural properties of queries in each benchmark}
    \vspace{-1\baselineskip}
    \label{tab:query-props}
\end{table}
\end{footnotesize}

%% file: figs/exp/tab-results.tex
\begin{table}
    \begin{tiny}
        \renewcommand{\arraystretch}{1.1}
        \setlength\tabcolsep{0.3pt}
        
\begin{tabular}{cr|rrrr|rrrr|rrrr|rrrr}
    \hline
    \multicolumn{2}{r|}{\textbf{Benchmark $\rightarrow$}}
        & \multicolumn{4}{c|}{\textbf{DSB}} & \multicolumn{4}{c|}{\textbf{JOB}} & \multicolumn{4}{c|}{\textbf{Musicbrainz}} & \multicolumn{4}{c}{\textbf{JOBLarge}} \\
    \textbf{Metric $\downarrow$} & \textbf{Method $\downarrow$}
        & Mean & Med. & 95th & 99th & Mean & Med. & 95th & 99th & Mean & Med. & 95th & 99th & Mean & Med. & 95th & 99th \\
    \hline
    \multirow{7}{*}{\bf\rotatebox[origin=c]{90}{Optimization time}}
        & \sys{}  & 0.15 ms & 0.10 ms & 0.32 ms & 1.01 ms & 0.38 ms & 0.26 ms & 0.94 ms & 1.04 ms & 2.62 ms & 0.37 ms & 10.45 ms & 45.05 ms & 0.59 ms & 0.37 ms & 1.67 ms & 4.17 ms \\
        & DPconv  & 0.08 ms & 0.03 ms & 0.36 ms & 0.37 ms & 8.32 ms & 0.16 ms & 38.22 ms & 45.08 ms & 7.37 s & 0.30 ms & 14.62 s & 244.44 s & 42.17 s & 1.26 s & $>$ 5 min & $>$ 5 min \\
        & DuckDB  & 0.32 ms & 0.12 ms & 1.83 ms & 1.87 ms & 1.13 ms & 0.32 ms & 5.58 ms & 8.43 ms & 1.89 ms & 0.47 ms & 11.34 ms & 18.59 ms & 0.62 ms & 0.55 ms & 1.23 ms & 3.81 ms \\
        & UnionDP  & 0.01 ms & 0.007 ms & 0.04 ms & 0.04 ms & 0.10 ms & 0.03 ms & 0.38 ms & 0.45 ms & 0.25 ms & 0.11 ms & 0.64 ms & 0.77 ms & 0.29 ms & 0.24 ms & 0.77 ms & 1.18 ms \\
        & Yannakakis$^+$  & 0.78 ms & 0.70 ms & 1.48 ms & 1.50 ms & 1.44 ms & 1.37 ms & 3.78 ms & 4.59 ms & 2.29 ms & 1.72 ms & 6.43 ms & 8.15 ms & 3.55 ms & 3.27 ms & 6.53 ms & 10.89 ms \\
        & LearnedRew.  & 0.89 s & 0.95 s & 1.03 s & 1.06 s & 0.96 s & 0.94 s & 1.13 s & 1.30 s & 1.11 s & 1.00 s & 1.78 s & 2.14 s & 1.09 s & 1.10 s & 1.29 s & 1.42 s \\
        & LLM-R$^2$  & 3.03 s & 2.92 s & 3.66 s & 4.71 s & 3.75 s & 2.99 s & 3.98 s & 24.96 s & 3.11 s & 3.03 s & 3.57 s & 6.10 s & 3.40 s & 3.11 s & 5.72 s & 6.84 s \\
    \hline
    \multirow{7}{*}{\bf\rotatebox[origin=c]{90}{Execution time}}
        & \sys{}  & 0.01 s & 0.01 s & 0.04 s & 0.04 s & 0.04 s & 0.03 s & 0.08 s & 0.11 s & 1.08 s & 0.52 s & 2.56 s & 13.75 s & 1.02 s & 0.17 s & 2.31 s & 23.84 s \\
        & DPconv  & 0.02 s & 0.01 s & 0.09 s & 0.10 s & 0.04 s & 0.04 s & 0.09 s & 0.11 s & 1.42 s & 0.50 s & 2.65 s & 28.27 s & 17.50 s & 0.24 s & 89.05 s & $>$ 5 min \\
        & DuckDB  & 23.85 s & 0.01 s & $>$ 5 min & $>$ 5 min & 0.04 s & 0.04 s & 0.09 s & 0.13 s & 95.92 s & 0.52 s & $>$ 5 min & $>$ 5 min & 22.46 s & 0.13 s & $>$ 5 min & $>$ 5 min \\
        & UnionDP  & 0.03 s & 0.01 s & 0.10 s & 0.11 s & 0.05 s & 0.04 s & 0.10 s & 0.12 s & 58.75 s & 0.28 s & $>$ 5 min & $>$ 5 min & 25.04 s & 0.15 s & $>$ 5 min & $>$ 5 min \\
        & Yannakakis$^+$  & 0.03 s & 0.02 s & 0.08 s & 0.08 s & 0.06 s & 0.06 s & 0.12 s & 0.15 s & 13.63 s & 0.41 s & 34.61 s & $>$ 5 min & 5.11 s & 0.13 s & 1.52 s & 232.60 s \\
        & LearnedRew.  & 23.95 s & 0.05 s & $>$ 5 min & $>$ 5 min & 26.91 s & 0.15 s & $>$ 5 min & $>$ 5 min & 42.47 s & 0.42 s & $>$ 5 min & $>$ 5 min & 42.03 s & 0.14 s & $>$ 5 min & $>$ 5 min \\
        & LLM-R$^2$  & 23.79 s & 0.01 s & $>$ 5 min & $>$ 5 min & 0.04 s & 0.04 s & 0.10 s & 0.14 s & 96.32 s & 0.53 s & $>$ 5 min & $>$ 5 min & 29.94 s & 0.12 s & $>$ 5 min & $>$ 5 min \\
    \hline
    \multirow{7}{*}{\bf\rotatebox[origin=c]{90}{\shortstack[c]{Overall evaluation\\time}}}
        & \sys{}  & 0.01 s & 0.01 s & 0.04 s & 0.04 s & 0.04 s & 0.03 s & 0.08 s & 0.11 s & 1.09 s & 0.52 s & 2.65 s & 13.75 s & 1.02 s & 0.17 s & 2.31 s & 23.84 s \\
        & DPconv  & 0.02 s & 0.01 s & 0.09 s & 0.10 s & 0.05 s & 0.04 s & 0.10 s & 0.12 s & 8.78 s & 0.51 s & 26.36 s & 246.63 s & 58.26 s & 4.32 s & $>$ 5 min & $>$ 5 min \\
        & DuckDB  & 23.85 s & 0.01 s & $>$ 5 min & $>$ 5 min & 0.04 s & 0.04 s & 0.09 s & 0.13 s & 95.93 s & 0.52 s & $>$ 5 min & $>$ 5 min & 22.46 s & 0.13 s & $>$ 5 min & $>$ 5 min \\
        & UnionDP  & 0.03 s & 0.01 s & 0.10 s & 0.11 s & 0.05 s & 0.04 s & 0.10 s & 0.12 s & 58.75 s & 0.28 s & $>$ 5 min & $>$ 5 min & 25.04 s & 0.15 s & $>$ 5 min & $>$ 5 min \\
        & Yannakakis$^+$  & 0.03 s & 0.02 s & 0.08 s & 0.08 s & 0.06 s & 0.06 s & 0.13 s & 0.15 s & 13.63 s & 0.41 s & 34.61 s & $>$ 5 min & 5.11 s & 0.13 s & 1.52 s & 232.60 s \\
        & LearnedRew.  & 24.84 s & 1.01 s & $>$ 5 min & $>$ 5 min & 27.87 s & 1.12 s & $>$ 5 min & $>$ 5 min & 54.93 s & 1.43 s & $>$ 5 min & $>$ 5 min & 43.12 s & 1.28 s & $>$ 5 min & $>$ 5 min \\
        & LLM-R$^2$  & 26.82 s & 3.03 s & $>$ 5 min & $>$ 5 min & 3.79 s & 3.05 s & 4.01 s & 24.99 s & 99.43 s & 4.06 s & $>$ 5 min & $>$ 5 min & 33.34 s & 3.41 s & $>$ 5 min & $>$ 5 min \\
    \hline
    \multirow{6}{*}{\bf\rotatebox[origin=c]{90}{\shortstack[c]{Overall speedup\\[-2pt]of metaDecomp\\[-2pt]over ...}}}
        & DPconv  & 1.22x & 0.99x & 2.90x & 3.29x & 1.18x & 1.08x & 2.77x & 3.11x & 1.34x & 0.99x & 15.49x & 138.59x & 27.61x & 21.23x & 2308.55x & 5014.14x \\
        & DuckDB  & 2.75x & 0.97x & 16938.94x & 18945.33x & 1.11x & 1.07x & 2.03x & 3.53x & 6.39x & 1.86x & 238.77x & 901.19x & 1.46x & 0.82x & 118.01x & 1872.97x \\
        & UnionDP  & 1.28x & 1.03x & 3.38x & 6.57x & 1.12x & 1.04x & 2.09x & 3.01x & 2.67x & 0.94x & 293.61x & 844.52x & 1.91x & 0.94x & 606.52x & 2328.21x \\
        & Yannakakis$^+$  & 1.83x & 1.91x & 2.77x & 7.03x & 1.61x & 1.50x & 3.75x & 5.07x & 1.07x & 0.72x & 47.66x & 656.33x & 0.76x & 0.73x & 2.13x & 22.05x \\
        & LearnedRew.  & 142.29x & 123.70x & 16993.51x & 19006.30x & 59.91x & 38.45x & 7958.89x & 11031.32x & 14.49x & 11.44x & 442.62x & 967.06x & 13.65x & 8.85x & 2444.33x & 7240.85x \\
        & LLM-R$^2$  & 379.88x & 282.20x & 17095.75x & 19117.48x & 103.68x & 98.29x & 484.18x & 832.41x & 65.93x & 78.95x & 927.52x & 2488.64x & 23.98x & 21.60x & 174.23x & 3632.91x \\
    \hline
\end{tabular}

    \end{tiny}
    \caption{Statistics of query evaluation time. The "mean" values are arithmetic mean for time and geometric mean for speedup factors.}
    \vspace{-1\baselineskip}
    \label{tab:stats-new}
\end{table}

%% file: figs/exp/tab-misestimation.tex
\begin{table}
    \begin{footnotesize}

\begin{tabular}{cc|cccc|cccc}
    \hline
    \multicolumn{2}{r|}{\textbf{Benchmark $\rightarrow$}}
        & \multicolumn{4}{c|}{\textbf{DSB}} & \multicolumn{4}{c}{\textbf{JOB}} \\
    \textbf{Method $\downarrow$} & \textbf{Cardinality $\downarrow$}
        & Mean & Median & 95th & 99th & Mean & Median & 95th & 99th \\
    \hline
    \multirow{2}{*}{\sys{}}
        & Exact & 0.01 s & 0.01 s & 0.04 s & 0.04 s & 0.04 s & 0.03 s & 0.08 s & 0.11 s \\
        & Misestimated & 0.02 s & 0.01 s & 0.06 s & 0.12 s & 0.05 s & 0.04 s & 0.11 s & 0.15 s \\
    \hline
    \multirow{2}{*}{DPconv}
        & Exact & 0.02 s & 0.01 s & 0.09 s & 0.10 s & 0.04 s & 0.04 s & 0.09 s & 0.11 s \\
        & Misestimated & 0.03 s & 0.01 s & 0.11 s & 0.14 s & 0.46 s & 0.05 s & 0.23 s & 2.07 s \\
    \hline
    \multirow{2}{*}{UnionDP}
        & Exact & 0.03 s & 0.01 s & 0.10 s & 0.11 s & 0.05 s & 0.04 s & 0.10 s & 0.12 s \\
        & Misestimated & 0.03 s & 0.01 s & 0.11 s & 0.14 s & 0.48 s & 0.05 s & 0.22 s & 1.86 s \\
    \hline
\end{tabular}

    \end{footnotesize}
    \caption{Query execution time when optimized with exact vs misestimated cardinalities}
    \label{tab:misestimation}
\end{table}

%% file: 7_conclusions.tex
\section{Conclusions and Future Work}
\label{sec:conclusions}

In this paper, we have introduced novel techniques for query optimization based on structural properties via meta-decompositions. They offer promising opportunities to integrate theoretically desirable structure-based optimization into practical database systems to efficiently evaluate large, complex queries.
A natural next step is to extend our width notion and meta-decompositions-based framework to cyclic queries. While width-1 query plans cannot exist for cyclic queries (according to \cref{thm:acyclicwidth}) we may, e.g., identify minimal-width query plans and represent all minimal-width hypertree decompositions.
In addition, to improve our current simple strategy of selection pushdown, it is possible to incorporate existing techniques for evaluating conjunctive queries with comparisons~\cite{SIGMOD22}. Similar is the case for many other operators that are so far excluded in this paper, such as aggregations~\cite{AJAR, FAQ}, top-$k$~\cite{wang2023relational}, differences~\cite{SIGMOD23, PODS24}, etc., for which there already exist theoretically desirable structure-guided approaches. The performance of such methods can also benefit from our cost-based optimization framework.
Finally, even though our technique has shown its robustness, it will still benefit from more accurate cardinality estimation. Recent research~\cite{LPBound, Cai2019CE} shows promising techniques that can be adopted to develop efficient, reliable cardinality estimators suitable within our proposed framework.

%% file: app_glossary.tex
\section{Glossary of Frequently-Used Terms and Notations in the Paper}
\label{app:glossary}
\cref{tab:glossary} gives a glossary of frequently-used terms and notations in the paper, each with a brief intuitive explanation.
\begin{small}
\begin{table}[H]
    \centering
    \begin{tabular}{>{\raggedright\arraybackslash}p{3.5cm}>{\raggedright\arraybackslash}p{9.5cm}}
        \hline
        {\bf Term / Notation} & {\bf Explanation} \\ \hline
        $f(S)$ on a set $S$ of elements & Set of all $f(a)$ for all elements $a$ in $S$. \\ \hline
        $\cup S$ on a set $S$ of sets & Union of all sets in $S$. \\ \hline
        Conjunctive query & A query with joins, with selection and projection involving only attributes in base relations. \\ \hline
        Full (conjunctive) query & Conjunctive query where no attribute is projected out. \\ \hline
        Boolean (conjunctive) query & Conjunctive query whose output contains no attribute, but only a Boolean value indicating whether there exists any tuple satisfying all predicates. \\ \hline
        Relation-dominated query & Conjunctive query in which all output attributes are present in some single relation. \\ \hline
        Hypergraph $H(Q)$ associated with query $Q$ & Hypergraph where each vertex is an attribute, and each hyperedge corresponds to a relation, containing all attributes in the relation. \\ \hline
        $\lambda$-label $\lambda(p)$ & In a hypertree decomposition or meta-decomposition, maps each tree vertex $p$ to a set of hyperedges (relations). Non-empty in hypertree decompositions, possibly empty in minor nodes of meta-decompositions. \\ \hline
        $\chi$-label $\chi(p)$ & In a hypertree decomposition or meta-decomposition, maps each tree vertex $p$ to a set of hypergraph vertices (attributes). Subset of vertices in the $\lambda$-label. \\ \hline
        $\kappa$-label $\kappa(p)$ & In a meta-decomposition, maps each tree vertex to the ``interface'', i.e., hypergraph vertices in the intersection between $p$ and other tree vertices not on the subtree rooted at $p$. \\ \hline
        $T_p$, subtree of $T$ rooted at $p$ & Subtree of $T$ containing $p$ and its descendants. \\ \hline
        $T_{p \to q}$ & If $q$ is a child of $p$, same as $T_q$. If $q$ is a parent of $p$, the subtree of $T$ with all vertices except those in $T_p$. \\ \hline
        Fan-out $f(p)$, $f(T)$ & For a tree vertex $p$, the number of children. For a tree $T$, the maximum number of children per vertex. \\ \hline
        Induced query $Q(p)$ & For vertices $p$ on a join tree or query plan, the join of all relations on the subtree rooted at $p$, with appropriate selection and projection applied. \\ \hline
        Interface $I(p)$ & For a vertex $p$ on a meta-decomposition or query plan, the intersection (of vertices) between $p$ and other vertices not in the subtree rooted at $p$. \\ \hline
        Width $w(p)$, $w(\mathcal{P})$ of (vertex on) query plan & For a vertex $p$ on a query plan, the number of relations needed to cover the attributes in the interface. \newline For an entire query plan $\mathcal{P}$, the maximum number of relations needed to cover the attributes in the interface of any intermediate step. \\ \hline
        Width $w(p)$, $w(\mathcal{T})$ of (vertex on) hypertree decomposition & For a vertex $p$ on a hypertree decomposition, the number of hyperedges / relations in the $\lambda$-label. \newline For a hypertree decomposition $\mathcal{T}$, the maximum number of hyperedges / relations in the $\lambda$-label of a single vertex. \\ \hline
        Ear, reducible hyperedge/relation & In GYO reduction on an acyclic hypergraph/query, a hyperedge/relation whose intersection (hypergraph vertices / query attributes) with the rest of the hyperedges/relations is covered by some other single hyperedge/relation. \\ \hline
        $o(e, H)$ & Set of overlapping vertices of hyperedge $e$ and the remaining hyperedges $E(H) \setminus \set{e}$. \\ \hline
    \end{tabular}
    \caption{Glossary of frequently-used terms and notations in the paper}
    \label{tab:glossary}
\end{table}
\end{small}

%% file: app_hierarchical.tex
\section{Omitted Proofs in Section~\ref{sec:hierarchical-query-plan}}
\label{app:hier}

\subsection{Proof of Theorem~\ref{thm:w1plan-size-bound}}

Let $p^* \in V(\mathcal{P})$ be a node in the plan where the maximum width is achieved, i.e., $w(p^*) = w(\mathcal{P})$.  By the definition of $w(p^*)$, there must exist a set of relations $S \subseteq E(H_{p^*})$ such that $|S| = w(p^*)$ and $I(p^*) \subseteq \bigcup_{e \in S} e$.  Therefore, the result of $Q(p^*)$ must be a subset of the join of all relations in $S$.  In the worst-case scenario, where all relations in $S$ share no common attributes and their join is a Cartesian product, the size of this join is $O(N^{|S|})$, which gives an upper bound of $O(N^{w(\mathcal{P})})$.

\subsection{Proof of Theorem~\ref{thm:w1plan-jt}}
    \paragraph{$(\Leftarrow)$}
    Let $T$ be the join tree of the query $Q$ that induces the query plan $\mathcal{P}$. Then, for any $q \in V(T)$, the interface $I_q$ between the induced hypergraph $H_q$ for subtree $T_q$ rooted at $q$ and the residual hypergraph $\overline{H}_q$ satisfies $I_q \subseteq q$, because of the connectedness condition of the join tree.  In addition, by \cref{def:w1plan-induced-by-jt}(1), we have $Q(p) = Q(q)$. This implies that the induced hypergraphs $H_p = H_q$. Therefore, for every $p \in \mathcal{P}$, we can find a set $S = \{q\}$ that covers $I_p$, and $w(p) = 1$ holds as $|S| = 1$.

    \paragraph{$(\Rightarrow)$}
    Given a width-1 query plan $\mathcal{P}$, we construct such a join tree $T$ by induction.  
    For each leaf node $p \in V(\mathcal{P})$, we simply construct a tree with the relation corresponding to the leaf node, and this is vacuously a join tree.
    For each join operation $p \in V(\mathcal{P})$,  assume that there exist valid join trees $T_{c_1}$ and $T_{c_2}$ for each of the two child nodes $c_1$ and $c_2$, repsectively.  Let $q_1 \in V(T_{c_1})$ be the relation that covers the interface $I_{c_1}$ and $q_2 \in V(T_{c_2})$ covers $I_{c_2}$, we make $q_1$ the root of $T_{c_1}$ and $q_2$ the root of $T_{c_2}$, then $I_{c_1} \cap I_{c_2} = q_1 \cap q_2$, and we can merge $T_1$ and $T_2$ to obtain $T'$ by connecting $q_1$ and $q_2$.  
    The new tree $T'$ still preserves the connectedness condition, since, for all $q_1' \in V(T_{c_1})$ and all $q_2' \in V(T_{c_2})$, $q_1$ and $q_2$ are on the path between $q_1'$ and $q_2'$, and therefore $\chi(q_1') \cap \chi(q_2') \subseteq I_{c_1} \subseteq \chi(q_1)$ and $\chi(q_1') \cap \chi(q_2') \subseteq I_{c_2} \subseteq \chi(q_2)$. Then, by the connectedness condition within each of the two subtrees $T_{c_1}$ and $T_{c_2}$, we have that $\chi(q_1') \cap \chi(q_2') \subseteq \chi(q_1') \cap \chi(q_1) \subseteq \chi(q_1'')$ for all $q_1''$ on the path between $q_1'$ and $q_1$, and $\chi(q_1') \cap \chi(q_2') \subseteq \chi(q_2') \cap \chi(q_2) \subseteq \chi(q_2'')$ for all $q_2''$ on the path between $q_2'$ and $q_2$.
    Therefore, $T'$ is a valid join tree for the query induced by the node $p$.
    
\subsection{Proof of \cref{thm:opt-hierarchical-time}}

Following \cref{thm:w1plan-jt}, we have

\begin{lemma}
    For all width-1 query plans $\mathcal{P}$ induced by a join tree $T$,
    for all non-leaf nodes $t \in V(T)$ with distinct children $c_1$ and $c_2$,
    if $p_1, p_2 \in V(\mathcal{P})$ are the nodes on $\mathcal{P}$ such that $Q(p_1) = Q(c_1)$ and $Q(p_2) = Q(c_2)$,
    then there exists no $q \in \mathcal{P}$ such that $Q(q) = Q(p_1) \bowtie Q(p_2) = Q(c_1) \bowtie Q(c_2)$.
\end{lemma}

\begin{proof}
    Suppose there is such a $q$, then by \cref{def:w1plan-induced-by-jt}, there would have been an edge between $c_1$ and $c_2$ on $T$, contradicting with the assumption that $c_1$ and $c_2$ are distinct children of some node $t$.
\end{proof}

This means that, given a width-1 query plan $\mathcal{P}$ induced by a join tree $T$, for each $t \in V(T)$, with children $c_1, \dots, c_{f(t)}$ and the corresponding $p \in V(\mathcal{P})$ such that $Q(t) = Q(p)$, if we consider the sub–query plans $\mathcal{P}_{c_1}, \dots, \mathcal{P}_{c_{f(t)}}$, induced by the sub–join trees $T_{c_1}, \dots, T_{c_{f(t)}}$ rooted at the children of $t$, each as a single relation, then the sub–query plan $\mathcal{P}_p$ rooted at $p$ is a \emph{left-deep plan} consisting of $t$, $\mathcal{P}_{c_1}, \dots, \mathcal{P}_{c_{f(t)}}$, and $t$ is an operand of the deepest join operation.

Therefore, the optimization can be done by a dynamic programming algorithm, as shown in \cref{alg:opt-hierarchical}. In the algorithm, $\mathsf{children}(v)$ is the set of all children of $v$.

\begin{algorithm}
\caption{Finding the optimal width-1 query plan induced by a join tree} \label{alg:opt-hierarchical}
\SetKwInOut{Input}{input}
\SetKwInOut{Output}{output}
\SetKwInOut{Known}{known}
\Input{A join tree $T$}
\Output{The optimal width-1 query plan induced by $T$}

\ForEach {$v \in V(T)$, in bottom-up order \label{alg-line:opt-hierarchical-traversal}} {
    \ForEach{child $c$ of $v$} {
        \If {$\lambda(c) \neq \emptyset$}{
            $P(v, \{ c \}) \gets P(c, \mathsf{children}(c) \cup \{ c \})$ \;
        } \Else {
            $P(v, \{ c \}) \gets P(c, \mathsf{children}(c))$ \;
        }
    }
    \If{$\lambda(v) \neq \emptyset$}{
        $P(v, \{ v \}) \gets $ query plan with one node $v$ \;
    }
    \If {$\lambda(v) = \emptyset$} {
        $S \gets 2^{\mathsf{children}(v)} \setminus \{ \emptyset \}$ \;
    } \Else {
        $S \gets \{ T \cup \{ v \} : T \in 2^{\mathsf{children}(v)} \setminus \emptyset \}$ \;
    }
    \ForEach {$T \in S$, ordered by increasing size \label{alg-line:opt-hierarchical-subsets}} {
        $P(v, T) \gets \arg\!\min_{p \in \{ P(v, T \setminus \{ u \}) \bowtie P(v, \{ u \}) : u \in T \}} \mathsf{cost}(p)$ \label{alg-line:opt-hierarchical-last} \;
    }
}
\Return {$P(\mathsf{root}(T), S(\mathsf{root}(T)))$} \;

\end{algorithm}

In the algorithm, the outer loop (\cref{alg-line:opt-hierarchical-traversal}) is executed $|V(T)| = |Q|$ times. The complexity of the loop body is dominated by enumerating all possible $T$ (\cref{alg-line:opt-hierarchical-subsets}) and $u \in T$ (\cref{alg-line:opt-hierarchical-last}). Noting that the number of children of $v$ is bounded by the fan-out $f(p)$, the complexity of this part is given by
\begin{align*}
    \sum_{i = 1}^{f(v)} \binom{f(v)}{i} \times i &= \sum_{i = 1}^{f(v)} \frac{f(v)!}{i! (f(v) - i)!} \times i = \sum_{i = 1}^{f(v)} \frac{f(v)!}{(i-1)! (f(v) - i)!} \\
    & = \sum_{i = 1}^{f(v)} f(v) \times \frac{(f(v) - 1)!}{(i-1)! ((f(v) - 1) - (i-1))!} \\
    & = \sum_{i = 1}^{f(v)} f(v) \times \binom{f(v) - 1}{i - 1} = f(v) \sum_{i = 0}^{f(v) - 1} \binom{f(v) - 1}{i} \\
    & = f(v) 2^{f(v) - 1} \leq \max_{v \in V(T)} f(v) 2^{f(v) - 1} = f(T) 2^{f(T) - 1}.
\end{align*}

So, overall, the complexity is $O(f(T) 2^{f(T) - 1} |Q|)$.

%% file: app_meta.tex
\section{Omitted Details in Section~\ref{sec:meta}}
\label{app:meta}

\subsection{Omitted Details on \cref{alg:enum}}
\label{app:enum}

\subsubsection{Detailed Version of \cref{alg:enum}}
A more detailed version of \cref{alg:enum} is shown in \cref{alg:enum-full}. It shows in particular the detailed steps to construct join trees from sub--join trees. Some functions being used are detailed in the following subsections.

\begin{algorithm}
\small
\caption{Enumerating all join trees based on a meta-decomposition} \label{alg:enum-full}
\SetInd{0.5em}{0.6em}
\SetKwInOut{Input}{input}
\SetKwInOut{Output}{output}
\SetKwInOut{Known}{known}
\SetKwProg{Fn}{Function}{:}{}
\SetKwFunction{enumRec}{$\mathsf{enumRec}$}
\Input{A meta-decomposition $M = (V(M), r(M), E(M), \lambda, \chi, \kappa)$}
\Output{All join trees of the hypergraph induced by $M$}

\Fn{\enumRec{$v$: a node on $M$}}{

$\mathcal{C} \gets $ set of children $c$ of $v$ on $M$, sorted by partial order $\supseteq$ on $\kappa(c)$ \;
\ForEach{$c \in C$} {
    $\mathcal{T}_c \gets \bigcup_{T \in \mathsf{enumRec}(c)} \mathsf{rerootings}(T, \emptyset, \kappa(c))$ \label{alg-line:enum-subtrees} \;
}
\If(\tcp*[f]{minor node}){$\lambda(v) = \emptyset$ \label{alg-line:enum-minor-cond}} { %
    $\mathcal{C}_{\sf origin} \gets \{ c \in \mathcal{C} \mid \kappa(c) = \chi(v) \}$ \tcp*[r]{they can be found at the head of the sorted $\mathcal{C}$}
    \If{$\kappa(v) = \chi(v)$}{
        $\mathcal{T}_P \gets \textsf{enumerateTrees}(\mathcal{C}_{\sf origin} \cup \{ v \})$, all rerooted to $v$ %
        \label{alg-line:enum-minor-dummy}
        \tcp*[r]{All non-isomorphic trees with vertices $\mathcal{C}_{\sf origin} \cup \{ v \}$. The minor node $v$ will eventually be replaced 
        on \cref{alg-line:enum-dummy-replace}
        by the parent of $r$.
        }
    } \Else(\tcp*[f]{$\kappa(v) \neq \chi(v)$, the parent should not participate}) {
        $\mathcal{T}_P \gets \textsf{enumerateTrees}(S)$ \label{alg-line:enum-minor-no-dummy} \;
    }
        \ForEach(\tcp*[f]{A skeleton of a way to connect subtrees}){$T_P \in \mathcal{T}_P$ \label{alg-line:enum-minor-origin-start}} {
            \ForEach {$c \in V(T_P)$} {
                \If {$c$ is a leaf node} {
                    $\mathcal{T}_c' \gets \mathcal{T}_c$ \;
                } \Else {
                    $\mathcal{T}_c' \gets \emptyset$ \;
                }
            }
            \ForEach {non-leaf $c \in V(T_P)$, in bottom-up order,
                $T_c \in \mathcal{T}_c$,
                    child $d$ of $c$ on $T_P$,
                        $T_d \in \bigcup_{T_d \in \mathcal{T}_d}\mathsf{rerootings}(T_d, \emptyset, \kappa(d))$,
                            $u \in V(T_c)$ such that $\kappa(d) \subseteq \chi(u)$}{
                                $T_c' \gets T_c$, with $T_c$ added as a child of $u$ \;
                                Add $T_c'$ to $\mathcal{T}_c'$ \;
                            }
            $\mathcal{T} \gets \mathcal{T} \cup \mathcal{T}_{r(T_P)}'$ \; \label{alg-line:enum-minor-origin-end}
        }
    
    $\mathcal{C}_{\sf proper} \gets \mathcal{C} \setminus \mathcal{C}_{\sf origin}$, preserving the order \;
} \Else {
    $\mathcal{T} \gets $ a set of one tree with one vertex $v$ only \label{alg-line:enum-physical-root} \;
    $\mathcal{C}_{\sf proper} \gets \mathcal{C}$, preserving the order \;
}
\ForEach {child $c \in \mathcal{C}_{\sf proper}$ in order \label{alg-line:enum-children-start}}{
    $\mathcal{T}' \gets \emptyset$ \;
    \ForEach{$T_c \in \mathcal{T}_c$, $T \in \mathcal{T}$} {
        \tcp*[l]{Enumerate all possible $v$ in some $T$ which $T_c$ can be attached to}
        \If(\tcp*[f]{The root of $T_c$ would be a minor node from \cref{alg-line:enum-minor-dummy}. We find a parent of each of its children.}){$\lambda(c) = \emptyset$ and $\kappa(c) = \chi(c)$}{
            \ForEach {child $d$ of the root of $T_c$}{
                $\mathcal{T}'' \gets \emptyset$ \;
                \ForEach{$T' \in \mathcal{T}'$, $u \in T'$ such that $\kappa(d) \subseteq \chi(u)$} {
                    $T'' \gets T'$ with $T_d$ attached as a child of $u$ \label{alg-line:enum-dummy-replace} \;
                    Add $T''$ to $\mathcal{T}''$ \;
                }
                $\mathcal{T}' \gets \mathcal{T}''$ \;
            }
        } \Else {
            \ForEach{$u \in T$ such that $\kappa(c) \subseteq \chi(u)$}{
                $T' \gets T$ with $T_c$ added as a child of $u$ \label{alg-line:enum-attach-subtree} \;
                Add $T'$ to $\mathcal{T}'$
            }
        }
    }
    $\mathcal{T} \gets \mathcal{T}'$ \label{alg-line:enum-children-end}
}

\Return{$\mathcal{T}$} \;
    
}

\Return{$\bigcup_{T \in \mathsf{enumRec}(r)} \mathsf{rerootings}(T, \emptyset, \emptyset)$ } \;
\end{algorithm}

\subsubsection{Enumeration of Non-Isomorphic Trees Based on \prufer{} Sequences}

For minor nodes $p$, as is illustrated by \cref{ex:star}, all of its neighbors $q$ for which $\kappa(q) = \chi(p)$ can be connected in an arbitrary tree structure with these neighbors being vertices. In order to enumerate all such possibilities without repetition, we (1) enumerate all non-isomorphic trees up to rerooting and (2) enumerate all rerootings for each tree enumerated in step (1).
We start with the enumeration of non-isomorphic trees up to rerooting---namely, if two trees are simply rerootings of each other, they are treated as the same and only one of them will be enumerated. This can be done based on \prufer{} sequences~\cite{prufer_neuer_1918}. A \prufer{} sequence is a sequence of $|V| - 2$ numbers, each of which is an integer between $1$ and $|V|$, inclusive. It can be shown that each sequence bijectively encodes one tree (up to rerooting) of size $|V|$, with vertices labeled by integers ranging from $1$ to $|V|$, inclusive. Therefore, we use \cref{alg:prufer} which simply enumerates all such numeric sequences and converts each sequence into a tree.

\begin{algorithm}[H]
\caption{$\textsf{enumerateTrees}(V)$: Enumerate all non-homomorphic trees, up to rerooting, with vertices $V$, using \prufer{} sequences}\label{alg:prufer}
\SetKwInOut{Input}{input}
\SetKwInOut{Output}{output}
\SetKwInOut{Known}{known}
\Input{Set of vertices $V$}
\Output{A set $\mathcal{T}$ of all non-isomorphic trees with vertices $V$}

Arbitrarily order the vertices in $V$ and let $v_i$ be the $i$-th vertex in $V$ under this ordering \;

$\mathcal{T} \gets \emptyset$ \;

\ForEach{sequence $P = p_1, \dots, p_{|V| - 2}$ of $|V|-2$ integers, each with value in range $[1, |V|]$}{
    \ForEach{$i = 1, \dots, |V|$}{
        $d(i) \gets 1$ \;
    }
    $S \gets [1, |V|]$ \tcp*[r]{maintains the set of indices $i$ with $d(v_i) = 1$, in increasing order}
    \ForEach{$i = 1, \dots, |V| - 2$}{
        $d(s_i) \gets d(p_i) + 1$ \;
        $S \gets S \setminus \{i\}$ \;
    }
    Initialize a tree $T$ with set of vertices $V$ and no edge \;
    \ForEach{$i = 1, \dots, |V| - 2$}{ %
        $j \gets$ remove the least element in $S$ \; %
        $d(i) \gets d(i) - 1$ \;
        $d(j) \gets d(j) - 1$ \;
        \If{$d(i) = 1$}{
            Add $i$ into $S$ \;
        }
        In $T$, add $v_j$ as a child of $v_i$ \;
    }
    $i, j \gets $ the remaining two elements in $S$ \;
    In $T$, add $v_j$ as a child of $v_i$ \tcp*[r]{the direction does not matter}
    $\mathcal{T} \gets \mathcal{T} \cup \{ T \}$
}
\Return{$\mathcal{T}$} \;
\end{algorithm}

\subsubsection{Enumeration of Rerootings}

We then recursively enumerate all possible rerootings of a tree $T$ such that the new root $r'$ has $\kappa \subseteq \chi(r')$, where $\kappa$ is the set of keys as required by the original \cref{alg:enum-full}. This can be done as shown in \cref{alg:reroot}.

\begin{algorithm}[H]
\caption{$\mathsf{rerootings}(T, p, \kappa)$: Enumerate all valid rerootings of a tree $T$ such that the new root $r'$ has $\kappa \subseteq \chi(r')$}
\label{alg:reroot}
\SetKwInOut{Input}{input}
\SetKwInOut{Output}{output}
\SetKwInOut{Known}{known}
\Input{A tree $T$, optionally the previous parent $p$ of the root $r$ of $T$ before the previous rotation, the expected key $\kappa$}
\Output{A set $\mathcal{T}$ of all valid rerootings of $T$ such that the new root $r'$ has $\kappa \subseteq \chi(r')$}

$\mathcal{T} \gets \emptyset$ \;
\ForEach{child $c$ of the root $r$ of $T$}{
    \If{$\kappa \subseteq \chi(c)$ and $c \neq p$ \label{alg-line:reroot-kappa-cond}}{
        $T' \gets T$ rerooted to $c$ \;
        $\mathcal{T} \gets \mathcal{T} \cup \{ T' \} \cup \textsf{rerootings}(T', r, \kappa)$ \;
    }
}
\Return{$\mathcal{T}$} \;
\end{algorithm}

\subsection{Proof of \cref{lem:add-back-minor}}

We start with the following lemma:

\begin{lemma} \label{lem:minor-overlap}
    Given an acyclic hypergraph $H$, if there exist two ears $e_1, e_2 \in E(H)$ such that $o(e_1, H) = o(e_2, H)$, then $o(e_1, H) = o(e_2, H) = e_1 \cap e_2$.
\end{lemma}

\begin{proof}
    ($\subseteq$): $o(e_1, H) = e_1 \cap (\cup (E(H) \setminus \{ e_1 \})) \subseteq e_1$, and, similarly, $o(e_2, H) \subseteq e_2$. So it follows that $o(e_1, H) = o(e_2, H) \subseteq e_1 \cap e_2$.

    ($\supseteq$): Since $e_1 \in E(H) \setminus \{ e_2 \}$, we have $o(e_1, H) = o(e_2, H) = e_2 \cap (\cup (E(H) \setminus \{ e_2 \})) \supseteq e_1 \cap e_2$.
\end{proof}

    A hypergraph is acyclic if and only if it has a join tree, satisfying conditions \ref{H1}--\ref{H3}.
    Since $e_1$ and $e_2$ are ears in $H$ such that $o(e_1, H) = o(e_2, H) \subseteq e_2$, in the GYO algorithm for $H$, $e_1$ can already be removed, with $e_2$ being its witness.
    Therefore, there exists a join tree $T$ of $H$ where there is a leaf node $p_1$ with $\lambda(p_1) = \{ e_1 \}$, and its parent $p_2$ has $\lambda(p_2) = \{ e_2 \}$. We construct a tree $T'$ that is exactly the same as $T$ except that $p_1$ and $p_2$ are merged into a single node $m$ with $\lambda(m) = \{ o(e_1, H) \}$ and $\chi(m) = o(e_1, H)$, and we show that $T'$ is a valid join tree of $H'$.
    \ref{H1} and \ref{H3} are satisfied by construction.
    It remains to show that $T'$ satisfies \ref{H2}.
    Suppose by way of contradiction that this is not the case.
    By \ref{H2} on the join tree $T$, this means, on $T'$, either (1) there exist $s, t \in V(T')$ such that $m$ is on the path between $s$ and $t$, and $\chi(s) \cap \chi(t) \not\subseteq \chi(m)$, or (2) there exists $s \in V(T')$ and $t$ on the path between $s$ and $m$, such that $\chi(s) \cap \chi(m) \not\subseteq \chi(t)$.
    
    For case (1), since, by \cref{lem:minor-overlap}, $\chi(m) = o(e_1, H) = e_1 \cap e_2 = \chi(p_1) \cap \chi(p_2)$,
    if $\chi(s) \cap \chi(t) \not\subseteq \chi(m)$,
    this means either $\chi(s) \cap \chi(t) \not\subseteq \chi(p_1)$
    or $\chi(p) \cap \chi(q) \not\subseteq \chi(p_2)$.
    On $T$, since $p_1$ is a leaf node, $p_2$ should be on the path between $s$ and $t$.
    By \ref{H2} for $T$, we have $\chi(s) \cap \chi(t) \subseteq \chi(p_2)$.
    So it can only be that $\chi(s) \cap \chi(t) \not\subseteq \chi(p_1)$, i.e., there exists some $v \in \chi(s) \cap \chi(t) \subseteq \chi(p_2)$ such that $v \not\in \chi(p_1)$.
    But then, $v \in o(e_2, H)$ but $v \not\in o(e_1, H)$, contradicting $o(e_1, H) = o(e_2, H)$.
    
    For case (2), we first note that $\chi(s) \cap o(e_1, H) = \chi(s) \cap e_1$. This is because ($\subseteq$) $o(e_1, H) = e_1 \cap (\cup (E(H) \setminus \{ e_1 \})) \subseteq e_1$, and ($\supseteq$) $\chi(s) \cap e_1 \subseteq (\cup (E(H) \setminus \{e_1 \})) \cap e_1 = o(e_1, H)$. Then, $\chi(s) \cap \chi(m) = \chi(s) \cap o(e_1, H) = \chi(s) \cap e_1 = \chi(s) \cap \chi(p_1)$. So, $\chi(s) \cap \chi(p_1) = \chi(s) \cap \chi(m) \not\subseteq \chi(p_2)$. Since $p_1$ is a leaf and $p_2$ is the parent of $p_1$ in $T$, $p_2$ must be on the path between $s$ and $p_1$, so this contradicts \ref{H2} for $T$.

\subsection{Proof of \cref{thm:meta-complexity}}

    We start by considering the while-loop (\cref{alg:meta-while}). At the start of each iteration, an ear must exist because the hypergraph $H'$ must be acyclic, as guaranteed by \cref{lem:add-back-minor}. Furthermore, $E(H')$ always decreases in size by at least 1 in each iteration, because we only remove ears from $E(H')$, except when we have mutually reducible ears, in which case we remove all (at least 2) those ears sharing the same intersection but add only one special hyperedge $e_m$ to $E(H')$.
    Therefore, the body of the while loop is executed at most $|E(H)|$ times.
    
    For each iteration of the while-loop, all sets $S$ (\cref{alg-line:meta-minor-origins-sets}) can be enumerated in $O(|E(H')|)$, if we maintain two hash maps: (1) $T_1$, which maps each $o \subseteq V(H')$ to all hyperedges $e \in E(H')$ such that $o(e, H') = o$, and (2) $T_2$, which maps each $v \in V(H')$ to the number of hyperedges $e \in E(H')$ such that $v \in e$.
    Each time a hyperedge $e$ is added or removed from $E(H')$, we update $T_2$ for each $v \in e$, and then update $T_1$ if some $v \in e$ is/was covered by exactly one hyperedge. These updates can be done in $O(|E(H')|)$, which is $O(|E(H)|)$.
    
    With these observations, the runtime $O(|E(H)|^3)$ can be obtained by following the control flow.

\subsection{Proof of \cref{thm:meta-correct}}
We show that $M$ output by the algorithm satisfies \ref{H1}--\ref{H5}.

\paragraph{\ref{H1}.} For each $e \in E(H)$, the algorithm has to remove $e$ from $E(H')$ on \cref{alg-line:meta-origin-remove} or \cref{alg-line:meta-ear-remove} before it terminates. In any of these cases, the algorithm will create a node $p$ in $M$ with $\lambda(p) = \{ e \}$, on \cref{alg-line:meta-origin-node} or \cref{alg-line:meta-ear-node}.

Before proceeding with \ref{H2}, we have the following lemmas:

\begin{lemma} \label{lem:meta-alg-edge-correct}
    For each edge $(q, p) \in E(M)$, where $\chi(p) = \{ e_p \}$ and $\chi(q) = \{ e_q \}$, we have either
    \begin{enumerate}
        \item $e_p$ is removed as ear by the algorithm before $e_q$ is, \emph{or}
        \item $\lambda(q) = \emptyset$ and $\kappa(q) = \chi(q)$
    \end{enumerate}
    Additionally, $\kappa(p) = \chi(p) \cap \chi(q)$.
\end{lemma}

\begin{proof}
    We first have $\kappa(p) \subseteq \chi(p)$ by definition on Lines~\ref{alg-line:meta-origin-node}, \ref{alg-line:meta-ear-node}, and \ref{alg-line:no-minor-create}.
    Since the algorithm adds the edge $(q, p)$, by Lines~\ref{alg-line:minor-origin-edge-cond} and \ref{alg-line:phys-find-parent}, we have $\kappa(p) \subseteq \chi(q)$.
    Therefore, $\kappa(p) \subseteq \chi(p) \cap \chi(q)$.
    
    We then consider the relative order in which $e_p$ and $e_q$ are removed from $E(H')$ in the algorithm.
    
    (1) If $e_p$ is removed before $e_q$ is, then at the time $e_p$ is removed, $e_q \in E(H')$, and therefore $\chi(p) \cap \chi(q) = e_p \cap e_q \subseteq e_p \cap (\cup (E(H') \setminus \{ e_p \})) = \kappa(p)$. Since also $\kappa(p) \subseteq \chi(p) \cap \chi(q)$, we have $\kappa(p) = \chi(p) \cap \chi(q)$.
    
    (2) If $e_p$ is removed after $e_q$ is, then, symmetric to case (1), we have $\chi(p) \cap \chi(q) \subseteq \kappa(q)$. Then, since we have shown that $\kappa(p) \subseteq \chi(p) \cap \chi(q)$, we have $\kappa(p) \subseteq \chi(p) \cap \chi(q) \subseteq \kappa(q)$.
    By \cref{alg-line:minor-origin-edge-cond} and \cref{alg-line:phys-find-parent}, this is only possible if $\lambda(q) = \emptyset$ and $\kappa(p) = \chi(q)$.
    Since $\kappa(p) \subseteq \chi(p)$ and $\kappa(p) = \chi(q)$, it is true that $\kappa(p) = \chi(p) \cap \chi(q)$.
\end{proof}

It follows that

\begin{corollary} \label{cor:descendent-removed-before}
    $p$ is a descendant of $q$ on $M$ if and only if $e_p$ is removed as ear by the algorithm before $e_q$ is.
\end{corollary}

\paragraph{\ref{H2}.} 
The algorithm iteratively adds edges $(q, p)$ to attach a connected subtree $M_p$ rooted at $p$ to a connected subtree $M^q$ that contains $q$.
Suppose by way of contradiction that \ref{H2} is violated at the end, and
suppose it is when processing $p \in V(M)$, that the addition of edge $(q, p)$ violates this condition for the first time.
Namely, $M_p$ and $M^q$ each satisfies \ref{H2}, but the resulting $M$ violates \ref{H2}, i.e., for some $s, t \in V(M)$, there exists $u$ on the path betwween $s$ and $t$ such that $\chi(s) \cap \chi(t) \not\subseteq \chi(u)$.
Without loss of generality, we take the shortest possible path from $s$ to $t$ where the condition is violated, where $u$ is the node immediately before $t$ on the path. Then, $\kappa(t) = \chi(t) \cap \chi(u)$.

Since $M_p$ and $M^q$ each satisfies \ref{H2}, the only way it can be violated is if $s \in V(M_p)$, i.e., $s$ is a descendent of $p$, $t \in M^q$, and thus $p$ and $q$ are on the path between $s$ and $t$.
For the purpose of the proof, we think of the subtree $M_p$ as being added to the subtree $M^q$ one node at a time, in a top-down order, starting from the root $p$. We then show that each time a node is added, the connectedness condition is preserved.
Let $\lambda(p) = \{ e_p \}$ and $\lambda(t) = \{ e_t \}$.

As the base case, when $p$ is added as a child of $q$, we consider two cases on the relative order in which $e_p$ and $e_t$ are removed from $E(H')$ by the algorithm.

(1) If $e_p$ is removed as the ear by the algorithm before $e_t$ is, then $\chi(p) \cap \chi(t) = e_p \cap e_t \subseteq \kappa(p) \subseteq \chi(q)$. And, for all $u$ between $q$ and $t$ on the path between $p$ and $t$, we have $\chi(p) \cap \chi(t) = \chi(p) \cap \chi(t) \cap \chi(t) \subseteq \chi(q) \cap \chi(t) \subseteq \chi(u)$, by the connectedness condition on $M^q$.

(2) If $e_t$ is removed as the ear by the algorithm before $e_p$ is, then $\chi(p) \cap \chi(t) = e_p \cap e_t \subseteq \kappa(t)$.
By \cref{cor:descendent-removed-before}, $t$ is a descendant of $q$, and so the parent $u$ of $t$ is on the path between $p$ and $t$.
Then, $\chi(p) \cap \chi(t) \subseteq \kappa(t) = \chi(t) \cap \chi(u) \subseteq \chi(u)$

As the induction step, assuming that the condition holds for the connected subtree with nodes in $M^q$ and all $p'$ that are added before $s$. Following a similar argument as for the base case, substituting $q$ with the parent of $s$, we can also derive that $\chi(s) \cap \chi(t) \subseteq \chi(u)$ for all $u$.

\paragraph{\ref{H3'}.} This is guaranteed by construction on Lines~\ref{alg-line:meta-origin-node}, \ref{alg-line:meta-ear-node}, and \ref{alg-line:no-minor-create}.

\paragraph{\ref{H4}(a).}
By \cref{lem:meta-alg-edge-correct}, for each $p \in V(M)$ and its parent $q$, we have $\kappa(p) = \chi(p) \cap \chi(q)$.
($\subseteq$): $\kappa(p) = \chi(p) \cap \chi(q) \subseteq \chi(q)$.
($\supseteq$): For all $s \in V(M) \setminus V(M_p)$, since $q$ is the parent of $p$, $q$ has to be on the path between $p$ and $s$. Then, by the connectedness condition \ref{H2} on $M$, we have $\chi(p) \cap \chi(s) \subseteq \chi(q)$. Therefore, $\kappa(p) = \chi(p) \cap \chi(q) \supseteq \chi(p) \cap \left( \bigcup_{s \in V(M) \setminus V(M_p)} \chi(s) \right) = \chi(p) \cap \chi(V(M) \setminus V(M_p))$.

\paragraph{\ref{H4}(b).} This is guaranteed by Lines~\ref{alg-line:minor-origin-edge-cond} and \ref{alg-line:phys-find-parent}.

\paragraph{\ref{H5}.} Lines~\ref{alg-line:meta-minor-origins-sets}--\ref{alg-line:meta-add-special-edge} and Lines~\ref{alg-line:no-minor-check}--\ref{alg-line:no-minor-create} create such minor nodes when necessary. 
The condition when creating a minor node on \cref{alg-line:meta-origin-node}, \cref{alg-line:meta-ear-node}, and \cref{alg-line:no-minor-create} ensures there are no duplicate minor nodes.

\subsection{Proof of \cref{thm:meta-size}}

    Vertices $V(M)$ of $M$ are either physical nodes or minor nodes. Let $V^P(M)$ be the set of physical nodes and $V^M(M)$ be the set of minor nodes, such that $V(M) = V^P(M) \cup V^M(M)$.
    Each time we create a physical node $p$ with $\lambda(p) = \{ e \}$ (Lines~\ref{alg-line:meta-origin-node} and \ref{alg-line:meta-ear-node}), we remove $e$ from $E(H')$ (Lines~\ref{alg-line:meta-origin-remove} and \ref{alg-line:meta-ear-remove}).
    So there are at most $|E(H)|$ physical nodes, i.e., $|V^P(M)| \leq |E(H)|$.
    Now let us consider the number of minor nodes, $|V^M(M)|$.
    Note that each minor node $m \in V^M(M)$ must have at least two children (see \cref{alg-line:meta-minor-origins-sets} and \cref{alg-line:no-minor-check}), i.e., the fanout $f(m) \geq 2$.
    Then, the total number of vertices
    \begin{align*}
        |V(M)| &= |V^P(M)| + |V^M(M)| = |E(M)| + 1 \\
        & = 1 + \sum_{p \in V(M)} f(p)
        > 1 + \sum_{m \in V^M(M)} f(m) \geq 1 + \sum_{m \in V^M(M)} 2 = 2|V^M(M)| + 1.
    \end{align*}
    Then, $|V^P(M)| + |V^M(M)| > 2|V^M(M)| + 1$, so $|V^M(M)| < |V^P(M)| - 1 = |E(H)| - 1$.
    So, $|V(M)| = |V^P(M)| + |V^M(M)| < |E(H)| + |E(H)| - 1 = 2|E(H)| - 1$.
    And $|E(M)| = |V(M)| - 1 \leq 2|E(H)| - 2$.

\subsection{Proof of \cref{the:correctness}}

    We prove by showing that all trees $T$ output by \cref{alg:enum} satisfy \ref{H1}, \ref{H2}, and \ref{H3}.

    For \ref{H1}, note that this condition must already hold for the meta-decomposition $M$. In \cref{alg:enum}, all physical nodes, for which $\lambda(r) \neq \emptyset$, will be kept in $\mathcal{T}$, on \cref{alg-line:enum-physical-root}.
    Minor nodes, for which $\lambda(r) = \emptyset$, will no longer exist in the output. If $\kappa(r) \neq \chi(r)$, trees in $\mathcal{T}_p$, and thus $\mathcal{T}$, will not contain minor nodes (\cref{alg-line:enum-minor-no-dummy}) but only nodes in the subtrees rooted at the children of $r$ (Lines~\ref{alg-line:enum-minor-origin-start}--\ref{alg-line:enum-minor-origin-end}). If $\kappa(r) = \chi(r)$, the minor node $r$ on \cref{alg-line:enum-minor-dummy} will eventually be removed and replaced by $v$ on \cref{alg-line:enum-dummy-replace}. Therefore, the trees output by \cref{alg:enum} contain only physical nodes. \ref{H3} then follows, by \ref{H3'} for $M$.

    It now remains to show \ref{H2}, the connectedness condition.
    We prove by structural induction on vertices $v \in V(M)$, i.e., for each $v \in V(M)$, all trees $T$ returned by the function $\mathsf{enumRec}$ satisfy \ref{H2}.
    In the base case, for all leaf nodes $v$, $v$ has to be physical, and the function simply returns a tree with one single node, so these conditions hold.
    For non-leaf nodes $v$, we assume as induction hypothesis that for all children $c$ of $v$, all trees returned by the function call $\mathsf{enumRec}(c)$ satisfy \ref{H2}. And we now show that it remains satisfied throughout the operations in $\mathsf{enumRec}(v)$.
    
    We start with the trees in $\mathcal{T}$ constructed on Lines~\ref{alg-line:enum-minor-cond}--\ref{alg-line:enum-physical-root}.
    If $\lambda(v) \neq \emptyset$, this vacuously holds as there is only one tree with one vertex $r$ in $\mathcal{T}$ (\cref{alg-line:enum-physical-root}).
    If $\lambda(v) = \emptyset$, we will show that it still holds for all $T_c'$ (constructed on Lines~\ref{alg-line:enum-minor-origin-start}--\ref{alg-line:enum-minor-origin-end}), assuming that it holds for $T_c \in \mathcal{T}_c$ from the previous iteration. 
    It suffices to show that, for each $(T_P, c, T_c, d, T_d, u)$ processed in the loop,
    for all $p \in T_c$ and $q \in T_d$, all $s$ on the path between $p$ and $q$ has $\chi(p) \cap \chi(q) \subseteq \chi(s)$.
    Let the root of $T_c$ be $r_c$ and the root of $T_d$ be $r_d$.
    The path from $p$ to $q$ needs to go through $u$ then $r_d$.
    For $s$ on the path between $p$ and $u$, we have $\chi(p) \cap \chi(q) \subseteq \chi(v) = \kappa(d) \subseteq \chi(u)$.
    Since $T_c$ satisfies \ref{H2}, we have $\chi(p) \cap \chi(q) \subseteq (\chi(p) \cap \chi(q)) \cap \chi(q) \subseteq \chi(u) \cap \chi(q) \subseteq \chi(s)$.
    For $s$ on the path between $r_d$ and $q$, we have $\chi(p) \cap \chi(q) \subseteq \chi(v) = \kappa(d) \subseteq \chi(r_d)$.
    Since $T_d$ satisfies \ref{H2}, we have $\chi(p) \cap \chi(q) \subseteq (\chi(p) \cap \chi(q)) \cap \chi(q) \subseteq \chi(r_d) \cap \chi(q) \subseteq \chi(s)$.
    
    We now continue to show that lines~\cref{alg-line:enum-children-start}--\cref{alg-line:enum-children-end} preserve the condition \ref{H2} for all $T \in \mathcal{T}$, by showing that the condition is preserved for all $T' \in \mathcal{T}'$ after each child $c \in X$ is processed, assuming that all the previous $T \in \mathcal{T}$ satisfy \ref{H2}.
    
    Let the root of $T_c$ be $r_c$.
    By the connectedness condition on $M$, we have $\chi(p) \cap \chi(q) \subseteq \chi(v)$, and $\chi(p) \cap \chi(q) \subseteq \chi(c)$. Since $c$ is a child of $v$, $\chi(p) \cap \chi(q) \subseteq \chi(v) \cap \chi(c) = \kappa(c)$.
    
    (1) If $\lambda(c) \neq \emptyset$ or $\kappa(c) \neq \chi(c)$, $T_c$ is added directly as a child of $u$ on $T'$ (\cref{alg-line:enum-attach-subtree}).
    For all $p \in V(T')$ and $q \in V(T_c)$, the path from $p$ to $q$ needs to go through $u$ then $r_c$.
    For $s$ on the path between $p$ and $u$, we have $\chi(p) \cap \chi(q) \subseteq \kappa(c) \subseteq \chi(u)$.
    Since $T'$ satisfies \ref{H2}, we have $\chi(p) \cap \chi(q) \subseteq (\chi(p) \cap \chi(q)) \cap \chi(q) \subseteq \chi(u) \cap \chi(q) \subseteq \chi(s)$.
    For $s$ on the path between $r_c$ and $q$, we have $\chi(p) \cap \chi(q) \subseteq \kappa(c) \subseteq \chi(r_c)$.
    Since $T_c$ satisfies \ref{H2}, we have $\chi(p) \cap \chi(q) \subseteq (\chi(p) \cap \chi(q)) \cap \chi(q) \subseteq \chi(r_c) \cap \chi(q) \subseteq \chi(s)$.
    
    (2) If $\lambda(c) = \emptyset$ and $\kappa(c) \neq \chi(c)$, for each child $d$ of the root of $T_c$, we attach $T_d$ as a child of $u$. Let the root of $T_d$ be $r_d$. Then, $\kappa(d) \subseteq \chi(r_d)$ and $\kappa(d) \subseteq \chi(u)$.
    For all $p \in V(T')$ and $q \in V(T_d)$, the path from $p$ to $q$ needs to go through $u$ then $r_d$.
    For $s$ on the path between $p$ and $u$, we have $\chi(p) \cap \chi(q) \subseteq \kappa(c) = \kappa(d) \subseteq \chi(u)$.
    Since $T'$ satisfies \ref{H2}, we have $\chi(p) \cap \chi(q) \subseteq (\chi(p) \cap \chi(q)) \cap \chi(q) \subseteq \chi(u) \cap \chi(q) \subseteq \chi(s)$.
    For $s$ on the path between $r_c$ and $q$, we have $\chi(p) \cap \chi(q) \subseteq \kappa(c) = \kappa(d) \subseteq \chi(r_d)$.
    Since $T_c$ satisfies \ref{H2}, we have $\chi(p) \cap \chi(q) \subseteq (\chi(p) \cap \chi(q)) \cap \chi(q) \subseteq \chi(r_d) \cap \chi(q) \subseteq \chi(s)$.
    
\subsection{Proof of \cref{thm:enum-complete}}

We start from a few lemmas:

\begin{lemma} \label{lem:kappa}
    For any non-root node $p$ on a meta-decomposition $M$ with parent $q$, $\kappa(p) = \chi(p) \cap \chi(q)$.
\end{lemma}

\begin{proof}
\textbf{($\supseteq$)}: Since $q \in V(M) \setminus V(M_p)$, $\chi(p) \cap \chi(q) \subseteq \chi(p) \cap \chi(V(M) \setminus V(M_p)) = \kappa(p)$.
\textbf{($\subseteq$)}: For any $v \in \kappa(p) = \chi(p) \cap \chi(V(M) \setminus V(M_p))$, i.e., any $v \in \chi(p) \cap \chi(s)$ for some $s \in V(M) \setminus V(M_p)$, the parent $q$ of $p$ lies on the path on $M$ between $p$ and $s$. Then, by the connectedness condition \ref{H2} for $M$, $v \in \chi(q)$.
\end{proof}

\begin{lemma} \label{lem:minor-all-origin}
    For all $p \in V(M)$ with $\lambda(p) = \emptyset$ and $\kappa(p) = \chi(p)$, all children $c$ of $p$ satisfy $\kappa(c) = \chi(p)$.
\end{lemma}
\begin{proof}
    Suppose by way of contradiction that there exists such a $p \in V(T)$ with child $c$ such that $\kappa(c) \neq \chi(p)$. By \cref{lem:kappa}, we have $\kappa(c) = \chi(c) \cap \chi(p) \subseteq \chi(p)$. Since $\kappa(p) = \chi(p) \neq \emptyset$, $p$ is not the root of $M$ and thus has a parent $q$. Again by \cref{lem:kappa}, we have $\kappa(p) = \chi(p) \cap \chi(q) \subseteq \chi(q)$. But then, $\kappa(c) \subseteq \chi(p) = \kappa(p) \subseteq \chi(q)$, contradicting \ref{H4}(b) since $\kappa(c) \neq \chi(p)$ but $q \in V(T) \setminus V(T_p)$.
\end{proof}

\begin{lemma} \label{lem:no-consec-special-minor}
    Given a meta-decomposition $M$, if some non-root node $p \in V(M)$ has $\lambda(p) = \emptyset$ and $\kappa(p) = \chi(p)$, then its parent $q$ has either $\lambda(q) \neq \emptyset$ or $\kappa(q) \neq \chi(q)$.
\end{lemma}

\begin{proof}
    Suppose by way of contradiction that, on a meta-decomposition $M$, there is $p \in V(M)$ such that $\lambda(p) = \emptyset$ and $\kappa(p) = \chi(p)$, and its parent $q$ has $\lambda(q) = \emptyset$ and $\kappa(q) = \chi(q)$.
    Then, $\kappa(p) \subseteq \chi(q) = \kappa(q)$. Since $\kappa(q) = \chi(q) \neq \emptyset$, $q$ has a parent as well. Let the parent of $q$ be $s$. Then, $\kappa(q) = \chi(q) \cap \chi(s) \subseteq \chi(s)$. So, $\kappa(p) \subseteq \kappa(q) \subseteq \chi(s)$. By \ref{H5}, $\kappa(p) \neq \kappa(q) = \chi(q)$, contradicting \ref{H4}(b).
\end{proof}

\begin{lemma} \label{lem:subtree-in-jt}
    Given a meta-decomposition $M$ of an acyclic hypergraph $H$, for all $v \in V(M)$,
    \begin{enumerate}[leftmargin=*]
        \item if $\lambda(v) = \emptyset$ and $\kappa(v) = \chi(v)$, then, for each join tree $T$ of $H$, there exists some $S = \{ p \in V(M) \mid \lambda(p) \neq \emptyset, \chi(v) \subseteq \chi(p) \}$ such that the set $V_v^P = \{ p \in V(M_{v}) \mid \lambda(p) \neq \emptyset \} \cup S$ induces a connected subtree $T_v$ of $T$;
        \item if $\lambda(v) \neq \emptyset$ or $\kappa(v) \neq \chi(v)$, then, for each join tree $T$ of $H$, the set $V_v^P = \{ p \in V(M_v) \mid \lambda(p)  \neq \emptyset \}$, i.e., physical nodes on the subtree $M_v$ rooted at $v$, induces a connected subtree $T_v$ of $T$.
    \end{enumerate}
\end{lemma}

\begin{proof}
    We prove by structural induction on $M$.
    In the base case, for leaf nodes $v \in V(M)$, there is only one node on the subtree $M_v$, so the statement is vacuously true.
    For non-leaf nodes $v \in V(M)$, we assume as induction hypothesis that the statement is true for all children $c$ of $v$.
    Let $c$ and $d$ be the children of $v$ such that $p \in V(M_c)$ and $q \in V(M_d)$.
    We now consider two cases of $v$.
    
    (1) If $\lambda(v) = \emptyset$ and $\kappa(v) = \chi(v)$, then, for any $s$ on the path between $p$ and $q$ on $T$ such that $s \notin M_v$, we have $\chi(p) \cap \chi(q) \subseteq \chi(s)$ by the connectedness condition on $T$.
    By (the contrapositive of) \cref{lem:no-consec-special-minor}, $\lambda(c) \neq \emptyset$ or $\kappa(c) \neq \chi(c)$, so by the inductive hypothesis, $V_c^P$ induces a connected subtree $T_c$ of $T$. Similarly, $V_d^P$ induces a connected subtree $T_d$ of $T$.
    If $c = d$, then $p$ and $q$ are in the same subtree rooted at the same child of $v$, and the statement would hold.
    If $c \neq d$,
    on the path from $p$ to $q$ on $T$, let the last node on $T_c$ be $x$ and the first node on $T_d$ be $y$. $s$ is between $x$ and $y$.
    Then, the path from $c$ and $d$ has to go through $x$ and $y$, and hence also $s$.
    Therefore, $\chi(c) \cap \chi(d) \subseteq \chi(s)$.
    By \cref{lem:minor-all-origin}, in this case, $\kappa(v) = \kappa(c) = \kappa(d) = \chi(c) \cap \chi(d) \subseteq \chi(s)$.

    (2) If $\lambda(v) \neq \emptyset$ or $\kappa(v) \neq \chi(v)$, we would like to show that for any $s$ on the path between $p$ and $q$ on $T$, $s \in V(M_v)$.
    Suppose by way of contradiction that this is not the case, i.e., there exists some $s$ on the path between $p$ and $q$ on $T$, but $s \in V(M) \setminus V(M_v)$. By the connectedness condition on $T$, we have $\chi(p) \cap \chi(q) \subseteq \chi(s)$. Furthermore, on $M$, $\chi(p) \cap \chi(s) \subseteq \chi(v)$ and $\chi(q) \cap \chi(s) \subseteq \chi(v)$, so $\chi(p) \cap \chi(q) = \chi(p) \cap \chi(q) \cap \chi(s) \subseteq \chi(v)$.
    Let the parent of $p$ on $M$ be $t$.
    Then, $\chi(t) \cap \chi(s) \subseteq \chi(v)$. Then, $\kappa(p) = \kappa(p) \cap \chi(s) = \chi(p) \cap \chi(t) \cap \chi(s) = (\chi(p) \cap \chi(s)) \cap (\chi(t) \cap \chi(s)) \subseteq \chi(v) \cap \chi(v) = \chi(v)$.
    But, by \ref{H4}, this means either $t = v$ or $t$ is a minor node with $\kappa(t) = \chi(t)$.
    In either case, we have $\chi(t) \cap \chi(q) \subseteq \chi(v)$, because
    (i) if $t = v$, $\chi(t) \cap \chi(q) = \chi(v) \cap \chi(q) \subseteq \chi(v)$, and
    (ii) if $t$ is a minor node with $\kappa(t) = \chi(t)$, then $\chi(t) \cap \chi(q) = \kappa(t) \cap \chi(q) = \kappa(p) \cap \chi(q) \subseteq \chi(p) \cap \chi(q) \subseteq \chi(v)$.
    Also, we have $\kappa(t) = \chi(p) \cap \chi(t) \subseteq \chi(q)$, because either $t$ is on the path between $p$ and $q$, or $p$ is on the path between $t$ and $q$.
    Let the parent of $v$ be $u$. Then, we have $\kappa(t) = \kappa(t) \cap \chi(q) = \chi(t) \cap \chi(q) = (\chi(t) \cap \chi(q)) \cap (\kappa(p) \cap \chi(q)) \subseteq (\chi(t) \cap \chi(q)) \cap (\chi(p) \cap \chi(q)) \subseteq \chi(v) \cap \chi(s) \subseteq \chi(u)$.
    But, since it cannot be that $\kappa(v) = \chi(v)$ in this case, this violates \ref{H4}, contradiction.
\end{proof}

\begin{lemma} \label{lem:enum-complete}
    Given a meta-decomposition $M$ of hypergraph $H$, for all $v \in V(M)$, let $V_v^P = \{ p \in V(M_v) \mid \lambda(p) \neq \emptyset \}$. Then,
    \begin{enumerate}[leftmargin=*]
        \item if $\lambda(v) = \emptyset$ and $\kappa(v) = \chi(v)$, then \cref{alg:enum} enumerates all join trees of the hypergraph $H_v$ with $E(H_v) = \lambda(V_v^P) \cup \{ \chi(v) \}$ and $V(H_v) = \cup E(H_v)$;
        \item if $\lambda(v) \neq \emptyset$ or $\kappa(v) \neq \chi(v)$, then \cref{alg:enum} enumerates all join trees of the hypergrpah $H_v$ with $E(H_v) = \lambda(V_v^P)$ and $V(H_v) = \cup E(H_v)$.
    \end{enumerate}
\end{lemma}

\begin{proof}
    We will show that each join tree $T$ of $H$ can be enumerated by the algorithm based on the meta-decomposition $M$ of $H$.
    We prove by structural induction on vertices $v \in V(M)$. If $v$ is a leaf node, there is only one node in the subtree $M_v$, so this statement is vacuously true. If $v$ is a non-leaf node, consider the following cases.
    
    (1) If $\lambda(v) = \emptyset$ and $\kappa(v) = \chi(v)$, by \cref{lem:no-consec-special-minor}, all children $c$ of $v$ have either $\lambda(c) \neq \emptyset$ or $\kappa(c) \neq \chi(c)$, so, by the induction hypothesis, \cref{alg:enum} enumerates all join trees of hypergraph $H_c$ with $E(H_c) = \bigcup_{p \in V_c^P} \lambda(p)$ and $V(H_c) = \bigcup_{e \in E(H_c)} e$.
    Also, note that it is possible to obtain a meta-decomposition $M_v'$ of the hypergraph $H_v$ from $M_v$ by simply replacing $v$ by a physical node $v'$ with $\lambda(v') = \{ \chi(v) \}$, $\chi(v') = \chi(v)$, and $\kappa(v') = \emptyset$.
    By \cref{lem:subtree-in-jt}, for each child $c$ of $v$, on any join tree $T_c$ of hypergraph $H_c$, $V_c^P$ induces a connected subtree $T_c$ of $T$.
    Furthermore, $T_c$ is a join tree of the meta-decomposition $M_c$, and can hence be enumerated on \cref{alg-line:enum-subtrees}.
    So it remains to show that the algorithm enumerates all possible ways to combine these subtrees along with $v$.
    For each possible join tree $T_v$ of $H_v$, without loss of generality, assume that $T_v$ is rooted at $v'$.
    $T_v$ can be constructed from the algorithm as follows.
    We first construct a tree $T_P$,
    by a top-down traversal of $T_v$, where, for each node $p$,
    if (i) $\chi(v) \subseteq \chi(p)$, (ii) $p \in M_c$, i.e., in the subtree $M_c$ rooted at some child $c$ of $v$, and (iii) its parent $q$ is in a different subtree $M_d$ rooted at some child $d \neq c$ of $v$, then we add an edge from $d$ to $c$.
    If $p = v$, and its parent $q \in V(M_d)$ for some child $d$ of $v$, we add an edge from $d$ to $v$.
    If $p \in V(M_c)$ for some child $c$ of $v$, and its parent $q = v$, we add an edge from $v$ to $c$.
    Note that, for all pairs of $c$ and $d$ where we add edges, $\kappa(c) = \kappa(d) = \chi(v)$.
    Now, this tree $T_P$ is a spanning tree of the clique with vertices $C \cup \{ v \}$, where $C$ is the set of all children of $v$, so it can be enumerated by the \prufer{} sequence on \cref{alg-line:enum-minor-dummy}.
    For each child $c$ of $v$, by \cref{lem:subtree-in-jt}, $V_c^P$ induces a connected subtree $T_c$ of $T_v$.
    We traverse the vertices $c \in V(T_P)$ in bottom-up order.
    For each $c$ and each child $d$ of $c$ on $T_P$, %
    let $r_d$ be the root of the connected subtrees $T_d$ of $T_v$ induced by $V_d^P$, and $u$ be the parent of $r_d$ on $T_v$.
    The algorithm will be able to reroot $T_d$ to $r_d$, find $u$ (since $\kappa(d) = \kappa(v) = \chi(v) = \chi(c) \subseteq \chi(u)$), and attach $T_d$ as a child of $u$.
    Finally, by \cref{lem:minor-all-origin}, in this case, $S = C$ and $X = \emptyset$, so the second part will not be executed at all.

    (2) If $\lambda(v) = \emptyset$ and $\kappa(v) \neq \chi(v)$, starting from Lines~\ref{alg-line:enum-minor-origin-start}--\ref{alg-line:enum-minor-origin-end}, note that it is not possible to have a child $c \in C$ such that $\kappa(c) = \chi(c)$, because if that is the case, $\kappa(c) = \chi(c) = \chi(v)$, violating the uniqueness condition \ref{H5}. Therefore, for all $c \in C$, we have $\lambda(c) \neq \emptyset$ or $\lambda(c) \neq \chi(c)$.
    Furthermore, we claim that the vertices $V_C^P = \bigcup_{c \in C} V_c^P$ induces a connected subtree $T_C$ on any join tree $T_v$ of $M_v$. This is because, if this is not the case, i.e., if there is some $s$ on the path between some $p \in V_c^P$ and $q \in V_d^P$, where $c$ and $d$ are children of $v$ with $\kappa(c) = \kappa(d) = \chi(v)$, such that $\kappa(s) \subsetneq \chi(v)$, then $\chi(p) \cap \chi(q) = \chi(v) \not\subseteq \kappa(s)$, violating the connectedness condition \ref{H2}.
    Therefore, we can follow the same method as in case (1) to construct $T_P$ with vertices $C$ and show that, at the end of this part, it is possible to obtain a partial join tree $T$ with vertices in $V_C^P$.
    
    Now consider all children $c \in X = \{ c \in C \mid \kappa(c) \subsetneq \chi(v) \}$.

    (i) If $\lambda(c) = \emptyset$ and $\kappa(c) = \chi(c)$, by \cref{lem:no-consec-special-minor}, all children $d$ of $c$ has either $\lambda(d) \neq \emptyset$ or $\kappa(d) \neq \chi(d)$, so, by \cref{lem:subtree-in-jt}, $V_d^P$ induces a connected subtree in any join tree. We build a tree $T_c$ by (a) collecting all maximal connected subtrees of the join tree with nodes $p \in V_c^P$ and (b) adding a common parent $c$ as the root node. This is a valid join tree of the hypergraph $H_c$ with $E(H_c) = \left( \bigcup_{p \in V_c^P} \lambda(p) \right) \cup \{ \chi(c) \}$ and $V(H_c) = \bigcup_{e \in E(H_c)} e$, and, by our inductive hypothesis, can be enumerated. For each child $d$ of $c$ on $T_c$, let the parent of $d$ on $T$ be $u$. The algorithm will be able to find such $u$ to attach $T_d$, because $\kappa(d) = \chi(d) \cap \chi(c) \subseteq \chi(u)$, by the connectedness condition on $T$.

    (ii) If $\lambda(c) \neq \emptyset$ or $\kappa(c) \neq \chi(c)$, similarly to case (i), by \cref{lem:subtree-in-jt}, $V_c^P$ induces a connected subtree on $T$, which, by our inductive hypothesis, can be enumerated. Then, the algorithm will be able to find such $u$ to attach $T_c$, because $\kappa(c) = \chi(c) \cap \chi(v) \subseteq \chi(u)$, by the connectedness condition on $T$.

    (3) If $\lambda(v) \neq \emptyset$, we can follow a similar argument for the children $c \in X$ in case (2).
\end{proof}

Finally, based on \cref{lem:enum-complete}, taking $v = r$, the root of $M$ (for which $\kappa(r) = \emptyset \neq \chi(r)$), we finish the proof of completeness for the algorithm.
\vfill

%% file: app_cost.tex
\section{Omitted Details in \cref{sec:opt}}
\label{app:opt}

\subsection{Supporting Re-branching for \cref{alg:cost}}

For each node $p$, if there exists a subtree $T_q$ that can be attached as a child of $p$, we simply use an additional bit in the DP table ($\sf plan$) to indicate whether $T_q$ is included. And we call $\sf optimizeLocal$ with or without the relations in $T_q$ accordingly. We note again that this is a very rare possibility in practical queries.

\subsection{Greedy Algorithm for Local Join Ordering}
\label{app:greedy}

The implementation of local join ordering using a greedy algorithm is as shown in \cref{alg:greedy}.

\begin{algorithm}
\caption{Greedy implementation of $\sf optimizeLocal$}
\label{alg:greedy}
\SetKwInOut{Input}{input}
\SetKwInOut{Output}{output}
\SetKwInOut{Known}{known}
\Input{Node $q$ on meta representation, a set of query plans $\mathcal{P}$}
\Output{A query plan that joins the relation in $q$ with all query plans in $\mathcal{P}$}

\If(\tcp*[f]{$\lambda(q) = \emptyset$, there is no relation in $q$}){$q$ is a minor node}{
    $P \gets $ remove $P_1$ from $\mathcal{P}$ with the minimum result cardinality \;
} \Else {
    $P \gets $ the relation in $\lambda(q)$ \;
}

\While{$\mathcal{P} \neq \emptyset$}{
    $P_i \gets $ remove $P_i$ from $\mathcal{P}$ that minimizes the cost to join $P$ with $P_i$ \;
    $P \gets $ the plan $P_i \bowtie P$ \;
}

\Return{$P$} \;

\end{algorithm}

%% file: app_exp.tex
\newpage
\section{Omitted Details in \cref{sec:exp}}
\label{app:exp}

\subsection{Overall Evaluation Time Using \sys{} vs Other Methods on Each Benchmark}

\subsubsection{\sys{} versus DPconv}
\label{app:total-time-dpconv-individual}

\paragraph{Histograms of speedups}
\cref{fig:overall-speedup-dpconv-individual} shows histograms of the speedups of the overall query execution time using \sys{} over using DPconv for queries in each of the four benchmarks. In DSB, JOB, and Musicbrainz, most of the queries show a speedup close to 1. In JOBLarge, however, we observe a large number of significant speedups by multiple orders of magnitude, mostly due to the high optimization time of DPconv for the large query sizes in this benchmark.

\begin{figure}[H]
    \begin{subfigure}[b]{0.49\textwidth}
        \centering
        \includegraphics[width=\linewidth]{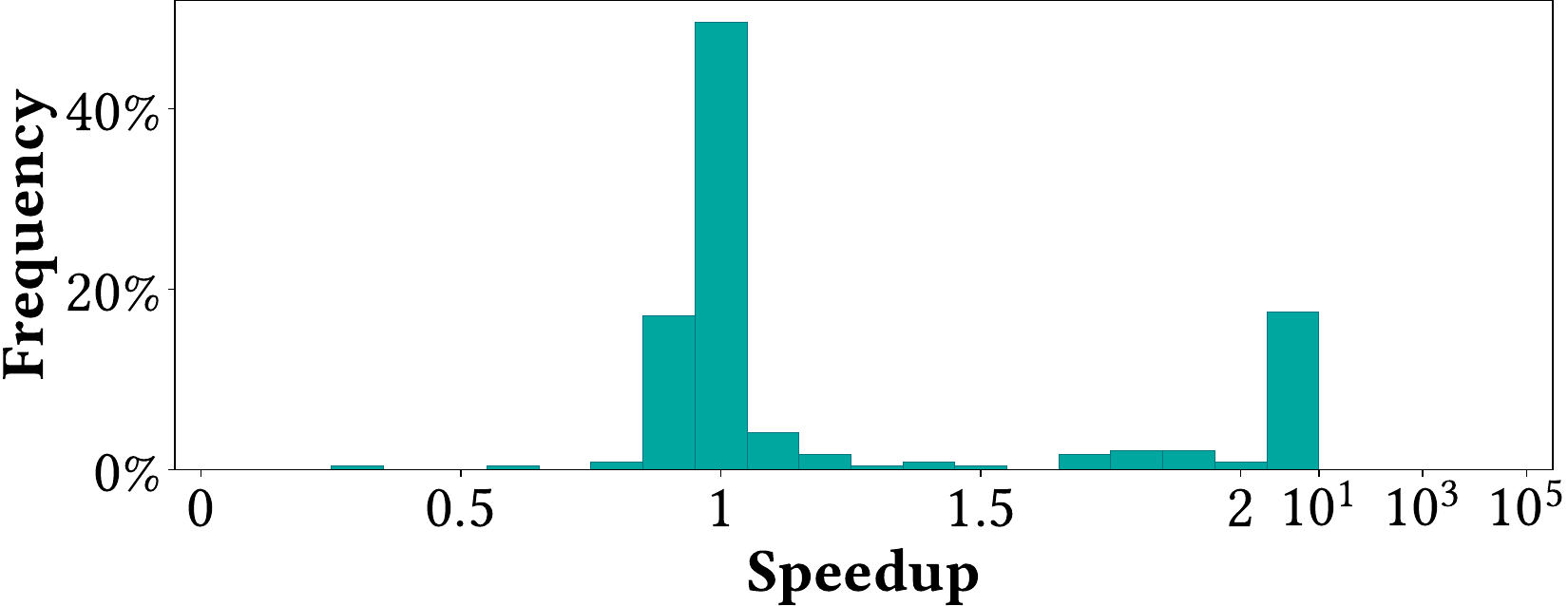}
        \caption{DSB}
        \label{fig:overall-speedup-dpconv-dsb}
    \end{subfigure}
    \hfill
    \begin{subfigure}[b]{0.49\textwidth}
        \centering
        \includegraphics[width=\linewidth]{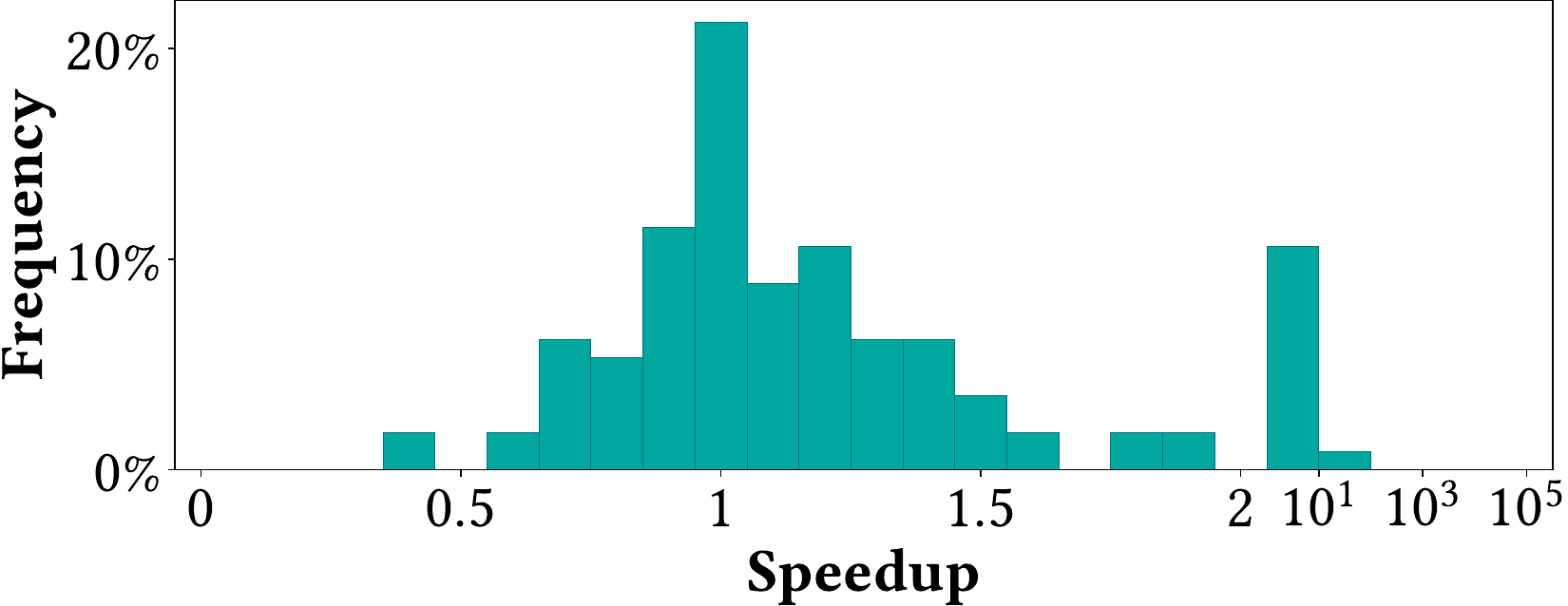}
        \caption{JOB}
        \label{fig:overall-speedup-dpconv-job-original}
    \end{subfigure}
    
    \vspace{0.5\baselineskip}

    \begin{subfigure}[b]{0.49\textwidth}
        \centering
        \includegraphics[width=\linewidth]{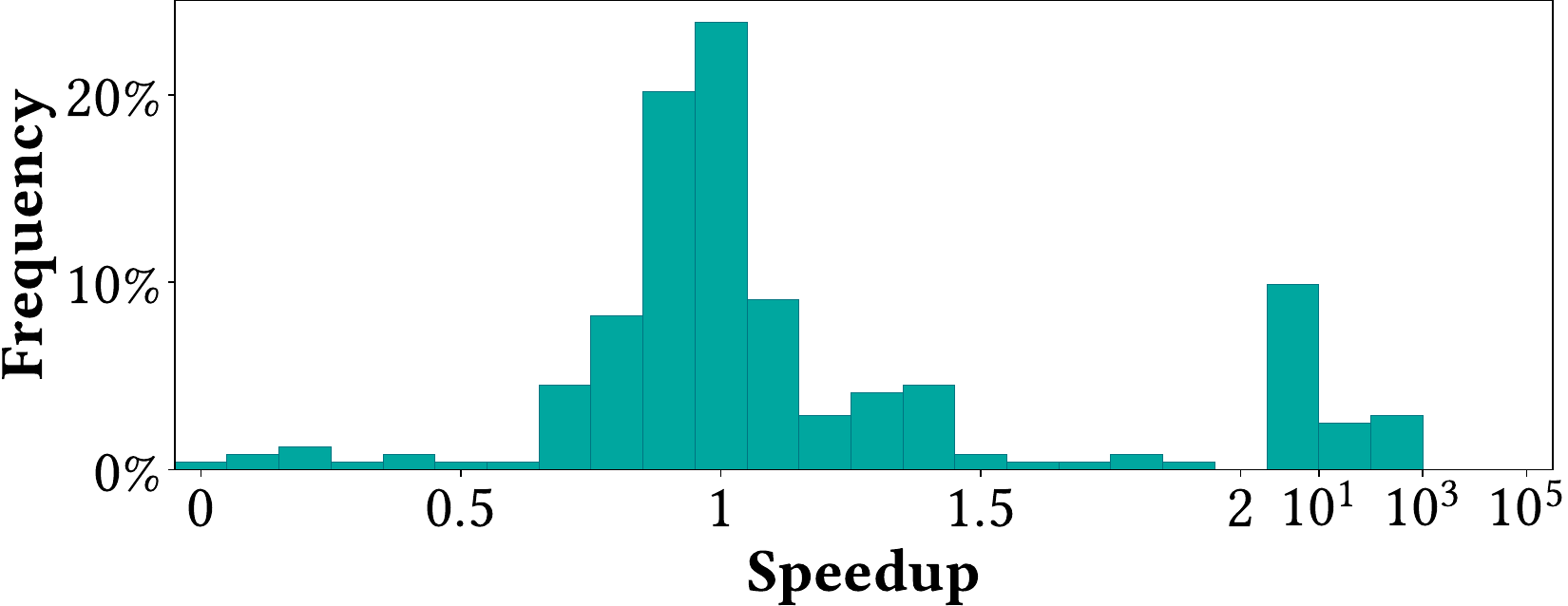}
        \caption{Musicbrainz}
        \label{fig:overall-speedup-dpconv-musicbrainz}
    \end{subfigure}
    \hfill
    \begin{subfigure}[b]{0.49\textwidth}
        \centering
        \includegraphics[width=\linewidth]{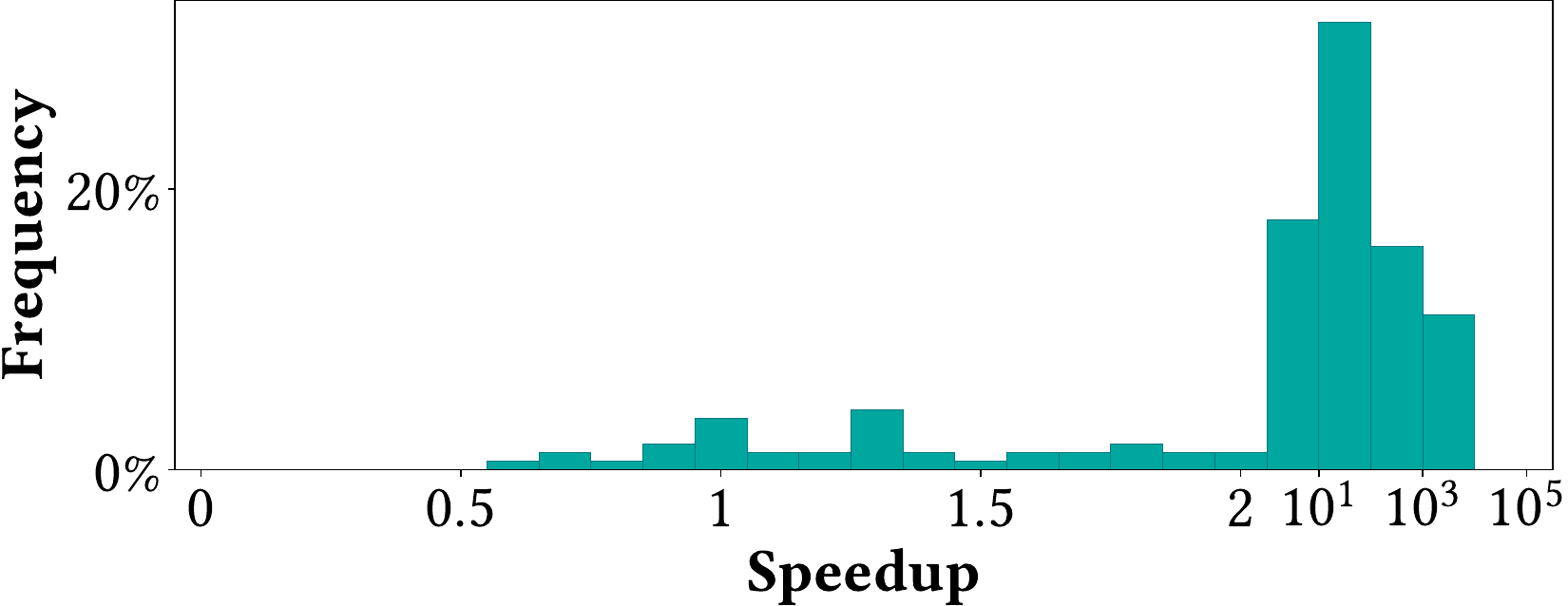}
        \caption{JOBLarge}
        \label{fig:overall-speedup-dpconv-job-large}
    \end{subfigure}
    \vspace{-0.5\baselineskip}
    \caption{Histograms of speedups of overall evaluation time of \sys{} over DPconv on each benchmark}
    \label{fig:overall-speedup-dpconv-individual}
\end{figure}

\paragraph{Scatter plots}
\begin{figure}
    \begin{subfigure}[b]{0.24\textwidth}
        \centering
        \includegraphics[width=\linewidth]{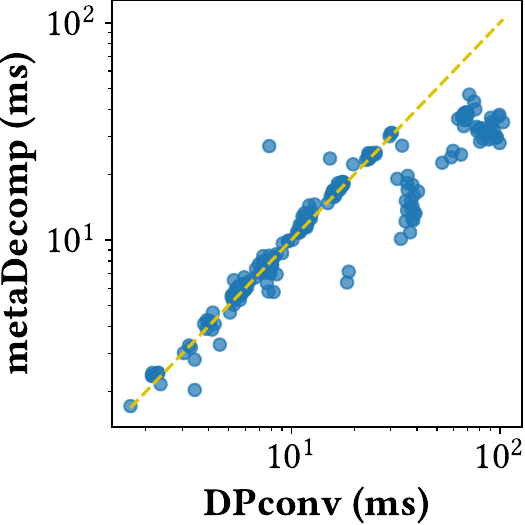}
        \vspace{-1.25\baselineskip}
        \caption{DSB}
        \label{fig:overall-scatter-dpconv-dsb}
    \end{subfigure}
    \hfill
    \begin{subfigure}[b]{0.24\textwidth}
        \centering
        \includegraphics[width=\linewidth]{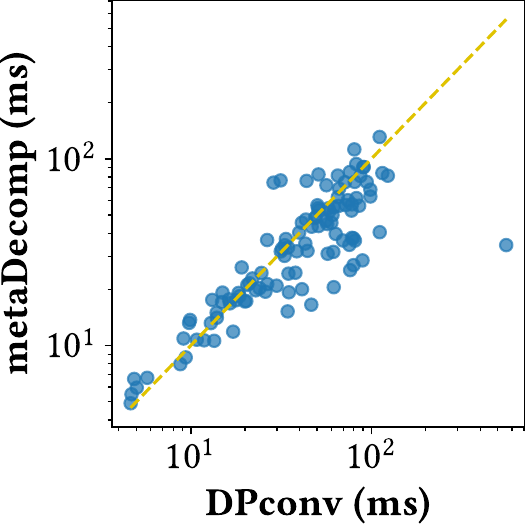}
        \vspace{-1.25\baselineskip}
        \caption{JOB}
        \label{fig:overall-scatter-dpconv-job-original}
    \end{subfigure}
    \hfill 
    \begin{subfigure}[b]{0.24\textwidth}
        \centering
        \includegraphics[width=\linewidth]{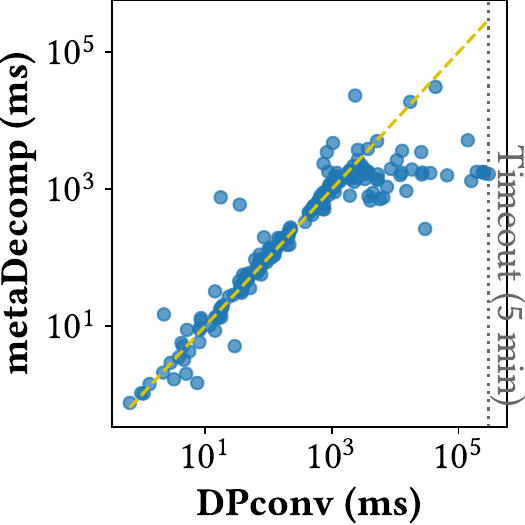}
        \vspace{-1.25\baselineskip}
        \caption{Musicbrainz}
        \label{fig:overall-scatter-dpconv-musicbrainz}
    \end{subfigure}
    \hfill
    \begin{subfigure}[b]{0.24\textwidth}
        \centering
        \includegraphics[width=\linewidth]{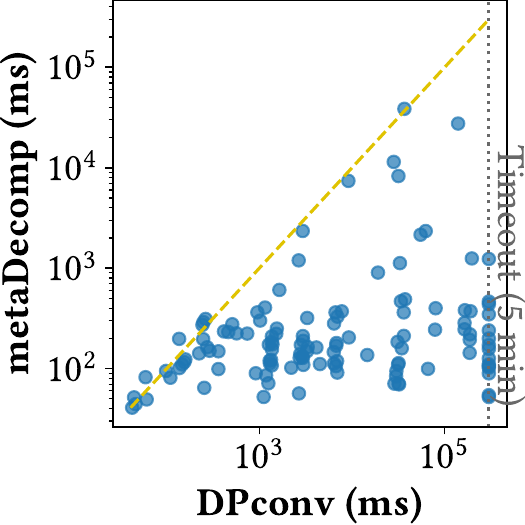}
        \vspace{-1.25\baselineskip}
        \caption{JOBLarge}
        \label{fig:overall-scatter-dpconv-job-large}
    \end{subfigure}
    \vspace{-0.5\baselineskip}
    \caption{Scatter plots of overall evaluation time of using \sys{} versus DPconv on each benchmark}
    \label{fig:overall-scatter-dpconv-individual}
\end{figure}
\cref{fig:overall-scatter-dpconv-individual} shows scatter plots of the overall evaluation time using \sys{} and DPconv, measured for queries in each of the four benchmarks. The trends that can be observed in each of the four benchmarks and the overall figure (\cref{fig:overall-scatter-dpconv}) are mostly consistent. The two approaches have similar performance for queries that takes DPconv less than 1 second ($10^3\ {\rm ms}$). For those taking DPconv more than 1 second, \sys{} starts to show significant speedups. Notably, for complex queries in Musicbrainz (\cref{fig:overall-scatter-dpconv-musicbrainz}) and JOBLarge (\cref{fig:overall-scatter-dpconv-job-large}) for which DPconv times out (taking more than 5 minutes), \sys{} can evaluate them within merely 1 second.

\newpage
\subsubsection{\sys{} versus DuckDB}
\label{app:total-time-metadecomp-duckdb-individual}

\paragraph{Histograms of speedups}
\cref{fig:overall-speedup-duckdb-individual} shows histograms of speedups of the overall evaluation time using \sys{} over DuckDB for queries in each benchmark. Most of the queries show a speedup close to 1, with some notable speedups especially for larger queries in the Musicbrainz benchmark.

\begin{figure}[H]
    \begin{subfigure}[b]{0.49\textwidth}
        \centering
        \includegraphics[width=\linewidth]{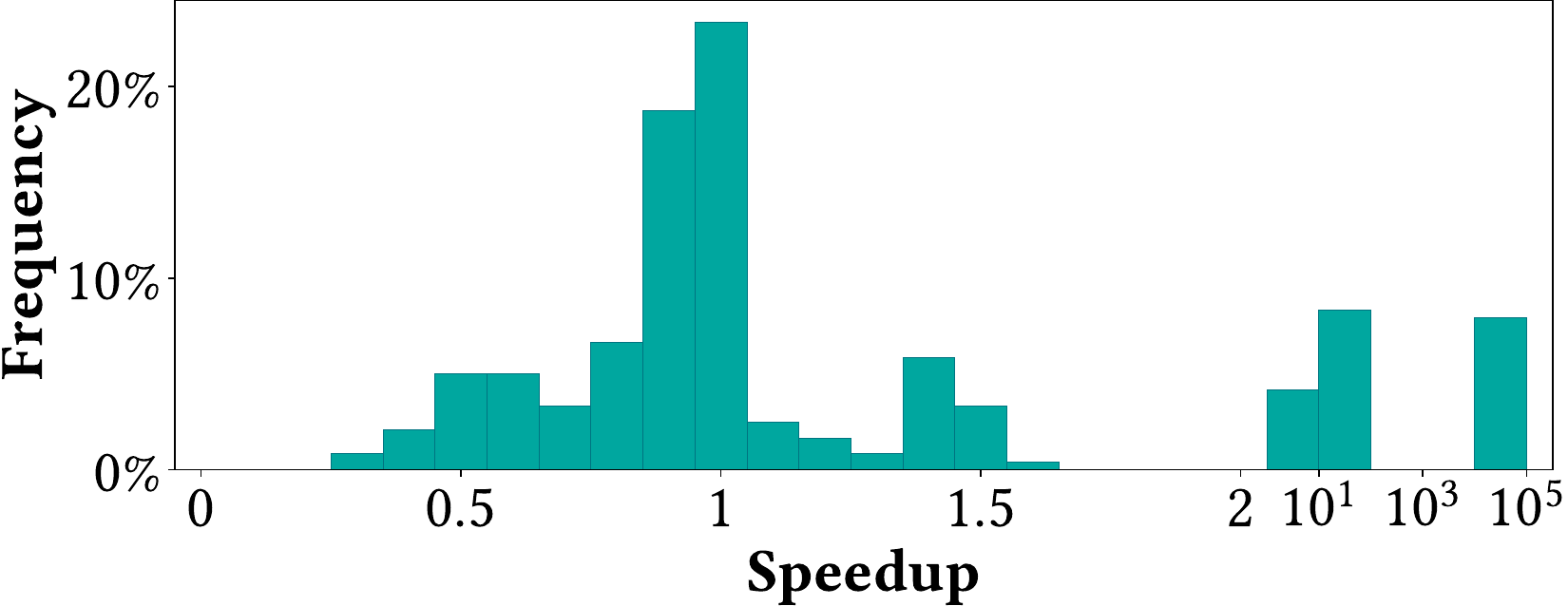}
        \vspace{-1.25\baselineskip}
        \caption{DSB}
        \label{fig:overall-speedup-duckdb-dsb}
    \end{subfigure}
    \hfill
    \begin{subfigure}[b]{0.49\textwidth}
        \centering
        \includegraphics[width=\linewidth]{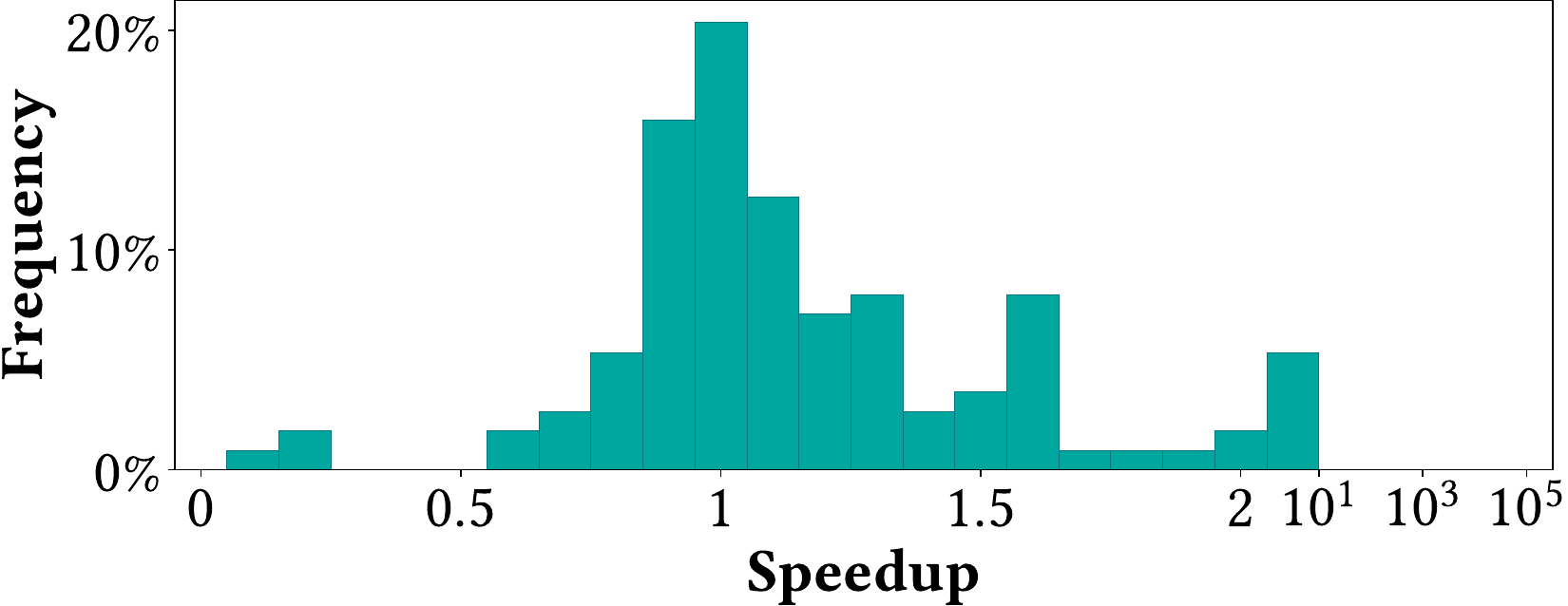}
        \vspace{-1.25\baselineskip}
        \caption{JOB}
        \label{fig:overall-speedup-duckdb-job-original}
    \end{subfigure}

    \vspace{0.5\baselineskip}

    \begin{subfigure}[b]{0.49\textwidth}
        \centering
        \includegraphics[width=\linewidth]{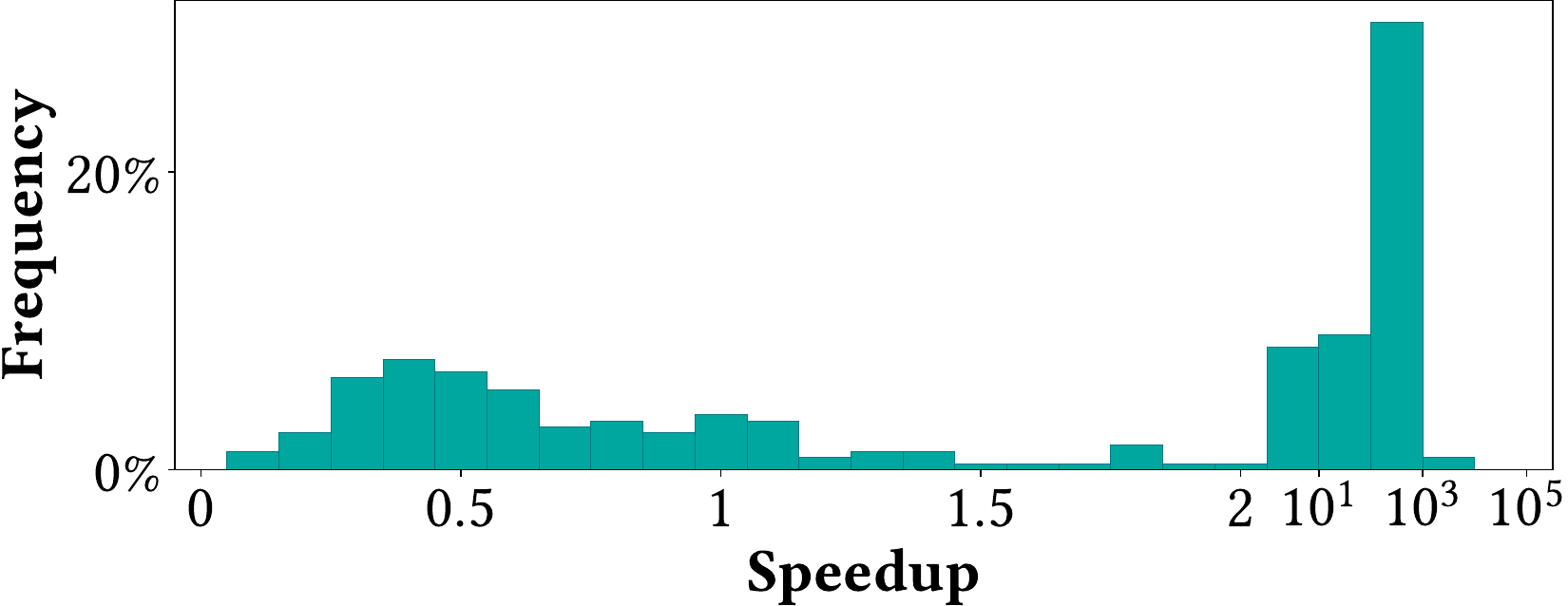}
        \vspace{-1.25\baselineskip}
        \caption{Musicbrainz}
        \label{fig:overall-speedup-duckdb-musicbrainz}
    \end{subfigure}
    \hfill
    \begin{subfigure}[b]{0.49\textwidth}
        \centering
        \includegraphics[width=\linewidth]{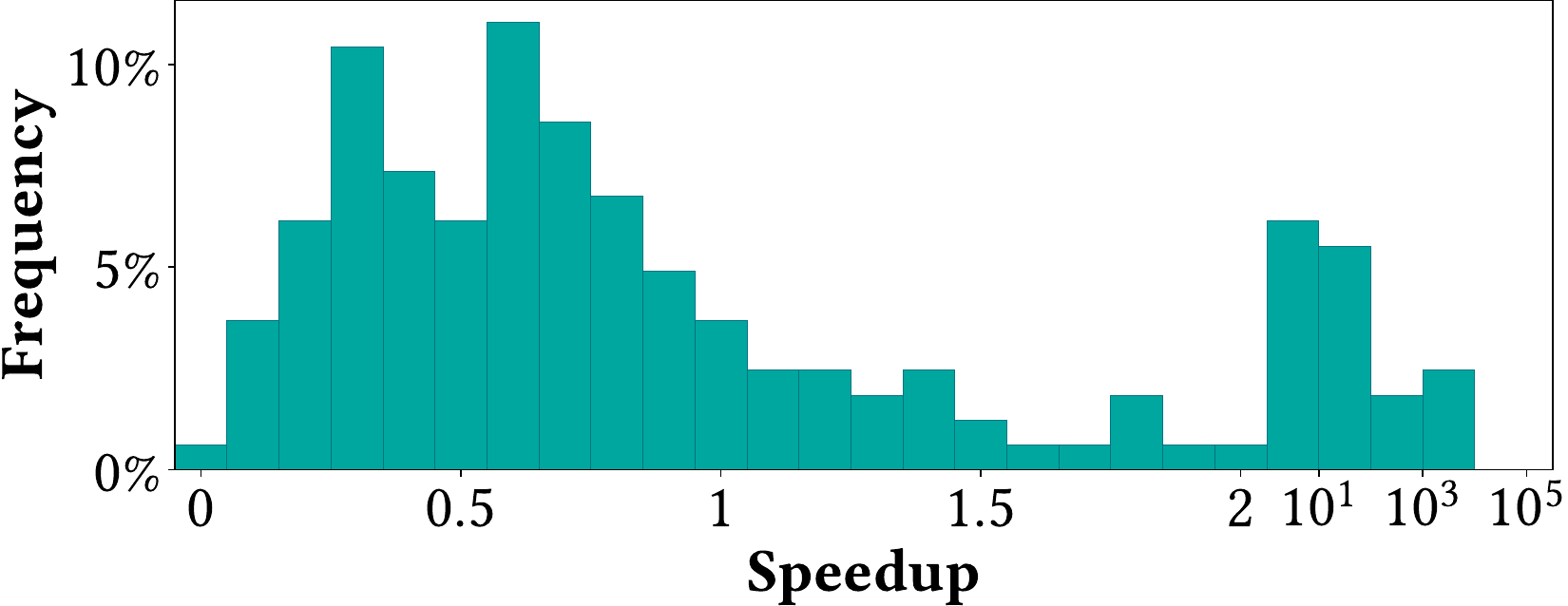}
        \vspace{-1.25\baselineskip}
        \caption{JOBLarge}
        \label{fig:overall-speedup-duckdb-job-large}
    \end{subfigure}
    \vspace{-0.5\baselineskip}
    \caption{Histograms of speedups of overall evaluation time of \sys{} over DuckDB on each benchmark}
    \label{fig:overall-speedup-duckdb-individual}
\end{figure}

\paragraph{Scatter plots}
\label{app:total-time-duckdb-individual}
\begin{figure}[b]
    \begin{subfigure}[b]{0.24\textwidth}
        \centering
        \includegraphics[width=\linewidth]{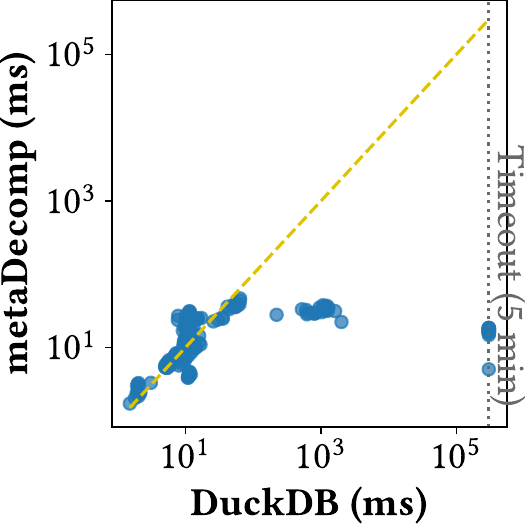}
        \vspace{-1.25\baselineskip}
        \caption{DSB}
        \label{fig:overall-scatter-duckdb-dsb}
    \end{subfigure}
    \hfill
    \begin{subfigure}[b]{0.24\textwidth}
        \centering
        \includegraphics[width=\linewidth]{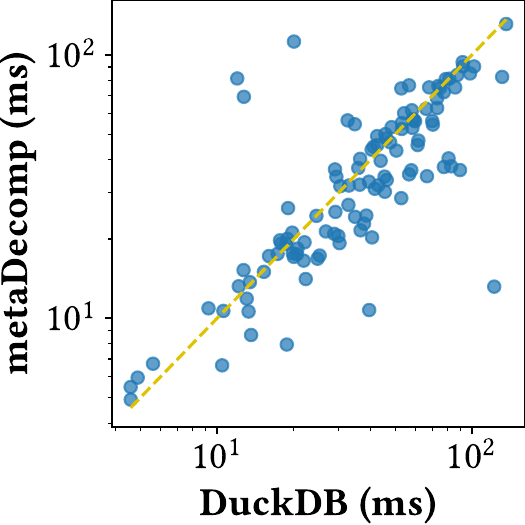}
        \vspace{-1.25\baselineskip}
        \caption{JOB}
        \label{fig:overall-scatter-duckdb-job-original}
    \end{subfigure}
    \hfill
    \begin{subfigure}[b]{0.24\textwidth}
        \centering
        \includegraphics[width=\linewidth]{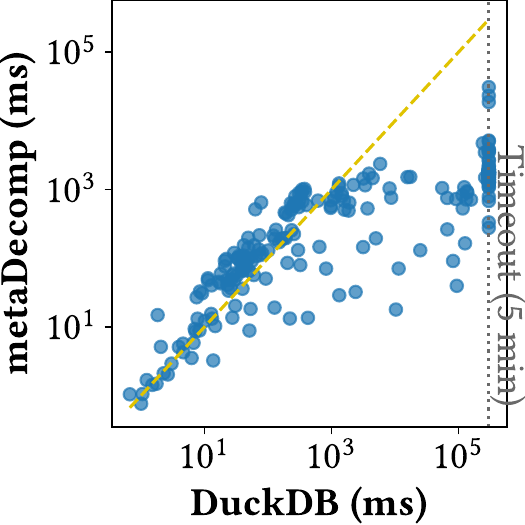}
        \vspace{-1.25\baselineskip}
        \caption{Musicbrainz}
        \label{fig:overall-scatter-duckdb-musicbrainz}
    \end{subfigure}
    \hfill
    \begin{subfigure}[b]{0.24\textwidth}
        \centering
        \includegraphics[width=\linewidth]{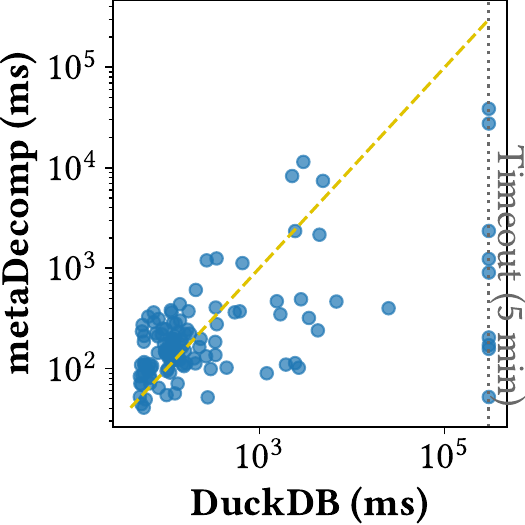}
        \vspace{-1.25\baselineskip}
        \caption{JOBLarge}
        \label{fig:overall-scatter-duckdb-job-large}
    \end{subfigure}
    \vspace{-0.5\baselineskip}
    \caption{Scatter plot of overall query evaluation time of using \sys{} versus DuckDB on each benchmark}
    \label{fig:overall-scatter-duckdb-individual}
\end{figure}
\cref{fig:overall-scatter-duckdb-individual} shows the detailed breakdown of overall query evaluation time using \sys{} and using DuckDB measured for queries in each of the four benchmarks.
With the exception of JOBLarge, all approaches have similar performance for queries that takes DuckDB less than 100 ms. For those taking DuckDB more than 100 ms, \sys{} starts to show significant speedups, mostly due to the greedy strategy of DuckDB used for large queries, which does not have optimality guarantees but may lead to highly suboptimal plans.

\subsubsection{\sys{} versus UnionDP}
\label{app:total-time-uniondp-individual}
\paragraph{Histograms of speedups}
\cref{fig:overall-speedup-uniondp-individual} shows histograms of the speedups of the overall query execution time using \sys{} over using UnionDP for queries in each of the four benchmarks. In the relatively smaller benchmarks DSB and JOB, the performance is similar. In Musicbrainz and JOBLarge, however, we observe a large number of significant speedups by multiple orders of magnitude, mostly due to the fact that UnionDP starts to use the heuristic partitioning which may lead to highly suboptimal plans.

\begin{figure}[H]
    \begin{subfigure}[b]{0.49\textwidth}
        \centering
        \includegraphics[width=\linewidth]{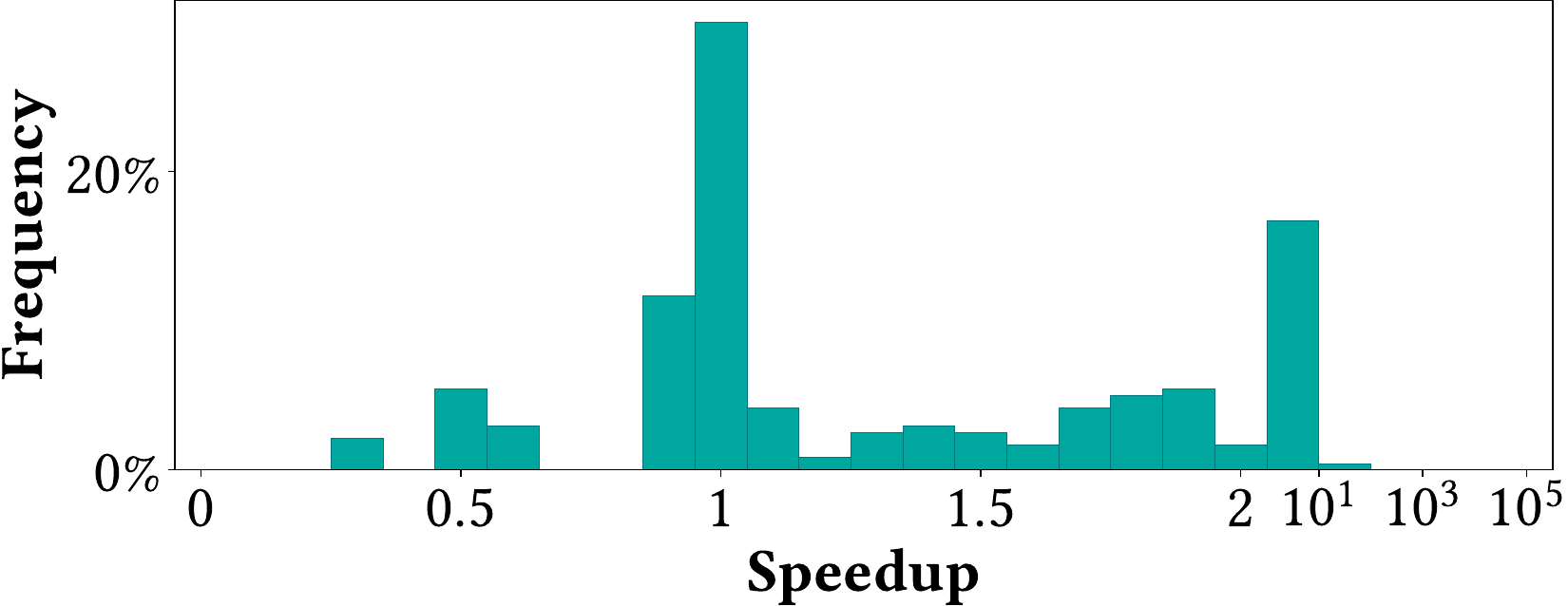}
        \caption{DSB}
        \label{fig:overall-speedup-uniondp-dsb}
    \end{subfigure}
    \hfill
    \begin{subfigure}[b]{0.49\textwidth}
        \centering
        \includegraphics[width=\linewidth]{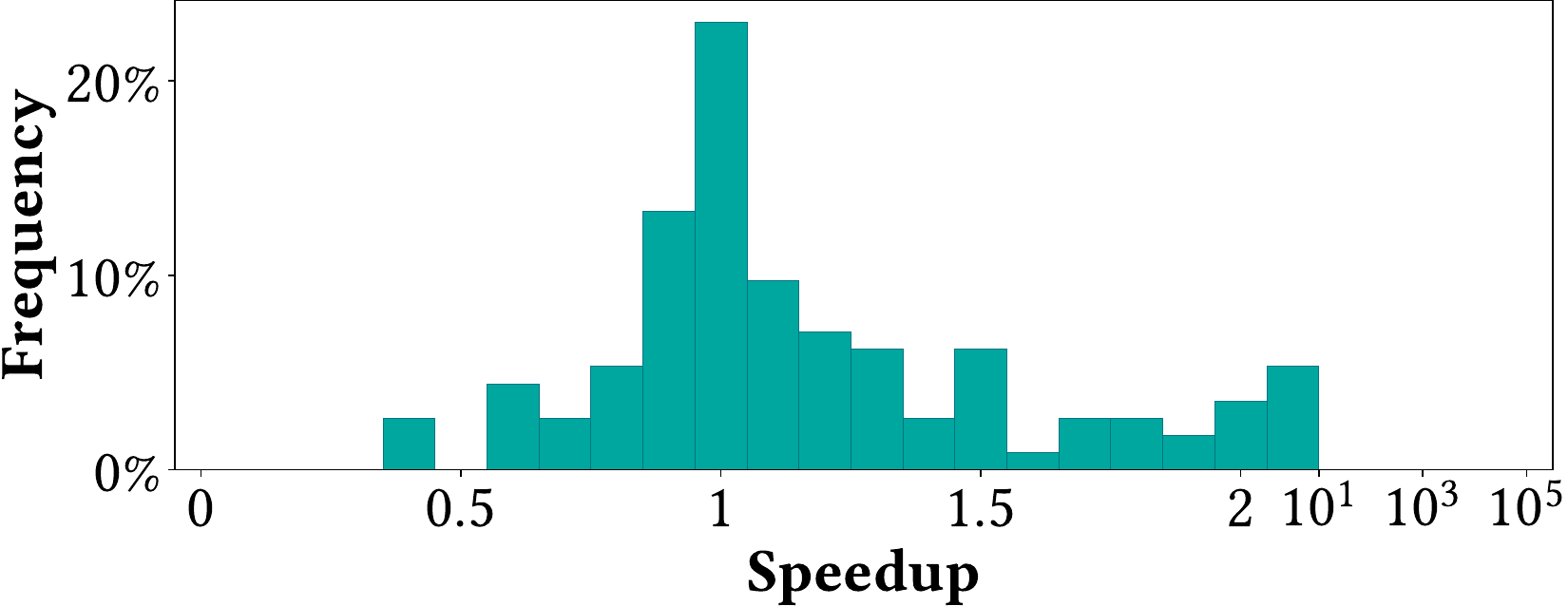}
        \caption{JOB}
        \label{fig:overall-speedup-uniondp-job-original}
    \end{subfigure}

    \vspace{0.5\baselineskip}

    \begin{subfigure}[b]{0.49\textwidth}
        \centering
        \includegraphics[width=\linewidth]{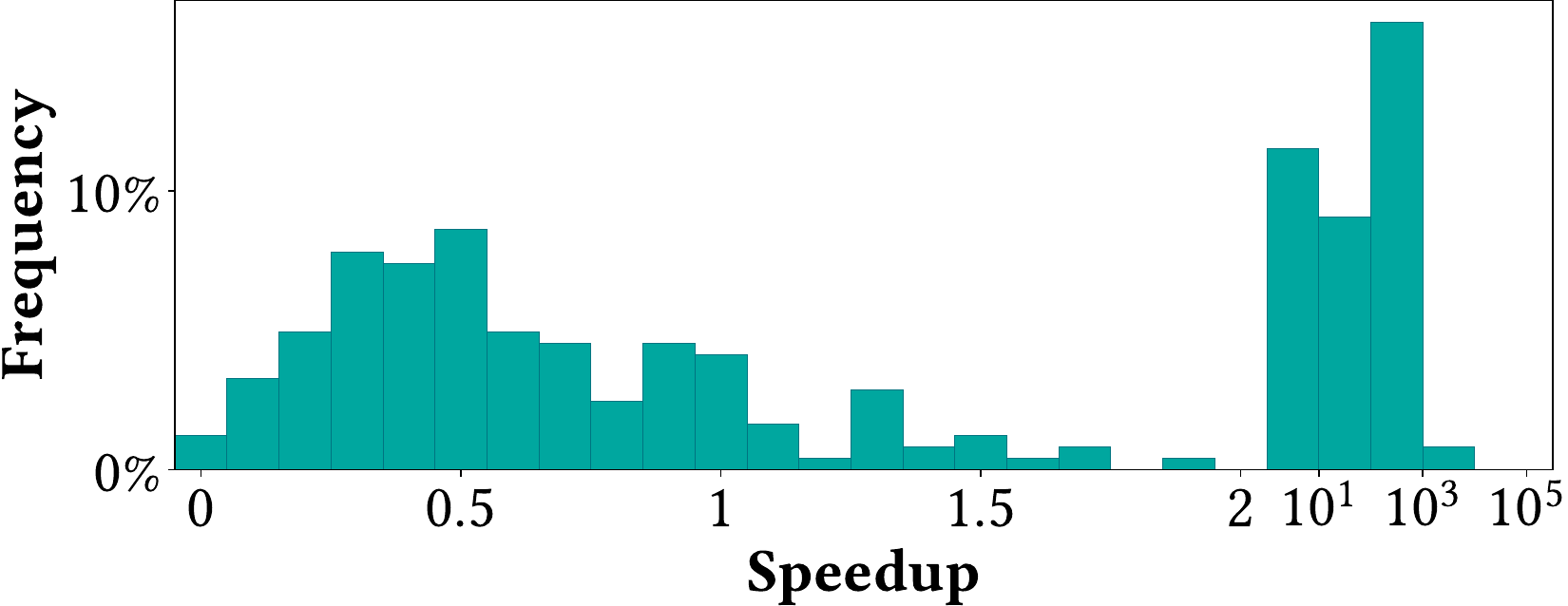}
        \caption{Musicbrainz}
        \label{fig:overall-speedup-uniondp-musicbrainz}
    \end{subfigure}
    \hfill
    \begin{subfigure}[b]{0.49\textwidth}
        \centering
        \includegraphics[width=\linewidth]{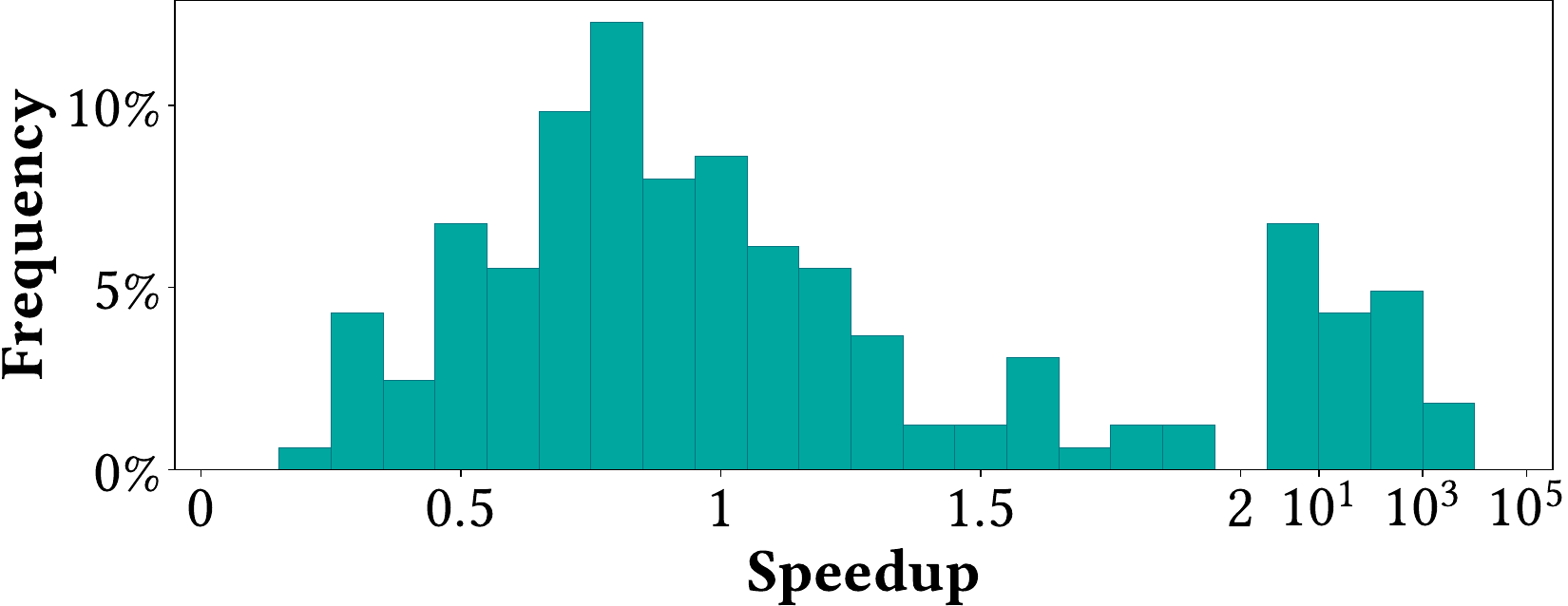}
        \caption{JOBLarge}
        \label{fig:overall-speedup-uniondp-job-large}
    \end{subfigure}
    \vspace{-0.5\baselineskip}
    \caption{Histograms of speedups of overall evaluation time of \sys{} over UnionDP on each benchmark}
    \label{fig:overall-speedup-uniondp-individual}
\end{figure}

\paragraph{Scatter plots} \cref{fig:overall-scatter-uniondp-individual} shows scatter plots of the overall query evaluation time using \sys{} and using UnionDP measured for queries in each of the four benchmarks. The two approaches have similar performance for queries that takes UnionDP less than 1 second ($10^3\ {\rm ms}$). For those taking UnionDP more than 1 second, \sys{} starts to show significant speedups. This is especially the case in the Musicbrainz and JOBLarge benchmarks.

\begin{figure}[H]
    \begin{subfigure}[b]{0.24\textwidth}
        \centering
        \includegraphics[width=\linewidth]{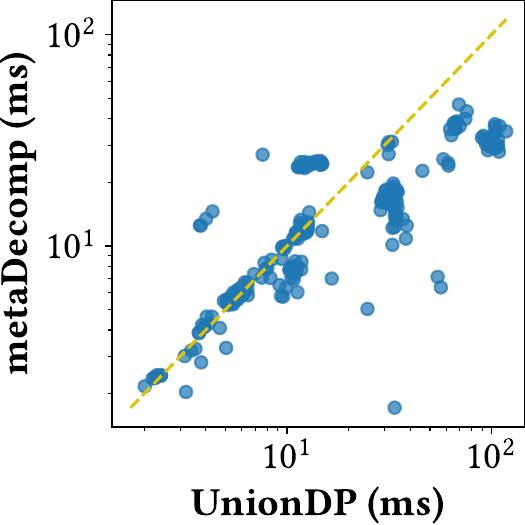}
        \caption{DSB}
        \label{fig:overall-scatter-uniondp-dsb}
    \end{subfigure}
    \hfill
    \begin{subfigure}[b]{0.24\textwidth}
        \centering
        \includegraphics[width=\linewidth]{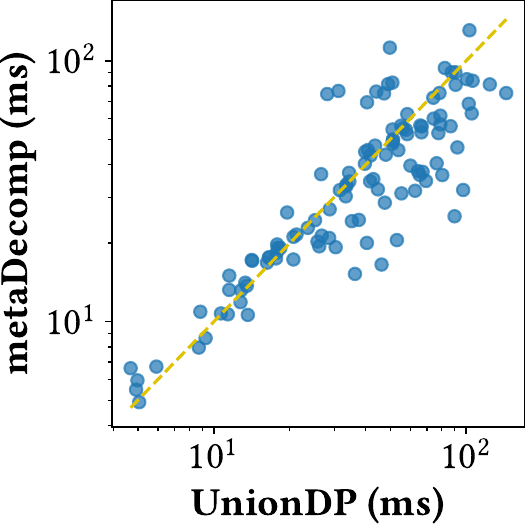}
        \caption{JOB}
        \label{fig:overall-scatter-uniondp-job-original}
    \end{subfigure}
    \hfill
    \begin{subfigure}[b]{0.24\textwidth}
        \centering
        \includegraphics[width=\linewidth]{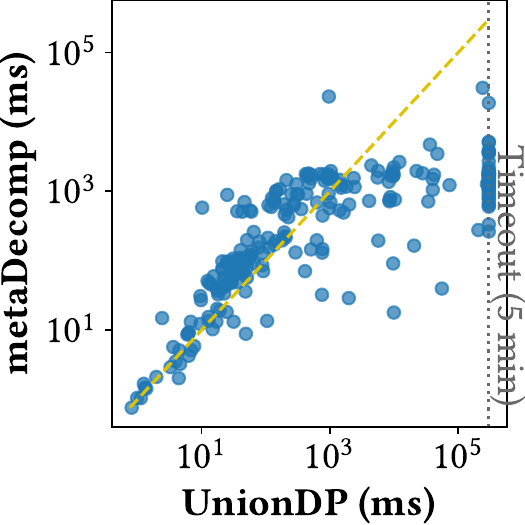}
        \caption{Musicbrainz}
        \label{fig:overall-scatter-uniondp-musicbrainz}
    \end{subfigure}
    \hfill
    \begin{subfigure}[b]{0.24\textwidth}
        \centering
        \includegraphics[width=\linewidth]{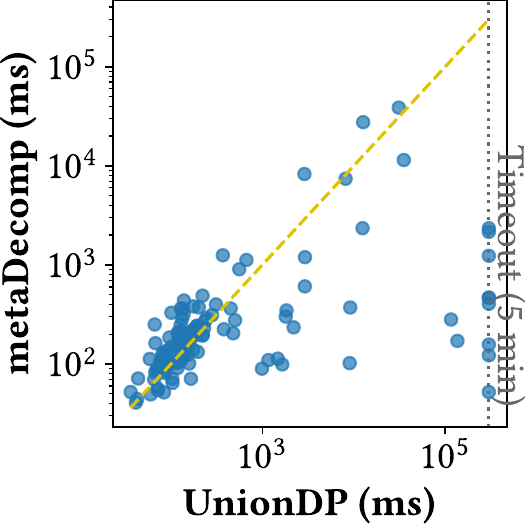}
        \caption{JOBLarge}
        \label{fig:overall-scatter-uniondp-job-large}
    \end{subfigure}
    \vspace{-0.5\baselineskip}
    \caption{Scatter plot of overall evaluation time of using \sys{} versus UnionDP on each benchmark}
    \label{fig:overall-scatter-uniondp-individual}
\end{figure}

\newpage

\subsubsection{\sys{} versus Yannakakis$^+$}
\label{app:total-time-yanplus-individual}
\paragraph{Histograms of speedups}
\cref{fig:overall-speedup-yanplus-individual} shows histograms of the speedups of the overall execution time using \sys{} over using Yannakakis$^+$ for queries in each of the four benchmarks. The major advantage of \sys{} over Yannakakis$^+$ is the elimination of the overhead of full reduction, and this is very clearly reflected in the speedups in the smaller benchmarks DSB and JOB. On the larger benchmarks Musicbrainz and JOBLarge, \sys{} also shows comparable performance, with significant speedups on a few queries in which the full reduction takes a very long time.

\begin{figure}[H]
    \begin{subfigure}[b]{0.49\textwidth}
        \centering
        \includegraphics[width=\linewidth]{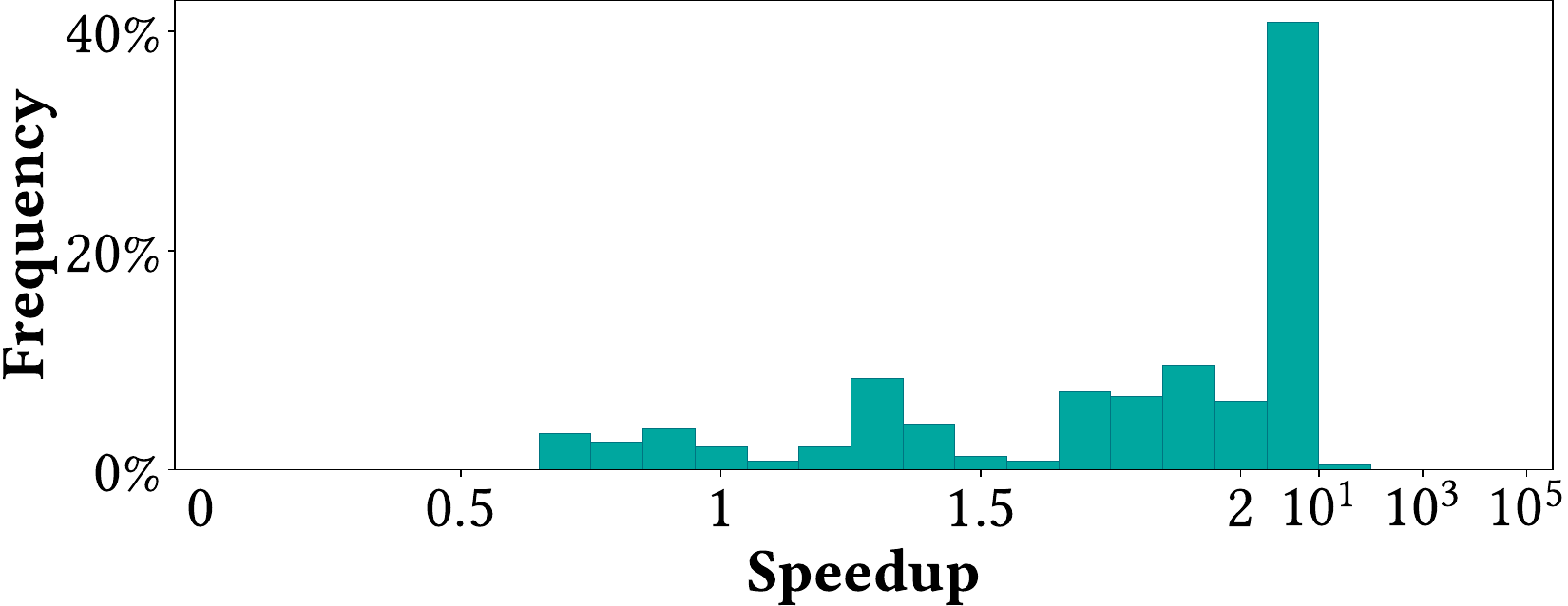}
        \caption{DSB}
        \label{fig:overall-speedup-yanplus-dsb}
    \end{subfigure}
    \hfill
    \begin{subfigure}[b]{0.49\textwidth}
        \centering
        \includegraphics[width=\linewidth]{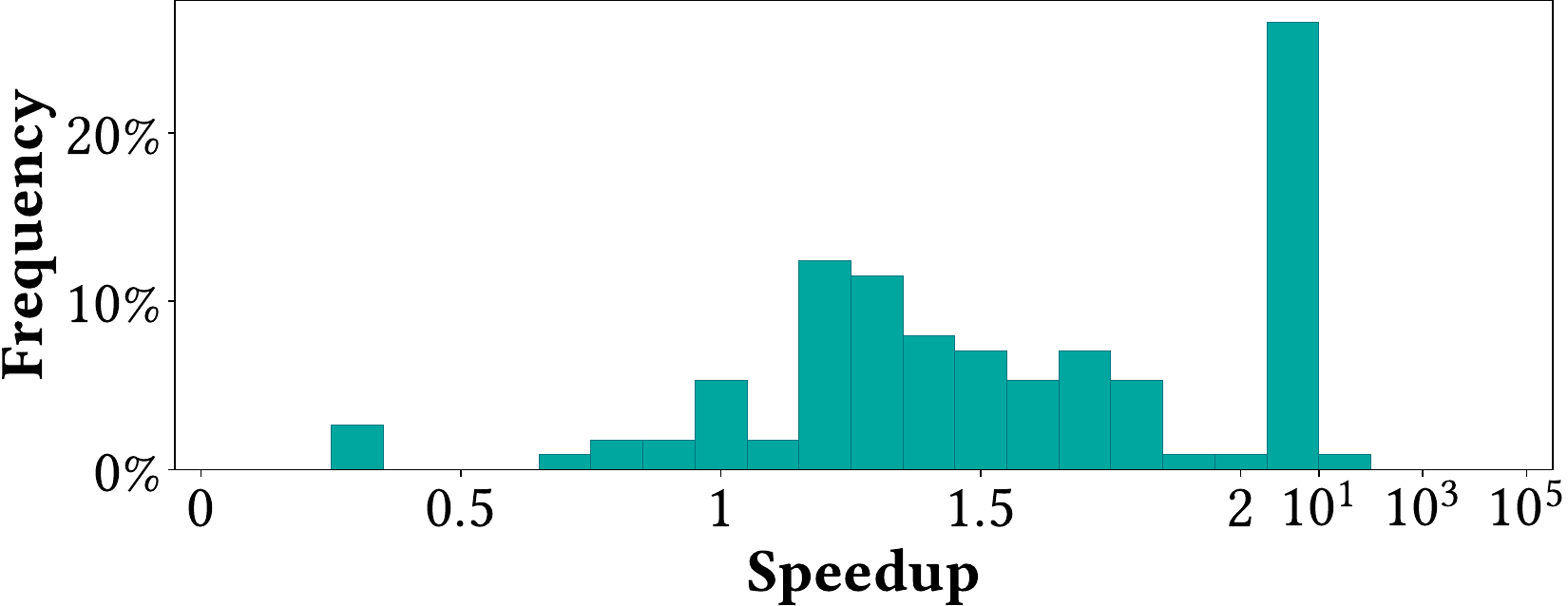}
        \caption{JOB}
        \label{fig:overall-speedup-yanplus-job-original}
    \end{subfigure}

    \vspace{0.5\baselineskip}

    \begin{subfigure}[b]{0.49\textwidth}
        \centering
        \includegraphics[width=\linewidth]{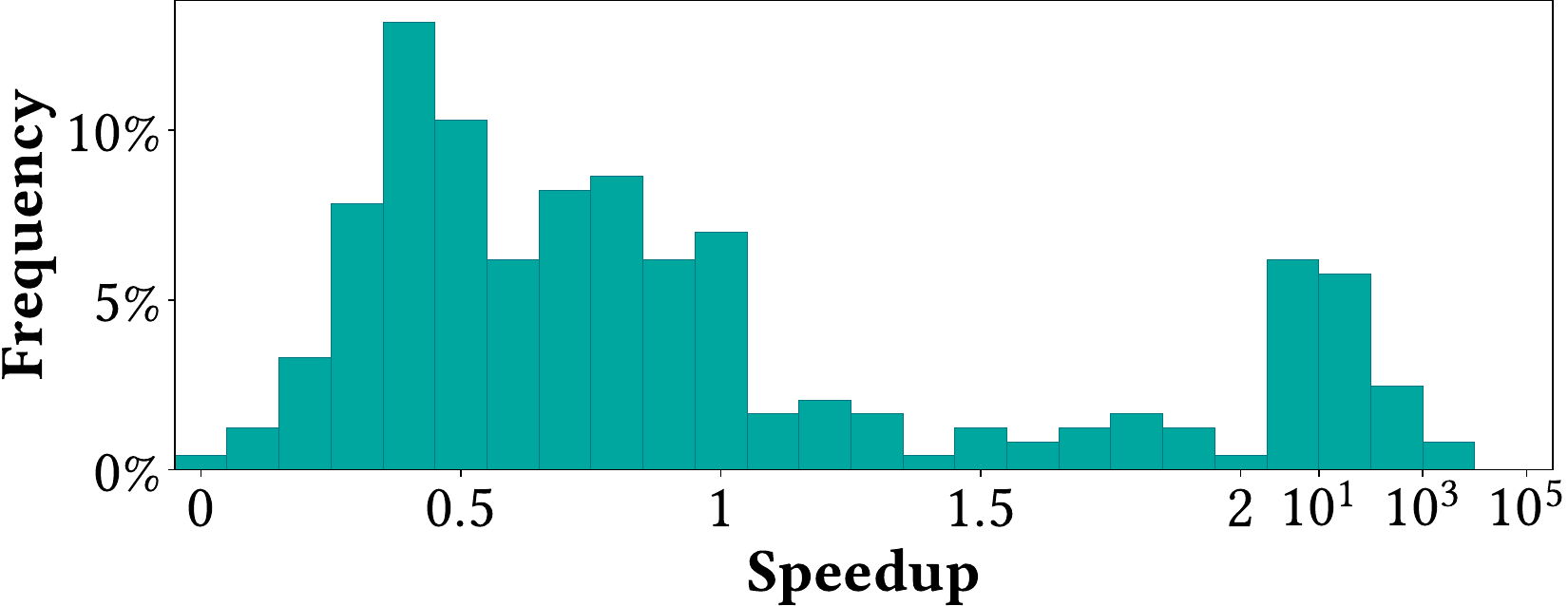}
        \caption{Musicbrainz}
        \label{fig:overall-speedup-yanplus-musicbrainz}
    \end{subfigure}
    \hfill
    \begin{subfigure}[b]{0.49\textwidth}
        \centering
        \includegraphics[width=\linewidth]{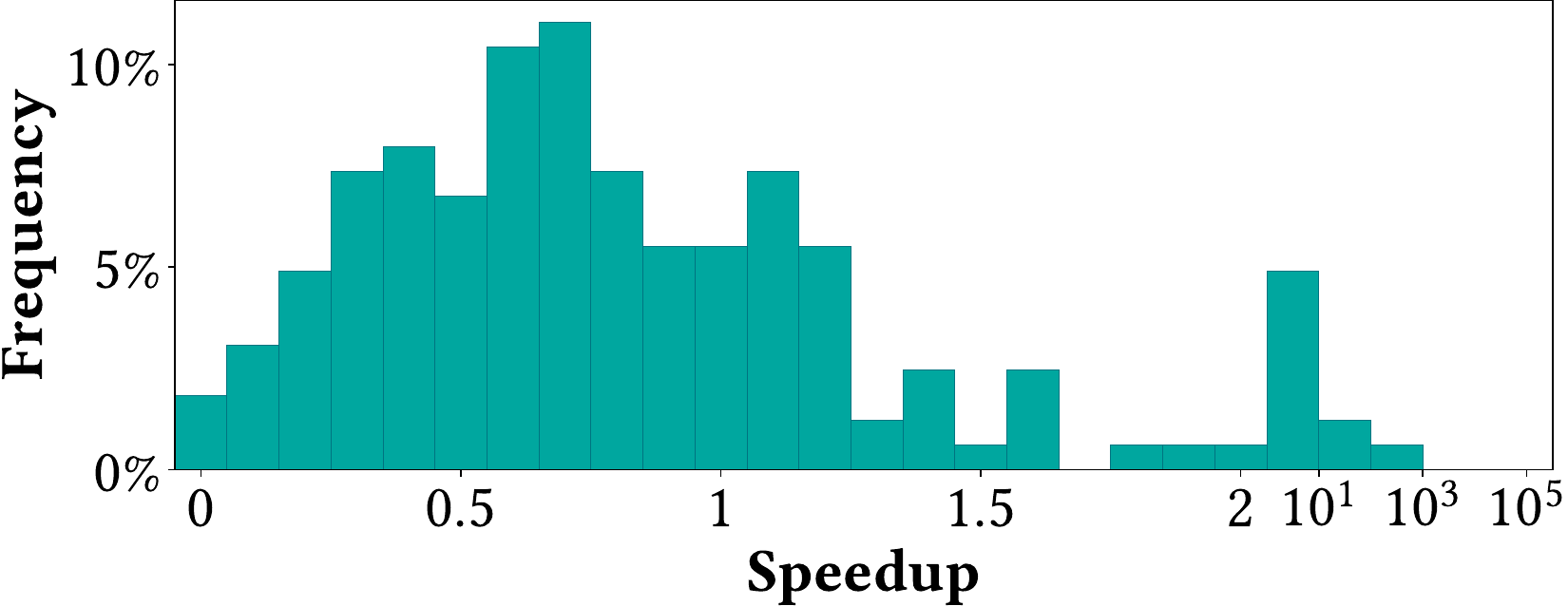}
        \caption{JOBLarge}
        \label{fig:overall-speedup-yanplus-job-large}
    \end{subfigure}
    \vspace{-0.5\baselineskip}
    \caption{Histograms of speedups of overall evaluation time of \sys{} over Yannakakis$^+$ on each benchmark}
    \label{fig:overall-speedup-yanplus-individual}
\end{figure}

\paragraph{Scatter plots} \cref{fig:overall-scatter-yanplus-individual} shows scatter plots of the overall query evaluation time using \sys{} and using Yannakakis$^+$ measured for queries in each of the four benchmarks. Again, the elimination of the overhead of full reduction allows for speedups in many queries in all four benchmarks, and there are significant speedups on a few queries in Musicbrainz and JOBLarge for which the full reduction takes a very long time.

\begin{figure}[H]
    \begin{subfigure}[b]{0.24\textwidth}
        \centering
        \includegraphics[width=\linewidth]{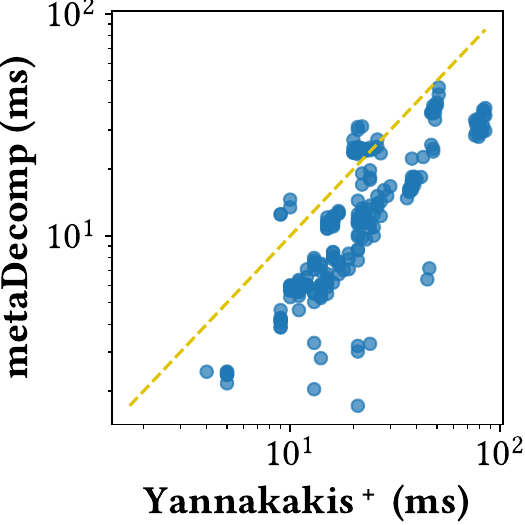}
        \caption{DSB}
        \label{fig:overall-scatter-yanplus-dsb}
    \end{subfigure}
    \hfill
    \begin{subfigure}[b]{0.24\textwidth}
        \centering
        \includegraphics[width=\linewidth]{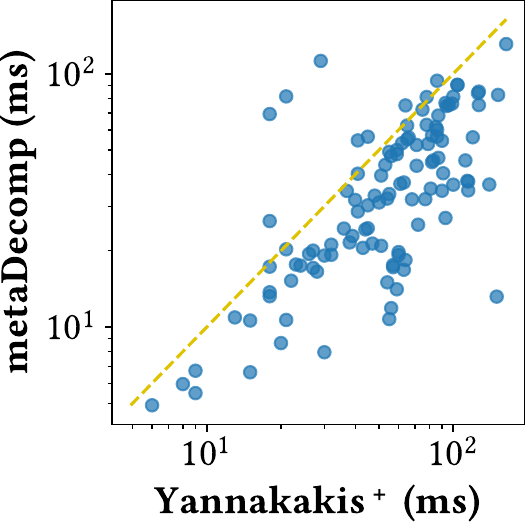}
        \caption{JOB}
        \label{fig:overall-scatter-yanplus-job-original}
    \end{subfigure}
    \hfill
    \begin{subfigure}[b]{0.24\textwidth}
        \centering
        \includegraphics[width=\linewidth]{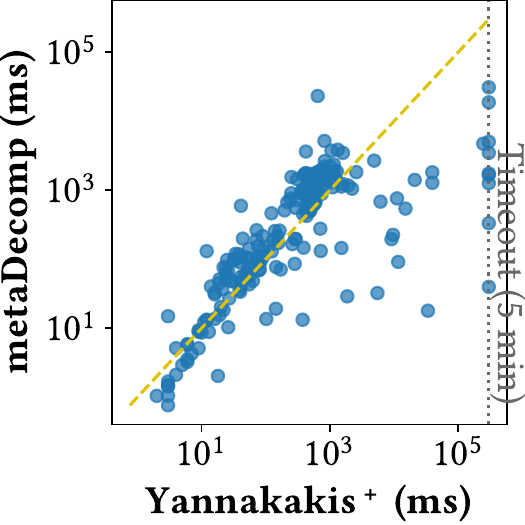}
        \caption{Musicbrainz}
        \label{fig:overall-scatter-yanplus-musicbrainz}
    \end{subfigure}
    \hfill
    \begin{subfigure}[b]{0.24\textwidth}
        \centering
        \includegraphics[width=\linewidth]{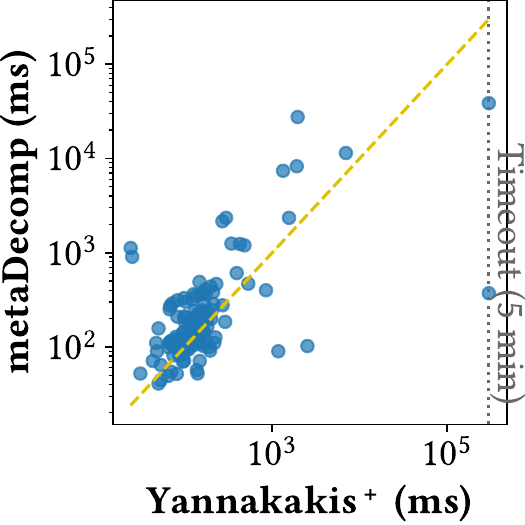}
        \caption{JOBLarge}
        \label{fig:overall-scatter-yanplus-job-large}
    \end{subfigure}
    \vspace{-0.5\baselineskip}
    \caption{Scatter plots of overall evaluation time of using \sys{} versus Yannakakis$^+$ on each benchmark}
    \label{fig:overall-scatter-yanplus-individual}
\end{figure}

\newpage

\subsubsection{\sys{} versus LearnedRewrite}
\label{app:total-time-learned-rewrite-individual}
\paragraph{Histograms of speedups}
\cref{fig:overall-speedup-learned-rewrite-individual} shows histograms of the speedups of the overall execution time using \sys{} over LearnedRewrite for queries in each of the four benchmarks. We can observe a clear overall advantage of \sys{}, due to (1) the large optimization overhead and (2) the limitation of the cost-based optimization and rewrite rules in LearnedRewrite.

\begin{figure}[H]
    \begin{subfigure}[b]{0.49\textwidth}
        \centering
        \includegraphics[width=\linewidth]{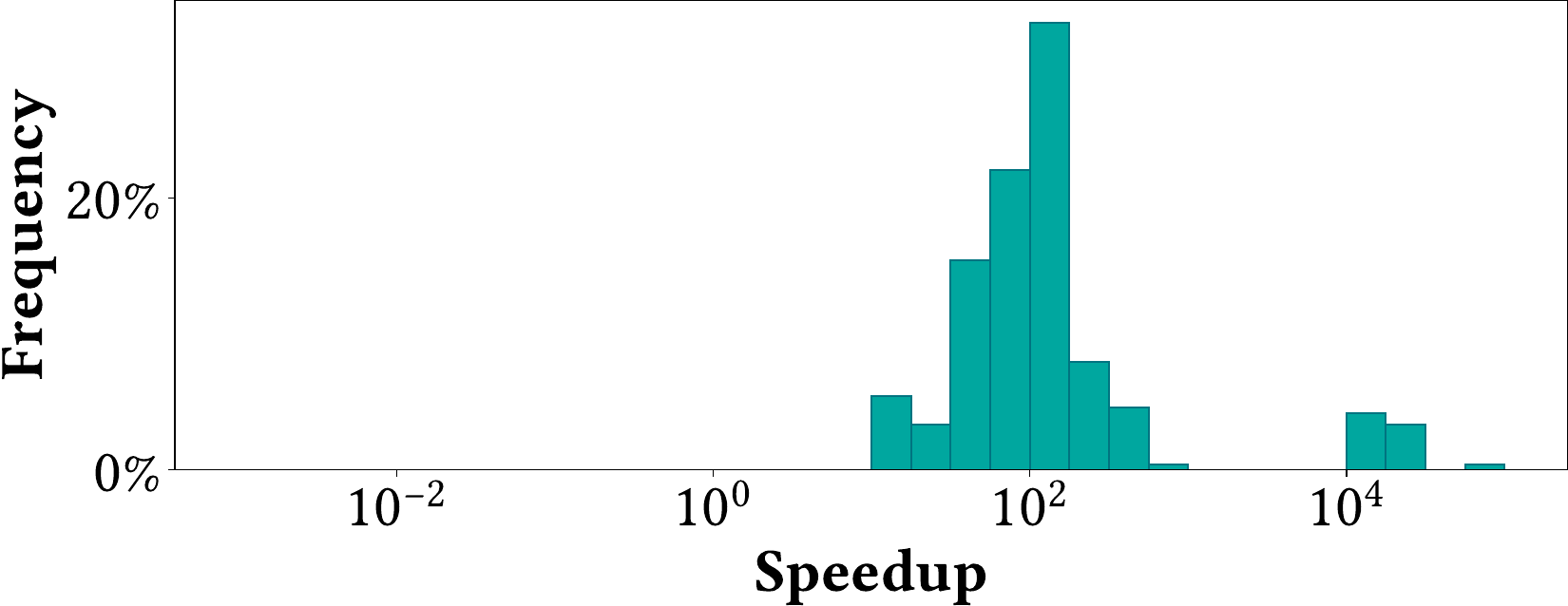}
        \caption{DSB}
        \label{fig:overall-speedup-learned-rewrite-dsb}
    \end{subfigure}
    \hfill
    \begin{subfigure}[b]{0.49\textwidth}
        \centering
        \includegraphics[width=\linewidth]{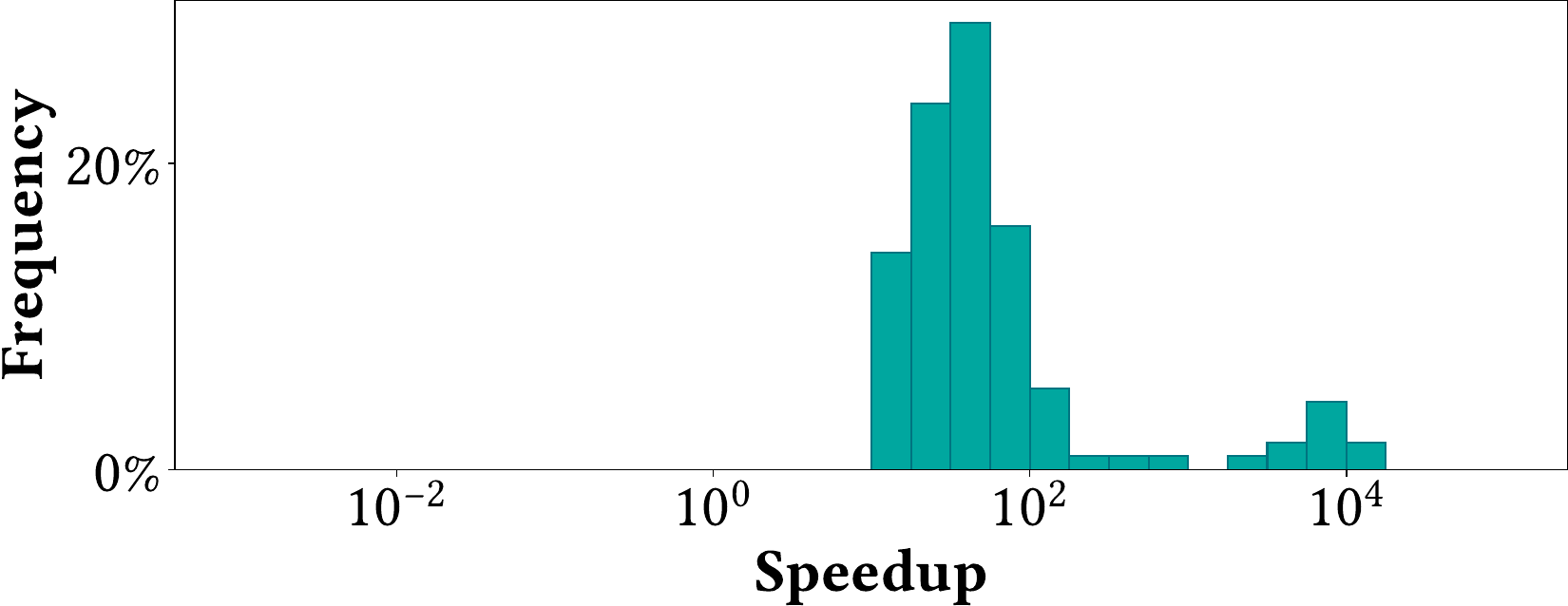}
        \caption{JOB}
        \label{fig:overall-speedup-learned-rewrite-job-original}
    \end{subfigure}

    \vspace{0.5\baselineskip}

    \begin{subfigure}[b]{0.49\textwidth}
        \centering
        \includegraphics[width=\linewidth]{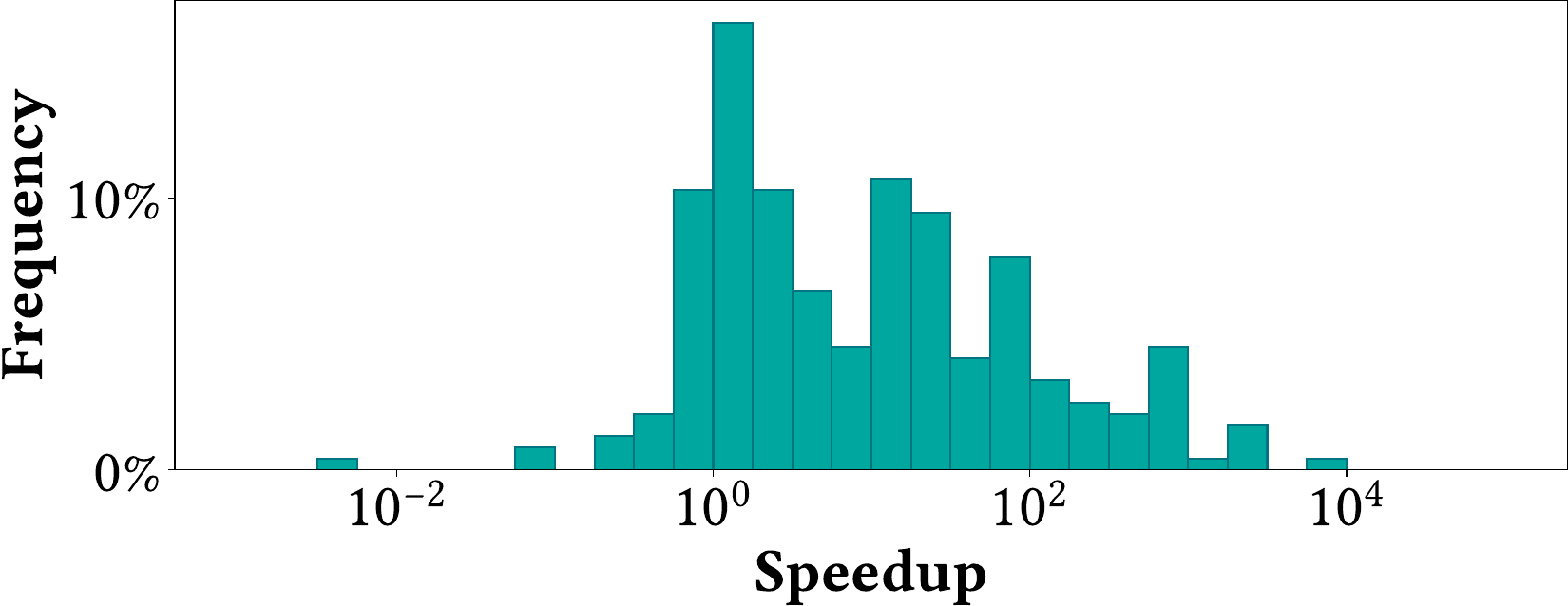}
        \caption{Musicbrainz}
        \label{fig:overall-speedup-learned-rewrite-musicbrainz}
    \end{subfigure}
    \hfill
    \begin{subfigure}[b]{0.49\textwidth}
        \centering
        \includegraphics[width=\linewidth]{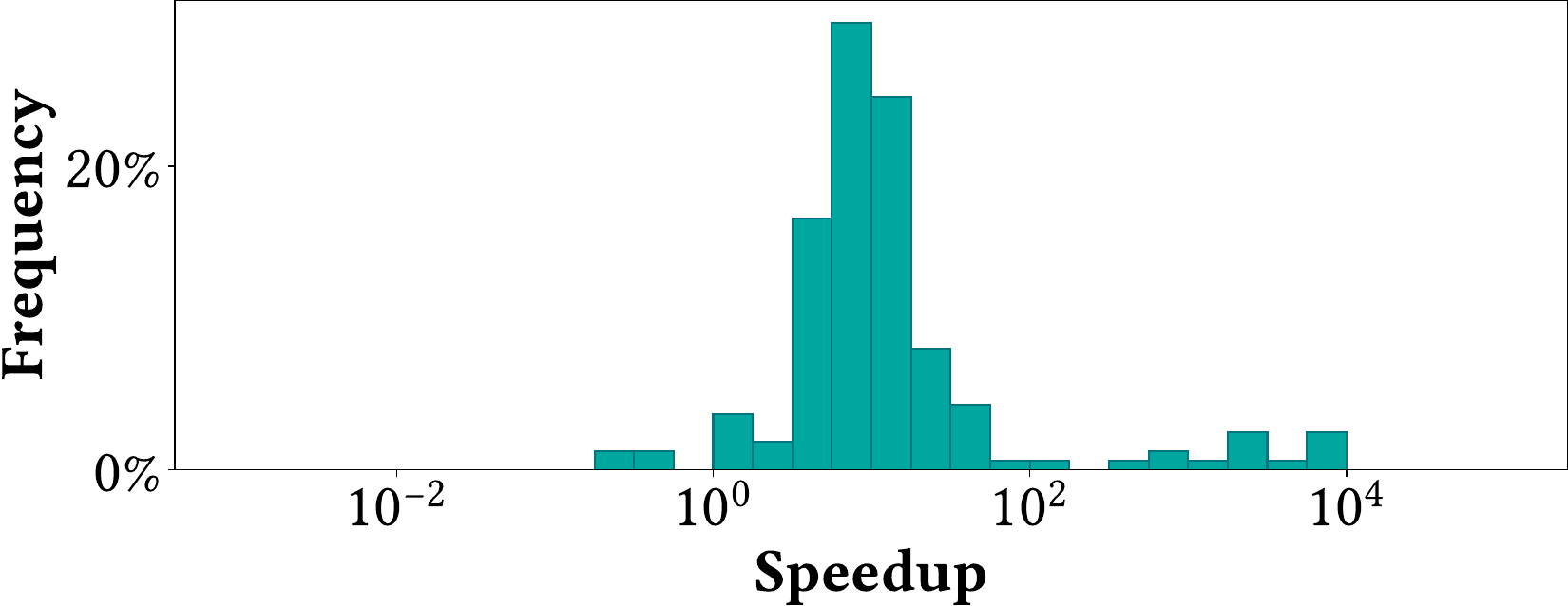}
        \caption{JOBLarge}
        \label{fig:overall-speedup-learned-rewrite-job-large}
    \end{subfigure}
    \vspace{-0.5\baselineskip}
    \caption{Histograms of speedups of overall evaluation time of \sys{} over LearnedRewrite on each benchmark}
    \label{fig:overall-speedup-learned-rewrite-individual}
\end{figure}

\paragraph{Scatter plots} \cref{fig:overall-scatter-learned-rewrite-individual} shows scatter plots of the overall query evaluation time using \sys{} and using Learned Rewrite, measured for queries in each of the four benchmarks. Again, we observe a clear overall advantage of \sys{} over Learned Rewrite in all four benchmarks. In particular, it takes Learned Rewrite an overhead of at least 1 second to optimize each query. For many small queries, this is already more than the overall evaluation time for \sys{} as well as other traditional query optimization approaches.

\begin{figure}[H]
    \begin{subfigure}[b]{0.24\textwidth}
        \centering
        \includegraphics[width=\linewidth]{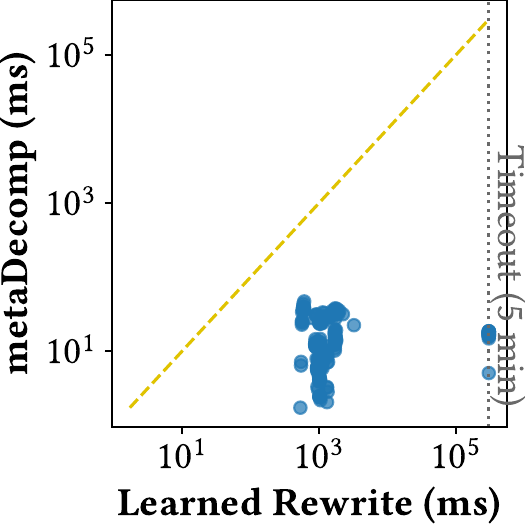}
        \caption{DSB}
        \label{fig:overall-scatter-learned-rewrite-dsb}
    \end{subfigure}
    \hfill
    \begin{subfigure}[b]{0.24\textwidth}
        \centering
        \includegraphics[width=\linewidth]{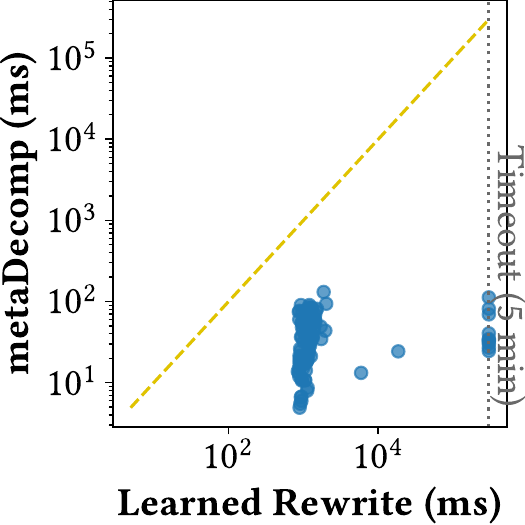}
        \caption{JOB}
        \label{fig:overall-scatter-learned-rewrite-job-original}
    \end{subfigure}
    \hfill
    \begin{subfigure}[b]{0.24\textwidth}
        \centering
        \includegraphics[width=\linewidth]{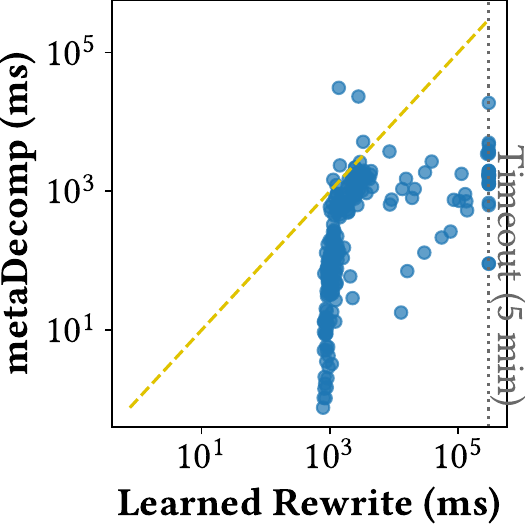}
        \caption{Musicbrainz}
        \label{fig:overall-scatter-learned-rewrite-musicbrainz}
    \end{subfigure}
    \hfill
    \begin{subfigure}[b]{0.24\textwidth}
        \centering
        \includegraphics[width=\linewidth]{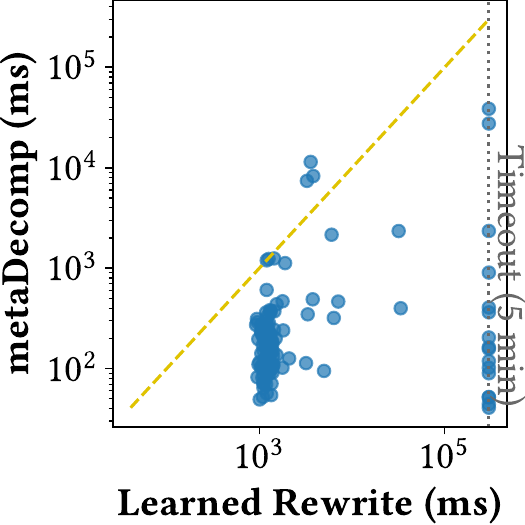}
        \caption{JOBLarge}
        \label{fig:overall-scatter-learned-rewrite-job-large}
    \end{subfigure}
    \vspace{-0.5\baselineskip}
    \caption{Scatter plots of overall evaluation time using \sys{} versus LearnedRewrite on each benchmark}
    \label{fig:overall-scatter-learned-rewrite-individual}
\end{figure}

\newpage

\subsubsection{\sys{} versus LLM-R$^2$}
\label{app:total-time-llm-r2-individual}
\paragraph{Histograms of speedups}
\cref{fig:overall-speedup-llm-r2-individual} shows histograms of the speedups of the overall query execution time using \sys{} over using LLM-R$^2$ for queries in each of the four benchmarks. Similar to the case for Learned Rewrite, we can observe a clear overall advantage of \sys{} over LLM-R$^2$, due to (1) the large optimization overhead and (2) the limitation of the cost-based optimization and rewrite rules of LLM-R$^2$.

\begin{figure}[H]
    \begin{subfigure}[b]{0.49\textwidth}
        \centering
        \includegraphics[width=\linewidth]{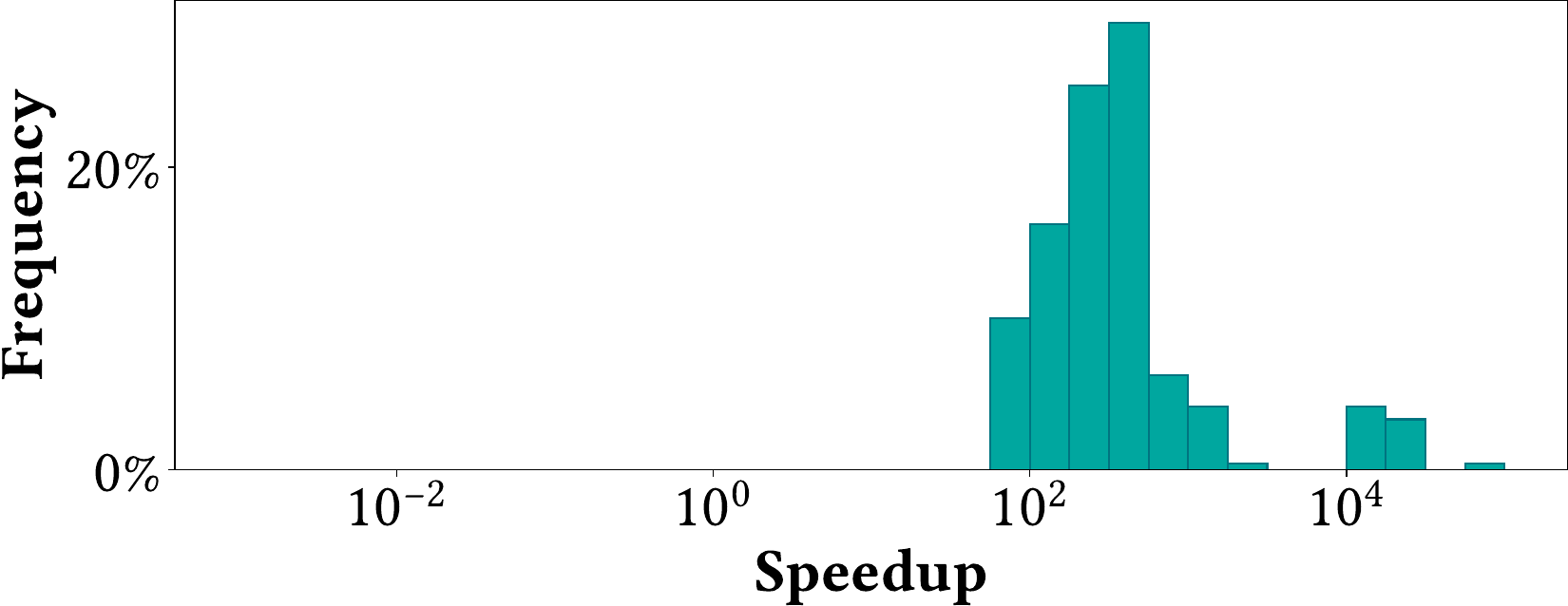}
        \caption{DSB}
        \label{fig:overall-speedup-llm-r2-dsb}
    \end{subfigure}
    \hfill
    \begin{subfigure}[b]{0.49\textwidth}
        \centering
        \includegraphics[width=\linewidth]{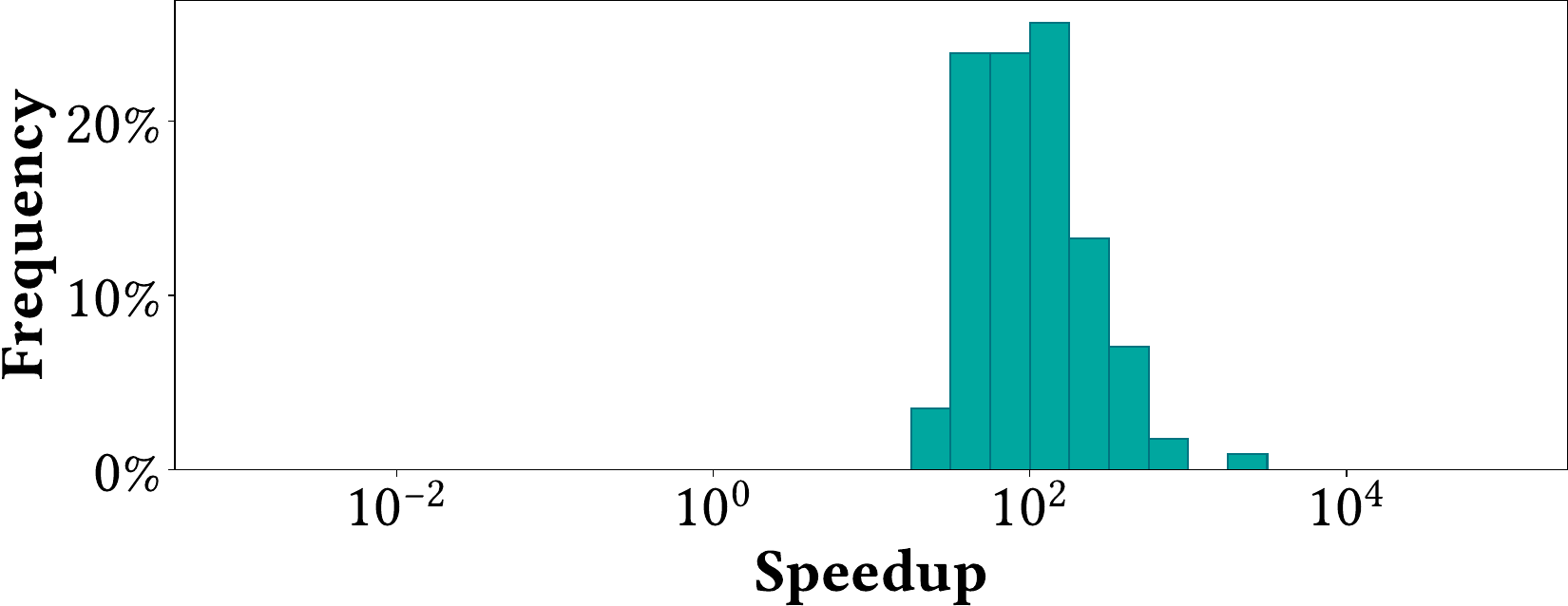}
        \caption{JOB}
        \label{fig:overall-speedup-llm-r2-job-original}
    \end{subfigure}

    \vspace{0.25\baselineskip}

    \begin{subfigure}[b]{0.49\textwidth}
        \centering
        \includegraphics[width=\linewidth]{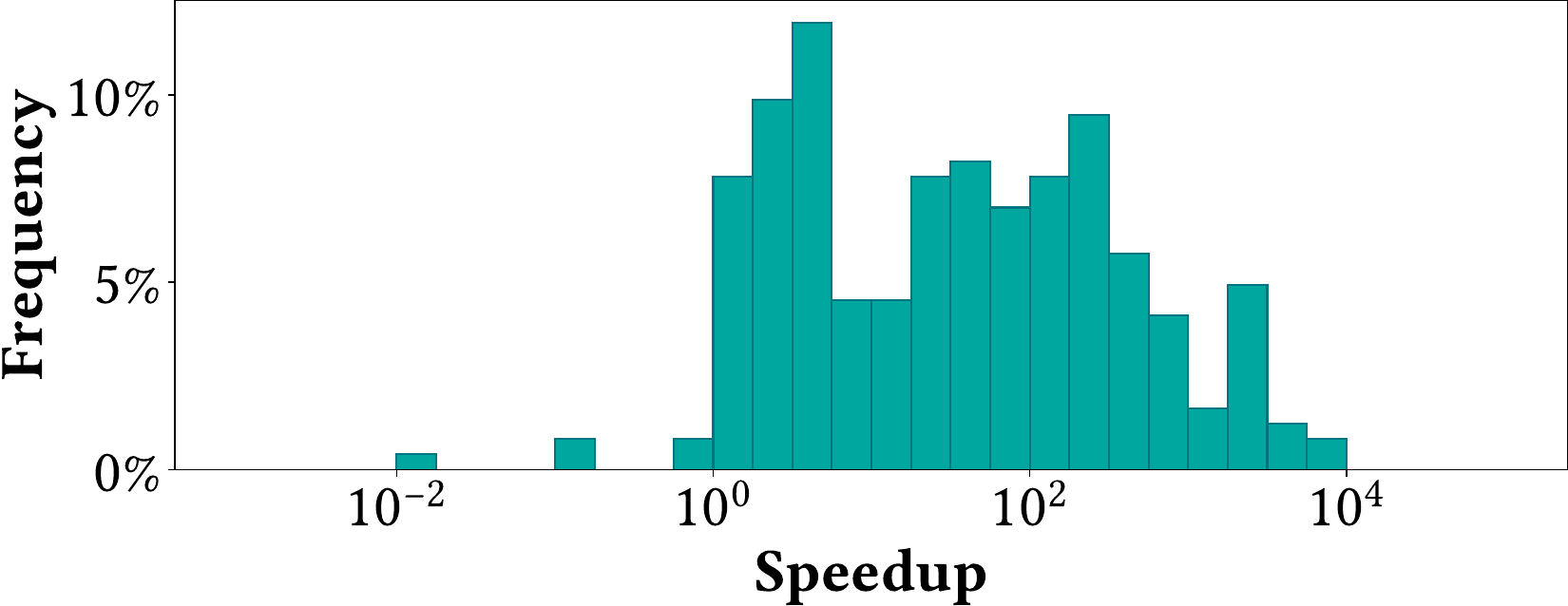}
        \caption{Musicbrainz}
        \label{fig:overall-speedup-llm-r2-musicbrainz}
    \end{subfigure}
    \hfill
    \begin{subfigure}[b]{0.49\textwidth}
        \centering
        \includegraphics[width=\linewidth]{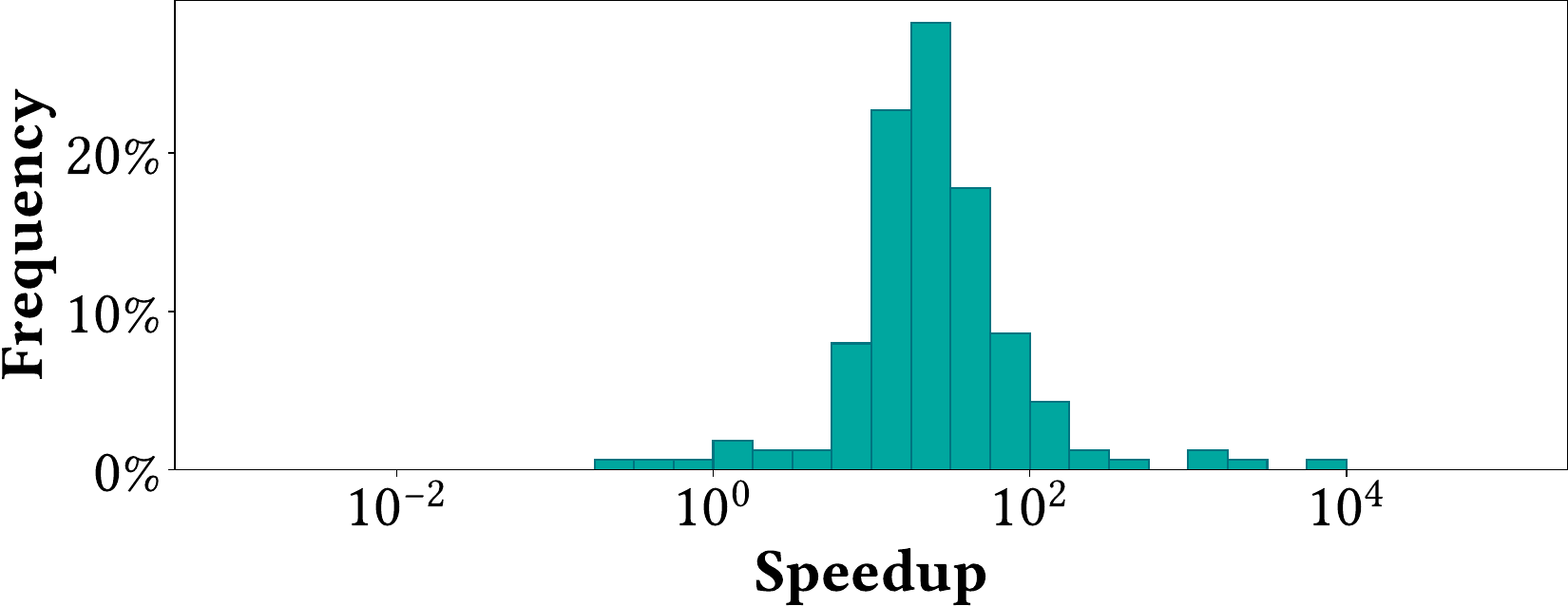}
        \caption{JOBLarge}
        \label{fig:overall-speedup-llm-r2-job-large}
    \end{subfigure}
    \vspace{-0.5\baselineskip}
    \caption{Histograms of speedups of overall evaluation time of \sys{} over LLM-R$^2$ on each benchmark}
    \label{fig:overall-speedup-llm-r2-individual}
\end{figure}

\paragraph{Scatter plots} \cref{fig:overall-scatter-llm-r2-individual} shows scatter plots of the overall query evaluation time using \sys{} and using LLM-R$^2$, measured for queries in each of the four benchmarks. Again, we observe a clear overall advantage of \sys{} over LLM-R$^2$ in all four benchmarks. Similar to the case with Learned Rewrite, it takes LLM-R$^2$ an even higher overhead of at least 3 seconds to optimize each query. For many small queries, this is already more than the overall evaluation time for \sys{} as well as other traditional query optimization approaches.

\begin{figure}[H]
    \begin{subfigure}[b]{0.24\textwidth}
        \centering
        \includegraphics[width=\linewidth]{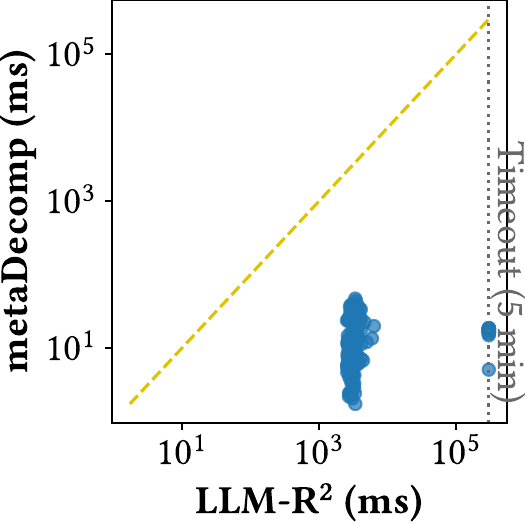}
        \caption{DSB}
        \label{fig:overall-scatter-llm-r2-dsb}
    \end{subfigure}
    \hfill
    \begin{subfigure}[b]{0.24\textwidth}
        \centering
        \includegraphics[width=\linewidth]{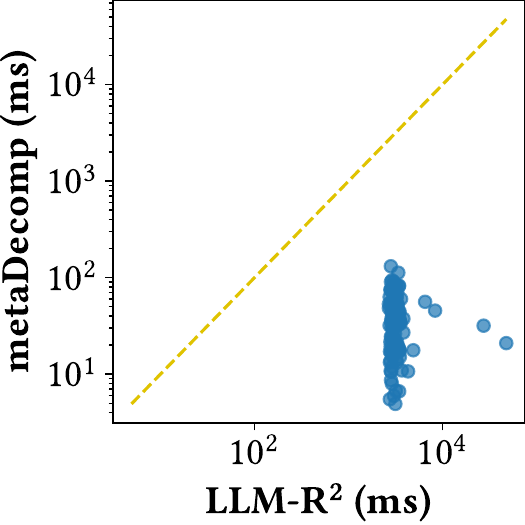}
        \caption{JOB}
        \label{fig:overall-scatter-llm-r2-job-original}
    \end{subfigure}
    \hfill
    \begin{subfigure}[b]{0.24\textwidth}
        \centering
        \includegraphics[width=\linewidth]{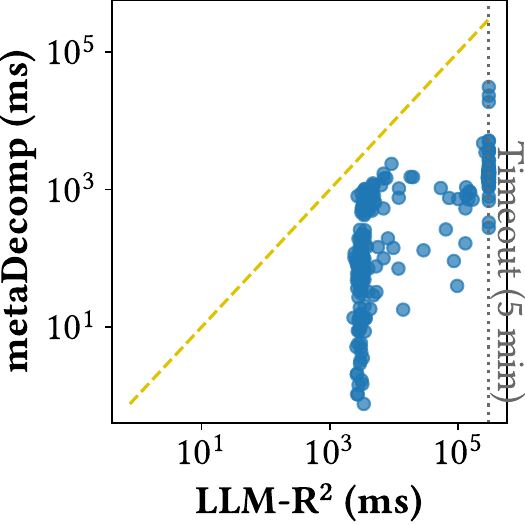}
        \caption{Musicbrainz}
        \label{fig:overall-scatter-llm-r2-musicbrainz}
    \end{subfigure}
    \hfill
    \begin{subfigure}[b]{0.24\textwidth}
        \centering
        \includegraphics[width=\linewidth]{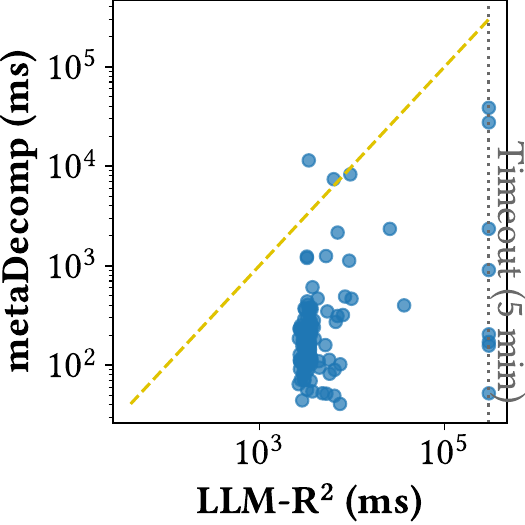}
        \caption{JOBLarge}
        \label{fig:overall-scatter-llm-r2-job-large}
    \end{subfigure}
    \vspace{-0.5\baselineskip}
    \caption{Scatter plot of overall evaluation time of using \sys{} versus LLM-R$^2$ on each benchmark}
    \label{fig:overall-scatter-llm-r2-individual}
\end{figure}